\documentclass[acmsmall,screen,nonacm]{acmart}

\usepackage{subcaption}
\usepackage{stmaryrd}
\usepackage{xcolor}
\usepackage{ifthen}
\newboolean{techreport}
\setboolean{techreport}{true}
\newcommand*{\appref}[1]{\ifthenelse{\boolean{techreport}}{\cref{#1}}{\cref*{#1}~\cite{tech_report}}}
\ifthenelse{\boolean{techreport}}{}{\input{appendixlabels}}

\usepackage{graphicx,wrapfig,lipsum}
\usepackage{makecell}
\usepackage{multirow}
\usepackage{makecell}
\usepackage[normalem]{ulem}
\usepackage{soul}
\usepackage{natbib}
\usepackage{graphicx}
\usepackage{changepage}
\usepackage{stmaryrd}
\usepackage{amsmath}
\usepackage{amsfonts}
\usepackage{relsize}
\usepackage{mdframed}
\usepackage{comment}
\usepackage{adjustbox}
\usepackage{tabularx}
\usepackage{xltabular}
\makeatletter
\usepackage{listings}
\newcommand*\mysize{%
  \@setfontsize\mysize{8.8}{9.0}%
}
\makeatother

\usepackage{tikz}
\usetikzlibrary{tikzmark,shapes.misc}
\tikzset{white border/.style={preaction={draw,white,line width=4pt}}}  

\newcommand{\wnsay}[1]{}

\newcommand{\whsay}[1]{}
\newcommand{\nownsay}[1]{}
\newcommand{\wnnay}[1]{}

\newcommand{\darius}[1]{}

\newcommand{\trung}[1]{}

\newcommand{\trungdelete}[1]{}
\newcommand{\trungdeleteOK}[1]{}

\newcommand{\todo}[1]{}

\usepackage{listings}

\usepackage[most]{tcolorbox}
\definecolor{mGray}{rgb}{0.5,0.5,0.5}
\usepackage{prooftree}
\usepackage{enumitem}
\usepackage{pifont}

\newcommand{\toolname}{\ensuremath{\ms{TypeHL}}}

\newtcblisting[auto counter]{codelisting}[2][]{sharp corners,
,colframe=gray, listing only, listing options={basicstyle=\tiny\ttfamily,language=java},
}

\usepackage{dutchcal}

\def\fdot{\coden{\,{.}\,}}
\newcommand{\code}[1]{\ensuremath{{\mathtt{#1}}}}
\newcommand{\coden}[1]{\ensuremath{{\mathtt{#1}}}}
\newcommand{\ci}[1]{\ensuremath{{\mathtt{#1}}}}

\newcommand{\ms}[1]{{\mathit{#1}}}
\newcommand{\m}[1]{{\mathit{#1}}}

\newcommand{\kw}[1]{{\tt{#1}}}

\newcommand{\gtype}[2]{\ensuremath{{\code{#1}{(\code{#2})}}}}  

\newcommand*{\notation}[2]{\expandafter\newcommand\csname #1\endcsname{{#2}}}

\newcommand{\emp}{\coden{emp}}

\notation{ann}{\code{{\it ann}}}
\notation{reqw}{\code{{\tt req}}}
\notation{ensw}{\code{{\tt ens}}}
\notation{casew}{\code{{\tt case}}}
\notation{DD}{\ensuremath{D}}
\notation{weaken}{\ensuremath{{\Rightarrow_{weaken}}}}
\newcommand{\req}{{\rm\reqw}}
\newcommand{\ecase}{{\rm\casew}}
\newcommand{\speckw}[1]{\textcolor{blue}{\textbf{\texttt{#1}}}}
\newcommand*{\eref}[1]{\coden{\m{ref}~#1}}

\notation{sh}{\ms{Sh}}
\notation{shz}{\ms{Sh}_0}
\notation{shn}{\ms{Sh}_n}
\notation{shnh}{\ms{Sh^{\#}_n}}
\notation{shh}{\ms{Sh^{\#}}}
\notation{shzh}{\ms{Sh^{\#}_0}}
\notation{rs}{\ms{Rs}}

\notation{sfh}{\#}

\newcommand*{\substassign}{{:=}}
\newcommand*{\subst}[3]{\coden{[#1\substassign#2]#3}}

\newcommand*{\ederef}[1]{\coden{!#1}}
\notation{loc}{\ell}
\newcommand{\ens}{{\rm\ensw}}

\newcommand{\vbar}{\code{\bf~|~}}

\newcommand{\true}{\m{true}}
\newcommand{\false}{\m{false}}

\newcommand{\assertPre}[1]{\req\coden{#1}}

\newcommand{\ensRet}[2]{\ensuremath{{\rm\tt{ens}}[\coden{#2}]\,\coden{#1}}}
\newcommand{\assertPostRet}[2]{\ensuremath{\ensRet{#1}{#2}}}

\newcommand{\seq}{{\bf;}}
\newcommand{\seqS}{\seq\,}

\newcommand{\equivA}{\ensuremath{~\Leftrightarrow~}}
\newcommand{\equivArr}{\ensuremath{~{\textcolor{blue}\Leftrightarrow}~}}
\newcommand{\equivArrB}{\ensuremath{\Leftrightarrow}~}
\newcommand{\implyArr}{\ensuremath{~{\textcolor{blue}\Rightarrow}~}}
\newcommand{\implyArrB}{\ensuremath{~{\Rightarrow}~}}

\newcommand{\hide}[1]{}

\definecolor{blue(pigment)}{rgb}{0.2, 0.2, 0.6}

\newcommand{\effect}{
{\ensuremath{\mathrm{\Phi}}}}

\newcommand{\uniqto}[2]{\ensuremath{{#1}{\mapsto}{#2}@U}}
\newcommand{\heapto}[2]{{#1}{\ensuremath{\mapsto}}{#2}}
\newcommand{\fnS}[1]{{\footnotesize#1}}

\newcommand{\heaptoSS}[2]{\heapto{{\fnS{#1}}}{\fnS{#2}}}

\newcommand{\typeof}[2]{{#1}{:}\,{#2}}

\newcommand*{\lam}{\lambda}

\notation{bigstepeval}{\leadsto}
\newcommand*{\bigstep}[6]{\ensuremath{#1, #2, #3 \bigstepeval #4, #5, #6}}
\notation{Rese}{R_e}
\newcommand{\incrRule}[3]{
 \begin{array}{c}
 \frac{
 \begin{array}{c}
 #1
 \end{array}
 }{
 \begin{array}{c}
 #2
 \end{array}
 }
~{#3}
\end{array}
}

\newcommand{\incrHRule}[3]{
\[
 \incrRule{#1}{#2}{#3}
\]
}

\def\ra{\ensuremath{{\rightarrow}}}
\def\sra{\ensuremath{{\fdot}}}

\newcommand{\HTriple}[3]{\{\,#1\,\}\,#2\,\{\,#3\,\}}
\newcommand{\HTripleIN}[3]{\{\,#1\,\}\,#2\,\{{\downarrow}\,#3\}}

\def\req{\coden{{\tt req{~}}}}
\def\ens{\coden{{\tt ens{~}}}}

\usepackage{pifont}

\def\sep{~\code{*}~}

\usepackage{stfloats}
\fnbelowfloat

\newcommand{\myit}[1]{\textit{#1}}

\def\xpure{\myit{Pure}}

\newcommand*{\logicallambda}{\lambda}
\newcommand{\codelamspB}[3]{\ensuremath{\logicallambda\,#1{\fdot}{#2}}}

\newcommand*{\efunC}[2]{\codelamspB{#1}{#2}{}}
\newcommand*{\tlambda}[2]{\efunC{#1}{#2}}

\def\fresh{\myit{fresh~}}

\newcommand{\heap}{\ensuremath{\mathrm{\sigma}}}
\newcommand{\pure}{\ensuremath{\mathrm{\pi}}}

\newcommand{\entailS}[2]{\ensuremath{#1 \vdash #2}}

\newcommand{\fnAnn}[2]{\code{#1\,{:}\,(#2)}}
\newcommand{\slambda}[2]{\lambda\,{#1}{\fdot}{#2}}
\newcommand{\fnAnnS}[3]{\code{#1\,{::}\,(\slambda{#2}{#3})}}

\def\tyInt{{{Int}}\,}
\def\Int{\tyInt}

\newcommand*{\todelH}[1]{}

\notation{store}{S}
\notation{hp}{h}

\newcommand*{\heapupdate}[3]{#1[#2{:=}#3]}

\newcommand*{\topI}{\Any}
\newcommand*{\AnyP}{\Any_{\mathtt{P}}}
\newcommand*{\AnyS}{\Any_{\mathtt{S}}}
\newcommand*{\botsys}{\ensuremath{\bot_{\mathtt{sys.err}}}}
\newcommand*{\resnorm}[1]{#1}

\notation{resshiftw}{\m{Shft}}

\notation{resshiftzw}{\m{Shft_0}}

\notation{resshiftnw}{\m{Shft_n}}

\notation{resshiftspw}{\m{Shft}}

\notation{resshiftspzw}{\m{Shft_0}}

\notation{resshiftspnw}{\m{Shft_n}}

\notation{rese}{R^e}
\notation{resphi}{R^\Phi}
\notation{reserr}{{\it{Err}}}
\notation{err}{{\ci{err}}}
\notation{ierr}{{\ci{ierr}}}
\notation{Err}{\ci{Err}}
\notation{IErr}{\ci{IErr}}
\notation{Dead}{\ci{Dead}}
\notation{abrt}{{\ci{abrt}}}
\notation{Abrt}{{\ci{Abrt}}}
\notation{Exc}{{\ci{Exc}}}
\notation{ChExc}{{\ci{ChExc}}}
\notation{UnExc}{{\ci{UnExc}}}
\notation{Any}{{\ci{Any}}}
\notation{coq}{{Lean}}
\notation{caseRA}{}
\notation{casewedge}{{\vdash}{\wedge}}
\notation{mSep}{{\code{|}}}
\notation{spatial}{{\it{separation}}}
\notation{Res}{R}

\newcommand*{\stagereq}{\assertPre}

\newcommand*{\stageensres}[2]{\assertPostRet{#2}{#1}}
\newcommand*{\specAnd}[2]{\coden{#1{\wedge}#2}}
\newcommand*{\specSep}[2]{\coden{#1{\sep}#2}}

\newcommand*{\specCase}[1]{\coden{\ecase~\{#1\}}}
\newcommand*{\specEns}[2]{\stageensres{#1}{#2}}
\newcommand*{\stageseq}[2]{\coden{#1{\seqS}#2}}

\newcommand*{\stageho}[2]{#1(#2)}

\newcommand*{\exi}[1]{\coden{\exists\,#1{\fdot}}}

\notation{subsume}{\mathrel{\sqsubseteq}}

\newcommand{\stagesubs}{\!\!\!\subsume\!\!\!}

\newcommand*{\seplogicmodels}[2]{\mathnormal{#1}\models#2}

\newcommand*{\bigstepres}[2]{\ensuremath{#1 \leadsto #2}}
\newcommand*{\vfun}[2]{\ensuremath{\lam\,#1\sra#2}}
\newcommand*{\vfunrel}[2]{\ensuremath{\lam\,#1\sra#2}}

\newcommand*{\ctr}{\ensuremath{\mathrm{\Delta}}}
\newcommand*{\dspec}{\effect}
\notation{normsto}{\leadsto_N}
\notation{reducesto}{\leadsto_R}

\newcommand*{\stagedisj}[2]{\coden{#1{\vee}#2}}

\newcommand*{\andd}{\ \myit{and}\ }
\newcommand*{\orr}{\ \myit{or}\ }
\newcommand*{\negg}{\myit{not}}

\newcommand*{\elet}[3]{\coden{\kw{let}~#1\,{=}\,#2~\kw{in}~#3}}

\newcommand*{\eletBlk}[2]{\coden{\kw{let}~#1~\kw{in}~#2}}

\newcommand*{\ecall}[2]{#1(#2)}
\newcommand*{\eassign}[2]{\coden{#1~{:=}~#2}}
\newcommand*{\econs}[2]{#1{::}#2}

\newcommand*{\ematch}[2]{\coden{\kw{match}~#1~\{#2\}}}
\newcommand*{\etrycatch}[2]{\coden{\kw{try}~#1~\kw{catch}~#2}}
\newcommand*{\efun}{\vfun}

\newcommand*{\stage}{\ensuremath{E}}
\notation{stageexists}{\exists}
\newcommand*{\stageex}[2]{\ensuremath{\stageexists\,#1{\fdot}{\coden{#2}}}}

\newcommand*{\heapdomain}{\m{dom}}
\newcommand*{\heapdom}[1]{\heapdomain(#1)}

\notation{specmodels}{\models}

\notation{specfinalstatew}{\models}

\newcommand*{\heapstoreupd}{{:=}}

\newcommand{\locle}[1]{}
\newcommand{\num}[1]{{\ttfamily \color{purple}{{{\#No}}}}}

\usepackage{thmtools}

\usepackage{thm-restate}

\usepackage{hyperref}

\definecolor{darklavender}{rgb}{0.3, 0.16, 0.4}
\usepackage{mathpartir}
\definecolor{babypink}{rgb}{0.96, 0.76, 0.76}
\definecolor{mossgreen}{rgb}{0.68, 0.87, 0.68}

\usepackage[mathscr]{euscript}
 \let\mathscr\relax
\usepackage[scr]{rsfso}

\newcommand*{\defeq}{\stackrel{\text{def}}{=}}

\newcommand{\syh}[1]{}
\definecolor{darkmagenta}{rgb}{0.55, 0.0, 0.55}

\definecolor{darkred}{rgb}{0.55, 0.0, 0.0}
\definecolor{forestgreen}{rgb}{0.13, 0.55, 0.13}
\usepackage{color}

\definecolor{pblue}{rgb}{0.13,0.13,1}
\definecolor{pgreen}{rgb}{0,0.5,0}
\definecolor{pred}{rgb}{0.9,0,0}
\definecolor{pgrey}{rgb}{0.46,0.45,0.48}
\usepackage[capitalise]{cleveref}
\crefname{section}{Sec.}{Secs.}
\Crefname{section}{Sec.}{Secs.}
\crefname{subsection}{Sec.}{Secs.}
\Crefname{subsection}{Sec.}{Secs.}
\crefname{subsubsection}{Sec.}{Secs.}
\Crefname{subsubsection}{Sec.}{Secs.}
\definecolor{airforceblue}{rgb}{0.36, 0.54, 0.66}

\usepackage{listings}

\begin{document}
\emergencystretch=3em

\title{{\large Type Safety via Hoare Logic with Separation and Pure Types}}

\author{Wenhua Li \quad Darius Foo \quad Quang Trung Ta \quad Wei-Ngan Chin}
\renewcommand{\shortauthors}{Li, Foo, Ta, and Chin}
\authorsaddresses{Department of Computer Science, National University of Singapore. \{liwenhua, dariusf, taqt, chinwn\}@nus.edu.sg}
\begin{abstract}
  Type safety has traditionally rested on carefully crafted type systems, under the
motto {\em ``well-typed programs cannot go wrong''}. Modern demands push type systems
past this basic guarantee: toward greater memory safety (e.g., Rust), stronger
data-structure invariants (e.g., Haskell's GADTs), and broader typability (e.g.,
MLstruct). In principle the classic motto absorbs each such property by enlarging the
set of states deemed ``wrong''; but this collapses them into a single binary verdict,
erasing the very distinctions that make them valuable. Heap ownership, flow-sensitive
changes to a variable's type, and the gap between a recoverable and a fatal error are
relational, stateful facts about a program's \emph{intermediate} states, facts one
``go wrong'' verdict cannot tell apart. Worse, each new demand typically arrives as its
own custom extension, making it hard to say what any one guarantees, whether two are
compatible, or how they combine.

What is missing is a single foundation in which such properties can be stated, compared,
and combined; Floyd-Hoare logic supplies it. We present a framework for type-safety
verification built from four ingredients:
(i) {\em case specifications} for path-sensitive typing;
(ii) {\em separation types}, inspired by separation logic, for flow-sensitive type
mutation and must-aliasing;
(iii) a disciplined distinction between \code{Err} (runtime error values our types track)
and \code{Abrt} (compile-time errors), yielding the refined motto {\em well-typed
programs must never abort}; and
(iv) {\em type predicates} for rich data-structure invariants.
All four are ordinary types in a single Boolean algebra rather than one-off extensions.
The framework therefore subsumes both GADTs and liquid types within one type logic
and spans a spectrum of guarantees: from weak specifications that tolerate \code{Err}
values to strong ones that eliminate them entirely. The algebra also keeps checking
tractable: subtyping reduces to a single decidable emptiness test, so one lightweight
procedure serves the whole framework with no SMT oracle in its trusted base.
We formalise the Hoare rules and prove soundness in a machine-checked Lean mechanisation;
by proof reflection, the mechanisation itself yields a self-certifying type-checker, which
we evaluate on a suite of benchmark programs.



\end{abstract}
\maketitle              


\newcommand{\qedbox}{$\square$}

\newtheorem{defn}{Definition}[section]
\newtheorem{remark2}{Remark}[section]
\newtheorem{remarkA}{Remark}[section]
\newtheorem{remark}{Remark}[section]
\section{Introduction}
\label{sec:intro}
Type systems are being asked to do more than ever. Beyond the classic guarantee that well-typed
programs do not go wrong, modern languages demand richer properties: memory safety through
ownership and borrowing (Rust), stronger data structure invariants through generalised algebraic
data types (GADTs in OCaml and Haskell), broader typability through semantic subtyping (MLsub
\cite{Doplan:POPL17}, MLstruct \cite{Lionel:OOPSLA22}),
and low-latency execution through modal resource management (OxCaml  \cite{LorenzenWDEL24,oxcaml}).
Dynamically-typed languages face the same pressure from a different direction, as gradual typing
retrofits (e.g. Hack \cite{hacklang} and Typed Racket \cite{tobin-hochstadt2008typedscheme}) attempt to recover static guarantees incrementally.
Folding all of these into a single ``go wrong'' predicate blurs distinctions that matter about a
program's \emph{intermediate} states: whether a heap location is owned, how a variable's type
evolves as the program runs, and whether an error
is a recoverable \code{Err} to be tolerated or a fatal \code{Abrt} to be ruled out
(\cref{sec:error}). This last distinction is the one we lean on throughout.
Meanwhile, the foundational framework has not kept pace: each such feature
arrives as its own extension, leaving it unclear how the resulting guarantees relate or compose.
%

In this paper, we argue that Floyd-Hoare logic provides exactly the right foundation for a unified
treatment of type safety. Hoare logic is inherently modular: the type property of each function is
captured by a pre/post specification that can be verified and reused independently. It is
flow-sensitive by design, tracking how types evolve through a program rather than assigning a
single static type to each variable. It is also expressive enough to accommodate both statically- and
dynamically-typed programs within the same framework, by adjusting the strength of the
specifications rather than changing the underlying logic.
\textit{The key question is whether type
signatures (the concise, familiar form of type information) can be faithfully and conveniently
represented as Hoare-style pre/post specifications.}
We show that they can, and that the resulting
framework supports strictly more than standard type systems.

We now illustrate these points through a series of examples, building from simple type signatures to
flow-sensitive specifications for heap-mutating programs.
Consider two functions from
an ML-like language: a conversion function \coden{str\_of\_int}
from \code{Int} to its \code{Str} counterpart
; and a polymorphic \code{(!)}
operator to dereference a mutable \code{Ref}.

\begin{minipage}{.8\textwidth}
\begin{lstlisting}[numbers=none]
  str_of_int : Int (*@$\rightarrow$@*) Str
  (!)        : (*@$\forall$@*)T. Ref(T) (*@$\rightarrow$@*) T
\end{lstlisting}
\end{minipage}


While type signatures are concise, it is almost as
easy to write their type specifications using pre/postconditions in Hoare logic. For the above examples, we can obtain the
corresponding type specifications by suitably naming
the inputs \code{v^*} and output \code{r} of each function using
the precondition clause \code{req[v^*]} and the
postcondition clause \code{ens[r]}, respectively (their formal definition is given in Sec.~\ref{sec:type-logic}),
as shown below.

\noindent
\begin{minipage}{.98\textwidth}
\begin{lstlisting}[numbers=none]
str_of_int : req[x] x:Int   ens[r] r:Str
(!)        : (*@$\forall$@*)T. req[x] x:Ref(T)   ens[r] r:T
\end{lstlisting}
\end{minipage}
\label{ex:apply}

For the examples above, each type signature of
the form \code{t_1~{\ra}~t_2} can be directly
translated into a type specification of the form
\code{req[x]\,x{:}t_1\,ens[r]\,r{:}t_2}, where
\code{x} and \code{r} are freshly named.
So far, the gain over type signatures is mainly notational.

\textit{The real advantage of Hoare logic}
emerges when type signatures fall short: a variable's type may \emph{change} as the program runs,
and the mutated location may be \emph{aliased}. In the simplified program below,
a field is coerced in place from its raw \code{Str} form to an \code{Int} (as when normalising
configuration or JSON data).
For instance, normalising a JSON record like \texttt{\{"port": "8080"\}} rewrites the
string-valued \texttt{port} field in place to the integer \texttt{8080}, which downstream code then
uses in arithmetic.
A mutable object is typically reached through several names --
an earlier binding, a record field, a captured variable, a function argument
-- so we use \code{alias} to illustrate an aliased reference. The write goes through \code{box}, while \code{alias} (bound earlier)
stays live and is read afterwards. No mainstream type system tracks this combination of type
mutation and aliasing soundly.

\hide{\trung{I think we can highlight these 2 points as the challenges of type signatures in the first paragraph. Then, we can link to that paragraph and make a coherent argument how the type specification can address these challenges.}\wnsay{Type signature are just syntactic sugar.
  I hesitate to claim that our contribution is type-specification}}

\noindent
{\lstset{morecomment=[l]{//},commentstyle=\color{black}}%
\begin{lstlisting}[numbers=none]
 coerce(b)   = update(b, str_of_int(!b))  // in-place type change
 caller(box) = let alias = box in         // alias, box: must-aliases
               coerce(box);               // box now stores an Int
               !alias + 1                 // safe iff alias : Int
\end{lstlisting}}

To capture this kind of flow-sensitive type change, a pure type \code{box{:}Ref(Any)} is too weak:
it fixes a single stored type and loses the precise \code{Str\,{\to}\,Int} change. We instead
introduce the \textit{separation type} \code{\heapto{box}{Ref(t_1)}}, which carries full ownership of the heap
location at \code{box} and can be updated to a different type \code{\heapto{box}{Ref(t_2)}} as the
program evolves. The notation is inspired by separation
logic~\cite{DBLP:conf/lics/Reynolds02}, but serves a different purpose. Separation logic reasons about
functional correctness of programs \emph{already assumed to be type-checked}. By contrast, we aim to deliver stronger \emph{type-safety} guarantees, here soundly tracking the
type of an aliased, mutable heap location.
Coupled with pre/post type
specification, we can now provide
a precise, flow-sensitive type specification for \code{coerce}, and, more generally, for the
\code{update} it calls, as shown below.


\noindent
\begin{lstlisting}[numbers=none]
coerce : req[b] b(*@$\mapsto$@*)Ref(Str) ens[(*@\_@*)] b(*@$\mapsto$@*)Ref(Int)
update : (*@$\forall$@*)T. case[x,u] { x:Ref(T) (*@$\land$@*) u:T (*@${\caseRA}$@*) ens[r] r:() ;
                         x(*@$\mapsto$@*)Ref(_) (*@$\land$@*) u:Any (*@${\caseRA}$@*) ens[r] x(*@$\mapsto$@*)Ref({u}) (*@$\land$@*) r:() }
\end{lstlisting}

Here \code{case[v^*]\{\ g_1 \caseRA s_1\,;\ \ldots\,\}} names the inputs \code{v^*} and lists
guarded behaviours: when the pre-state satisfies guard \code{g_i}, the function behaves as
\code{s_i} (an \code{ens[r]}-postcondition on the result \code{r}). The clauses, separated by
\code{;}, are pairwise disjoint, so at most one matches any call. (The separation type
\code{\heapto{x}{T}} and the \code{\_} wildcard are introduced in Sec.~\ref{sec:flow-sen}.)

For the \code{update} method, we use a case specification
to capture both
the flow-insensitive pure-type scenario \code{x{:}Ref(T)}, and
the flow-sensitive separation-type scenario \code{\heapto{x}{Ref(\_)}}.
Which case applies at a given call is determined by the caller's static pre-state: the case
whose guard the caller's incoming types satisfy is the one selected. Thus, in a pure-type context the first case applies,
while a caller holding full ownership \code{\heapto{x}{Ref(\_)}} selects the second. In that case's
postcondition \code{\heapto{x}{Ref(\{u\})}}, the value argument \code{u} is used as the singleton type \code{\{u\}},
recording that the cell now stores exactly \code{u}. To be explicit, we can write \code{\{u\}} to
denote each singleton type \code{u}. Since pure
and separation types are disjoint (\code{x{:}Ref(T)} and \code{\heapto{x}{Ref(\_)}} can
never hold at once), the two cases are mutually exclusive.

\paragraph{Why this is beyond existing type systems.} Trace the aliased coercion. The call
\code{coerce(box)} selects the separation case of \code{update}, rewriting \code{box}'s stored type
from \code{Str} to \code{Int}; the post-state records \code{\heapto{box}{Ref(Int)}}. The subtle point
is \code{alias}: again using a variable as a singleton type, we record the must-alias as
\code{alias{:}\{box\}} (read \code{alias{=}box}). Our logic then propagates the change,
so \code{\heapto{alias}{Ref(Int)}} holds afterwards and \code{!alias\,{+}\,1} type-checks. Type mutation
and must-aliasing are both essential: drop either and the program is wrongly rejected or unsoundly
accepted. This is precisely where existing systems fall short.
\begin{itemize}[leftmargin=1.4em,itemsep=1pt,topsep=2pt]
\item \emph{OCaml, Haskell.} A mutable cell (\code{'a~ref}, \code{IORef~a}) fixes its stored type at
creation: only its \emph{value}, never its type, may change, so the \code{Str\,{\to}\,Int} update is
inexpressible.
\item \emph{Rust, uniqueness/linear types.} Strong updates are permitted, but only for
\emph{unaliased} locations; the live \code{alias} is rejected outright, so the safe program cannot be
written without an \code{unsafe} or \code{Rc{<}RefCell{>}} escape hatch.
\item \emph{TypeScript (gradual).} Control-flow narrowing may change \code{box}'s static type
locally, but the refinement is discarded across the \code{coerce} call and never propagates to
\code{alias}; the code type-checks only via an unchecked cast \code{(!alias~as~Int)} that silently
escapes the type system.\footnote{See e.g. TypeScript issues \#35972 and \#51851 on narrowing being
lost across calls and bypassed under aliasing.}
\item \emph{Python (dynamic).} The mutation is permitted but wholly unchecked: if \code{coerce} is
omitted or misordered, \code{!alias\,{+}\,1} raises a runtime \code{TypeError}, an \code{\Err}
our framework rules out statically.
\end{itemize}
No mainstream type system soundly tracks the stored type of an \emph{aliased} heap location as it
changes; our separation types, combining full ownership with must-aliasing, close exactly this gap.
This use of a variable singleton to record aliasing is not itself new; it underpins Scala's
path-dependent types and their formal calculus pDOT~\cite{RapoportL19}, where a singleton
\code{x.type} states that two paths denote the same object. There, however, singletons are confined
to \emph{immutable} (stable) paths, precisely because a value's type never changes, so they support
only type-\emph{equality} reasoning. We instead pair the variable singleton with a separation type
carrying full ownership, which licenses a \emph{type-changing} update of the aliased
cell, the capability those systems by design cannot support.

In this paper, we develop a
Hoare logic framework for type safety in
which
{\em separation} types (with must-aliasing) provide flow-sensitivity,
while flow-insensitive {\em pure} types
(with arbitrary aliasing) are used whenever type mutation is not required.

A further benefit of the distinct \code{Err}/\code{Abrt} hierarchy is more precise error reporting:
an \code{Err} is a recoverable value that may live in a well-typed program, whereas an \code{Abrt}
marks a genuine type error flagged as a compilation error, rather than collapsing both into a
single opaque ``type error''.

Our main contributions are:
\begin{itemize}[nosep]
\item We propose a
Floyd-Hoare logic approach to type-safety verification, and introduce separation types to support
must-aliasing and flow-sensitive type mutation. Used with flow-insensitive pure types, these two
provide a unified framework with wider coverage.

  \item Our approach supports
  path-sensitivity via
  case specifications, and strengthened preconditions
  via type predicates,
  subsuming GADTs and liquid types. 
\item We distinguish
  recoverable runtime errors (\code{\Err})
  from
  compile-time errors (\code{\Abrt}), absent in well-typed programs, and fold
   exceptions (\code{\Exc}) as first-class types.
  \hide{
  \item To support generic types with subtyping, we allow variance
  annotations for both pure and separation types.
  Variance annotation can support more precise subtyping
  based on how values are flowing into and out of data structures.
  For separation types that are unchanged in
  the post-condition, we use the borrow mechanism
  for more 
  precise type specifications.
  }
\item We design a small core language, define its operational
  semantics, and construct a \emph{single} type logic in which
  \emph{separation} types and \emph{pure} types coexist as first-class citizens.
  Distinctively, the two are related by a controlled \emph{one-way} conversion: a separation
  type can be weakened to its corresponding pure type (giving up heap ownership), but
  the weakening is \emph{irreversible}, since ownership once released cannot be reclaimed.

\item We give a self-contained machine-checked
  soundness proof in Lean~4 (Sec.~\ref{sec:lean}) across the pure, separation, and
  higher-order fragments, proving well-typed programs never abort.

\item \emph{Expressive, certified type-checking.} From the mechanisation we obtain a
  \emph{self-certifying type-checker}: a checker proved sound once against the Hoare rules,
  certifying each program by evaluation. One automatic checker covers both fast pure-type
  checking and expressive separation-type and predicate specs, so automation and
  expressiveness are not traded off.

\end{itemize}

\paragraph{Positioning.}
Our framework occupies a deliberate point on two axes. On the
\emph{automation} versus \emph{expressiveness} axis it sits between automatic SMT-based
separation-logic provers and interactive program logics such as Iris. Like the
former it is fully automatic and, in fact, needs \emph{no} SMT solver (entailment
is a lightweight, custom procedure). Yet like the latter it carries a
machine-checked soundness proof, discharged \emph{once} for the rule set
(Sec.~\ref{sec:lean}), not per program. On the \emph{property} axis it targets
\emph{type safety} (semantic subtyping, flow-sensitive type mutation, and the
\code{Err}/\code{Abrt} distinction) rather than arbitrary functional
correctness; scoping to type properties is precisely what buys back full
automation. The payoff is a single checker that is at once automatic, SMT-free,
and certified sound, while reaching properties no mainstream type system offers:
the coexistence of separation and pure types under a one-way weakening, and the
subsumption of GADTs and liquid types within one type logic. Concretely, for the \emph{checking} direction, this
answers in the affirmative whether such
a logic could serve as the type system of a compiler that verifies well-typedness
automatically. App.~\ref{sec:leanchecker} realises the equality fragment of a
decidable core (App.~\ref{sec:decidable}), and, for the mechanised
\code{SortedList} flagships' obligations, the inequality fragment, as a
machine-checked, executable checker
that certifies each program's typing derivation by reflection
(Sec.~\ref{sec:checkerlia}, \ref{sec:checkerlia-ineq}). Deciding the inequality fragment in general, and type
\emph{inference}, remain future work. Detailed comparisons appear in
Sec.~\ref{sec:compare} (GADTs, liquid types) and Sec.~\ref{sec:related}
(related logics).

\paragraph{Paper structure.}
The rest of our paper is organised as follows.
Sec.~\ref{sec:motivate} covers the novel aspects of our
type logic via examples and can be read independently as
an overview.
Sec.~\ref{sec:lang} presents a core language and
the type logic 
that 
supports 
expressive type specifications 
for both statically-typed
and dynamically-typed languages.
Sec.~\ref{sec:Hoare} presents 
forward-style Hoare rules for type safety, building on 
Sec.~\ref{sec:lang}.
Sec.~\ref{sec:compare} shows how our type predicate framework
subsumes both GADTs and liquid types.
Sec.~\ref{sec:Soundness} defines the semantics of
our type logic and establishes
soundness of the Hoare rules. 
Sec.~\ref{sec:lean} presents the machine-checked Lean mechanisation.
Sec.~\ref{sec:related} discusses related work.
The appendices cover additional features and extensions.
\hide{
  We also present our solutions for stronger types,
GADT (see \appref{sec:GADT})
and type-classes (see \appref{sec:typeclasses}).
Related works are discussed in Sec.~\ref{sec:related}.
Sec.~\ref{sec:exception} shows how
checked and unchecked exceptions are handled in our core language..
}


\section{Improvements to Type-Safety via Hoare Logic}
\label{sec:motivate}

The following four sections introduce four novel ingredients
of our proposal.
Case specifications
(Sec.~\ref{sec:semsubtype}) give path-sensitivity.
Separation types (Sec.~\ref{sec:flow-sen})
provide the expressive power needed for flow-sensitive heap reasoning.
The \code{\Err}/\code{\Abrt}
error hierarchy (Sec.~\ref{sec:error}) makes error
coverage explicit, verifiable and
customisable to the precision required by the caller. Type predicates (Sec.~\ref{sec:GADT}) capture data structure
invariants. Although independent contributions, all four compose coherently, as we summarise
at the end of Sec.~\ref{sec:motivate}.


\hide{
  This section presents four novel aspects of our Hoare-logic approach to type safety, each
addressing a different limitation of standard type systems. Section 2.1 shows how semantic
subtyping, lifted into a logic of types under Hoare-style reasoning, enables path-sensitive type
specifications via case specifications. Section 2.2 introduces {\em separation types} for flow-sensitive
heap reasoning and shows how they interact with pure types and aliasing. Section 2.3
introduces a disciplined error hierarchy -- distinguishing recoverable runtime errors (Err) from
compile-time aborts (Abrt) -- and explains why this distinction is essential for our refined notion
of type safety. Section 2.4 introduces {\em type predicates} as a complement to algebraic data types .....
These three ingredients are independent contributions, but they work together:
case specifications give path-sensitivity, separation types give flow-sensitivity, error
hierarchy makes the coverage of type safety explicit and
{\em type predicates} provides a systematic way to support
data structure invariants.
}

\hide{We close
the section with a combined example (Section 2.5) that uses all four ingredients simultaneously, previewing
how they interact before the formal machinery of Section 3 begins.}

\hide{
  The first two features are \underline{adapted} from existing
type systems. 
\trung{I think if we emphasize that 2/3 of the features are adapted from existing works, then we will be complained for the lack of novelty.}
The last feature cover our \underline{novel} flow-sensitive
separation types and how they are related to pure types.
\hide{
\begin{wrapfigure}{r}{0.38\textwidth}
  \begin{tikzpicture}[x=1.8cm,y=1.0cm]
\node at (0,-1.5)        (Bot)       {$\bot$};

\node at (-0.5,-0.5)     (botSys)    {\botsys};

\node at (-0.5,0.5)      (UnExc)     {\UnExc};
\node at (0.5,0.5)       (ChExc)     {\ChExc};
\node at (0.5,2.4)       (Exc)       {\Exc};

\node at (1.5,1.5)       (Abrt)      {\Abrt};
\node at (-1.5,1.5)      (Err)       {\Err};

\node at (-1.0,2.4)      (AnyP)      {$\AnyP$};
\node at (-0.5,2.4)      (AnyS)      {$\AnyS$};
\node at (-0.5,3.3)      (Any)       {\Any};

\node at (0,4.2)         (Top)       {$\top$};

\draw[->] (Bot)  --  (botSys);
\draw[->] (Bot)  --  (Abrt);
\draw[->] (botSys)  --  (UnExc);
\draw[->] (botSys)  --  (ChExc);
\draw[->] (botSys)  --  (Err);
\draw[->] (UnExc)  --  (Exc);
\draw[->] (UnExc)  --  (AnyP);
\draw[->] (UnExc)  --  (AnyS);
\draw[->] (ChExc)  --  (Exc);
\draw[->] (Err)  --  (AnyP);
\draw[->] (AnyP)  --  (Any);
\draw[->] (AnyS)  --  (Any);
\draw[->] (Any)  --  (Top);
\draw[->] (Exc)  --  (Top);
\draw[->] (Abrt)  --  (Top);
\end{tikzpicture}

\caption{Lattice of Types}
\label{fig:lattice}
\end{wrapfigure}
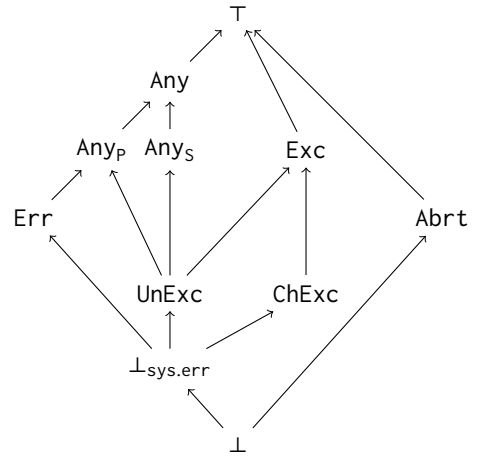
}

\trung{I realize that our presentation from Section 2 to Section 3 is a bottom-up approach. We describe many details in Section 2 (which can be hard to follow for the reviewers since there are many details), and then will only present the overall language \& type design in Section 3. The reviewers might base on the fact that some of the details in Section 2 are adopted from existing works to reject our paper.

How about we follow a top-down approach by first presenting the overall language \& type design (move Section 3 to Section 2), and then present the details of each feature in Section 3 (which is currently Section 2). This way, in Section 2, we can claim that our language and type design is novel, and then in Section 3, we can elaborate on each feature (some of which are adapted from existing works). This way, the reviewers cannot complain about the lack of novelty of the overall framework.
}
\wnsay{You have a good point.
  Let me rewrite it to highlight novel aspects of these
  three points instead.}
}
\hide{
  \trung{When reading into this section, I realize that if we can address the challenges in the Introduction to align with these 3 features, then we can have a more consistent and coherent story about our type specification framework.}
\wnsay{Error types and semantic subtying are not covered in Intro, since they
  are important but not significant contributions}
}










\subsection{Semantic Subtyping, Path-Sensitivity and Case Specifications}
\label{sec:semsubtype}
Semantic-based subtyping interprets each type \code{t}
as a set of values, and  defines subtyping
\code{t_1{<:}t_2} 
via
set containment \code{Set(t_1){\subseteq}Set(t_2)}.
This lets types be composed using a Boolean algebra:

\centerline{$t := t\lor t\ |\ t\land t\ |\ \neg t\ |\ ...$}

Consider the ML-style function below. ML's type system rejects it because unification fails when it tries to
combine the \code{Int} output of one branch with the
\code{Str} output of the other:

\begin{lstlisting}[numbers=none]
inc_if_int(x) = match x of { Int -> x+1;
                             _   -> "Not supported!" }
\end{lstlisting}


Semantic subtyping resolves this: one valid formulation assigns the output type as \code{Int {\vee} Str},
yielding the following type signature:

\begin{lstlisting}[numbers=none, commentstyle=\color{gray}]
inc_if_int :  Any (*@$\rightarrow$@*) Int (*@$\lor$@*) Str
\end{lstlisting}

The most precise set-theoretic type is an intersection of
arrow types, \code{(Int {\to} Int) \wedge (\neg{Int} \to Str)}
\cite{DBLP:books/sp/24/Castagna24},
which our case specification below directly mirrors:


\begin{lstlisting}[name=map, numbers=none]
inc_if_int : case[x] { x:Int  (*@${\caseRA}$@*) ens[r] r:Int;
                       x:(*@$\neg$@*)Int (*@${\caseRA}$@*) ens[r] r:Str }
\end{lstlisting}

Case specifications differ from intersection types in one crucial respect: disjointness is a
verified obligation, not a programmer convention. The type checker rejects any case
specification whose guards overlap, making the specification exhaustive and non-redundant by
construction. A programmer writing the intersection type
 \code{(Int {\ra} Int) \wedge (\Any {\ra} Int \vee Str)} receives
no such check: the overlapping cases pass silently. This enforcement is what makes the
otherwise clause and the \code{\Abrt} discipline of Sec.~\ref{sec:error} well-defined.

\label{sec:path}

\hide{
\begin{verbatim}
Using intersection type (from POPL24), it
is possible to infer:
  f : (Int -> Int)
    /\ (~Int -> Str)

Alternatively, we could also use multi pre/post
which has the same expressivity as intersection type.

  f : req[x] x:Int ens[r] r:Int
    /\ req[x] x:~Int ens[r] r:Str

Better still, we can use case-specification..
as shown below, which allows more complete reasoning
(with the help of case analysis) to be supported.

  f : case[x] { x:Int ens[r] r:Int;
               x:~Int ens[r] r:Str }

valid types = primitives, data constructor, record types

multiple inheritance allowed so that traits and
interfaces may also be supported.

recent work in this space has led to
advancement in type inference algorithms for
semantic-based typing

Our proposal here advancement in spec and verification setting..
\end{verbatim}
}




\subsection{Flow-Sensitivity and Separation Types}
\label{sec:spatial}
\label{sec:flow-sen}

Many languages support pass-by-value-result
parameters (e.g. \code{inout} parameters in Swift) that allow variable
arguments
to be mutated by a
method call. Dynamically-typed languages go further, allowing
the type of a variable to change arbitrarily, as follows:





\begin{lstlisting}[numbers=none]
  inc_transform (inout x) = x := str_of_int(x+1)
\end{lstlisting}

One might suggest \code{Int{\vee}Str\,{\ra}\,\texttt{()}} as the type signature, but this is not
type-safe. Standard type systems (including Swift's) are flow-insensitive, so the one type assigned
to \code{x} must serve both as the \code{Int} argument to \code{(+)} and as the \code{Str} value finally
stored, and \code{Int{\vee}Str} satisfies neither.

To type such a program, we make the mutation explicit in the heap: the mutable parameter
\code{x} is modelled as a heap cell of type \code{Ref(\_)}, dereferenced by the body
(\code{!x}), which writes back the string form via a strong update.
\begin{lstlisting}[numbers=none]
  inc_transform(x) = let v = str_of_int(!x+1) in update(x,v)
\end{lstlisting}
This write changes the cell's \emph{type}, not merely its value: \code{x}
points to a \code{Ref(Int)} on entry and to a \code{Ref(Str)} on exit, so a type system must
support a \emph{flow-sensitive type change on a heap location}. Separation types capture exactly this, recording the owned cell's type before and
after the call:
\begin{lstlisting}[numbers=none]
inc_transform : req[x] x(*@$\mapsto$@*)Ref(Int) ens[_] x(*@$\mapsto$@*)Ref(Str)
\end{lstlisting}
The precondition owns \code{\heapto{x}{Ref(Int)}} and the postcondition the strongly-updated
\code{\heapto{x}{Ref(Str)}}, a change that neither a single flow-insensitive type for \code{x}
nor a plain \code{Ref(Int{\vee}Str)} can express.

\paragraph{Where existing type systems stand.} As the introduction's comparison showed, existing
systems handle such type-changing updates only partially, if at all: none soundly tracks the stored
type of a \emph{heap} location as it changes under aliasing. Closing this gap requires accounting for
must-aliases among heap locations. Linear (uniqueness) types sidestep it by guaranteeing a single
reference per location and no aliases at all, but thereby rule out programs where aliasing is
intentional. We instead introduce a {\bf separation type}, written \code{\heapto{x}{T}}, which carries
full ownership of the heap location at \code{x} while still tracking must-aliases (references
guaranteed to point to the same location) explicitly. Full ownership is what makes a strong update
sound: every alias of an owned location is a known must-alias, so all aliases observe the new type
consistently. By contrast, pure types permit arbitrary (unknown) aliasing and so forbid type mutation.

\paragraph{Why must-aliasing matters: a preview.} We first highlight the key feature that
distinguishes separation types from uniqueness types. Consider the \code{swap} function:
\begin{lstlisting}[numbers=none]
 swap(x,y) = let v1=!x in let v2=!y in
             update(x,v2); update(y,v1)
\end{lstlisting}
The interesting case is when \code{x} and \code{y} must-alias the same location, making the swap a
no-op: uniqueness types cannot express it (they forbid aliasing), whereas separation types capture it
via the singleton \code{y{:}\{x\}} (meaning \code{y{=}x}). We give the full case specification below.



\paragraph{Formal specifications.} Separation types yield stronger, flow-sensitive specifications for the core
heap operations. We also introduce spatial conjunction \code{\sep} for pairs of disjoint heap locations:

\begin{lstlisting}[numbers=none]
mkRef: req[x] x:(*@$\topI$@*) ens[r] r(*@$\mapsto$@*)Ref(x)
(!): (*@$\forall$@*)T:(*@$\AnyP$@*),Q:(*@$\topI$@*). case[m] { m:Ref(T) (*@${\caseRA}$@*) ens[r] r:T;
                            m(*@$\mapsto$@*)Ref(Q) (*@${\caseRA}$@*) ens[r] m(*@$\mapsto$@*)Ref(Q) (*@$\land$@*) r:Q}
update: (*@$\forall$@*)T:(*@$\AnyP$@*). case[m,v] { m:Ref(T) (*@$\land$@*) v:T (*@${\caseRA}$@*) ens[r] r:();
                            m(*@$\mapsto$@*)Ref(_) (*@$\land$@*) v:(*@$\topI$@*) (*@${\caseRA}$@*) ens[r] m(*@$\mapsto$@*)Ref(v) (*@$\land$@*) r:()}
\end{lstlisting}

Here \code{Any} is a top type, split into the pure types \code{Any_P} (arbitrary aliasing) and the
separation types \code{Any_S} (exclusive heap ownership, covering mutable data).

Throughout, lowercase identifiers (\code{x,y,u,v}) denote value variables and single uppercase letters
(\code{T,R,A}) type variables. Accordingly, \code{x{:}t} is used in two ways: for a value variable \code{x}, it asserts \code{x} has type \code{t}; for a type variable \code{T}, the
bounded quantifier  \code{\forall T{:}R} means  \code{T} ranges over subtypes of  \code{R}; subtyping is otherwise written
 \code{T {<:} R}. Thus \code{x{:}Ref(T)} is a pure type assertion on a value. A reference \code{Ref(T)} stores a
value of type \code{T}; \code{Ref(u)} (as in \code{\heapto{x}{Ref(u)}}) stores the specific value \code{u}, the
singleton \code{Ref(\{u\})}, and \code{\heapto{x}{Ref(y)}} asserts heap ownership where \code{y} is a type or
singleton \code{\{v\}}. Predicate membership \code{x{:}P(\ldots)} uses the same colon: \code{x} is a value of the
(inductive) predicate \code{P(\ldots)} (Sec.~\ref{sec:pred}).

The precondition of \code{mkRef} is maximally weak (\code{Any} covers both pure and separation types), its
postcondition maximally strong. The two-case structure of \code{update} reflects the pure vs.\ separation distinction, made precise by
Definitions~\ref{def:type-consistent-pure} and~\ref{def:type-consistent-sep} below: a pure \code{m{:}Ref(T)}
permits only a type-consistent write, returning \code{()} unchanged; owned \code{\heapto{m}{Ref(\_)}} permits
type mutation, yielding \code{\heapto{m}{Ref(v)}}. A single specification serves both regimes, the case
guard selecting which applies.

\begin{defn}[Type-Consistent Mutation for Pure Types]
\label{def:type-consistent}
\label{def:type-consistent-pure}
A write update(x,v) to a location \code{x{:}Ref(T)} in pre-state \code{\ctr}
is {\em type-consistent} if \code{v{:}T}.
\end{defn}

\noindent For pure types, changing the stored type (\code{T} to \code{T_2{\neq}T}) is disallowed:
arbitrary aliasing would let other references to \code{x} observe an inconsistent type.

\begin{defn}[Type Mutation for Separation Types]
\label{def:type-consistent-sep}
Type mutation at a location
\code{\heapto{x}{Ref(T)}} is the replacement of the stored type \code{T} with a distinct type \code{T_2}, effected by a write
\code{update(x,v)} where \code{v{:}T_2}. The postcondition reflects the updated type: \code{\heapto{x}{Ref(T_2)}}.
\end{defn}

\noindent This is sound because full ownership rules out any unknown alias of \code{x} in the pure heap.

\paragraph{Irreversible weakening.} Separation and pure types occupy two disjoint heap partitions. The
{\em separation heap} gives exclusive ownership \code{\heapto{x}{T}} (full rights to read, mutate the
type, and track must-aliases); the {\em pure heap} gives only shared
access \code{x{:}T} (no ownership, and type-consistent writes only). Accordingly, a separation type
\code{x{\mapsto}T} for heap memory
is \emph{not} a subtype of the corresponding pure type \code{x{:}T}: the two are
mutually disjoint, connected
only by a one-way, irreversible weakening, not by subtyping:

\vspace{-0.5em}
\begin{small}
\[
\begin{array}{lll}
  \code{x{\mapsto}T} ~\land~\code{T\,{<:}\,A}~\land~ \code{A\,{<:}\,Any_P} 
  & \weaken &  \code{x:A}
\end{array}
\]
\end{small}

Since this weakening is the sole crossing rule, a separation type may nest pure components but never
the reverse, so well-formedness holds by construction, with no separate check required.
Once ownership is surrendered via weakening, the right to type-mutate is permanently lost:
\code{\heapto{r}{Ref(T)}} (from \code{mkRef}) can be weakened to  \code{r:Ref(T)} when type mutation is no
longer required. The reverse is disallowed, since other aliases may already exist. This relationship is captured by the equivalences below (\code{\equivA} is semantic equivalence); it differs for primitive and non-primitive types:
\[
\begin{array}{ll}
\text{for a primitive } \code{p} \text{ (e.g.\ } \code{Int}, \code{Str}\text{):} & \code{x{\mapsto}p} ~\equivA~ \code{x:p} \\[0.4em]
\text{for a non-primitive heap } \code{T} \text{ (e.g.\ } \code{Ref(\ldots)}\text{):} & \code{x{\mapsto}T} ~\land~ \code{x:T} ~\equivA~ \code{\false}
\end{array}
\]

\paragraph{Must-aliasing: the swap example.} We can now give the case specification for swap:
\begin{lstlisting}[numbers=none]
swap:(*@$\forall$@*)T:(*@$\AnyP$@*),A,B:(*@$\topI$@*). case [x,y] {
   x:Ref(T) (*@$\land$@*) y:Ref(T) (*@${\caseRA}$@*) ens[r] r:();
   x(*@$\mapsto$@*)Ref(A) * y(*@$\mapsto$@*)Ref(B) (*@${\caseRA}$@*) ens[r] x(*@$\mapsto$@*)Ref(B) * y(*@$\mapsto$@*)Ref(A) (*@$\land$@*) r:();
   x(*@$\mapsto$@*)Ref(A) (*@$\land$@*) y:{x} (*@${\caseRA}$@*) ens[r] x(*@$\mapsto$@*)Ref(A) (*@$\land$@*) r:() }
\end{lstlisting}

The first case handles flow-insensitive pure-type access; the second, two distinct locations (disjoint heap ownership) whose values are swapped. The third uses the singleton type \code{y{:}\{x\}} to express that y
must-alias x (i.e. \code{y{=}x}). Encoding the alias as a type rather than a separate equation keeps
must-alias reasoning within the type language, and swapping is then a no-op preserving the separation type. This case is
inexpressible with uniqueness types, confirming separation types are strictly more expressive (see Appendix~\ref{sec:unique}).


\begin{remark}[Singleton types and aliasing]
  A singleton type \code{\{z\}} denotes the one-element
set \code{\{z\}} --- the most precise type for a value, naming its exact witness. A singleton
is itself pure (\code{\{z\} <: Any_P}) --- value equality needs no ownership, so a bare
singleton never claims the heap; to combine the witness with ownership, it appears as an
owned cell's exact content, \code{\heapto{x}{Ref(\{z\})}}. Aliasing consequences follow
the assertion's precision: (i) \code{\heapto{x}{Ref(\{z\})}} --- full ownership, exact
content named, no unknown aliases, type mutation sound; (ii) \code{x{:}\{z\}} --- witness
known but pure heap, unknown aliases may mutate independently, type mutation unsound;
(iii) \code{x{:}A} with \code{A <: Any_P} --- no witness named, arbitrary unknown aliases,
type mutation unsound. Yet (ii) is strictly more informative than (iii) --- from
\code{x{:}\{z\}} the verifier knows \code{x{=}z}, supporting equality and must-alias
reasoning.
\end{remark}

\begin{remark}[Mixed scenarios via weakening]
  The spec omits mixed cases where one parameter is a pure type and
  the other a separation type, specifically: \code{x{:}Ref(A) \wedge \heapto{y}{Ref(B)}} and \code{\heapto{x}{Ref(A)} \wedge y{:}Ref(B)}.
  A caller willing to surrender ownership can apply
the weakening lemma to reduce to Case 1; otherwise the implicit otherwise clause applies,
yielding \code{Abrt}.

\end{remark}

\subsection{A Disciplined Error Hierarchy}
\label{sec:error}

We begin with a concrete example that motivates our error hierarchy before presenting it in full. Consider the \code{head}
function on lists:

\begin{minipage}{.8\textwidth}
\begin{lstlisting}[numbers=none]
  head : (*@$\forall$@*)T. List(T) (*@$\rightarrow$@*) T
\end{lstlisting}
\end{minipage}

Two distinct kinds of failure can arise. If \code{head} is called with the wrong argument type (say an \code{Int}), the type
checker can detect this statically and flag an \code{\Abrt} error: a distinguished value representing
compile-time type-errors that must never appear in the postcondition
of a well-typed program. The program is not well-typed and must not be allowed to run.

If \code{head} is called with a \code{Nil} value (the correct type, but an empty list), a runtime error occurs: the program has
gone wrong in a way that was not caught at compile time, and the result is best represented by an \code{Err} value that
propagates through subsequent computations.

Traditional type systems (e.g. OCaml, Scala) conflate runtime errors with the bottom type \code{\bot},
treating them as outside the type system entirely. This conflation loses precision: \code{\bot} denotes
non-termination or unreachability: a function with postcondition \code{r{:}\bot} is one that never
returns. A runtime error, by contrast, is a propagatable value that flows into subsequent
computations. By internalising runtime errors as the first-class value \code{\Err} in our type lattice (Figure~\ref{fig:lattice}),
our framework can track where errors arise and propagate, something \code{\bot}-based
treatments cannot express.

\paragraph{Three type specifications for \code{head}.}
Our framework allows the programmer to choose the level of error
specification appropriate for their context:


\noindent
\begin{minipage}{.98\textwidth}
\begin{lstlisting}[numbers=none]
  head : (*@$\forall$@*)T. req[x] x:List(T) ens[r] r:T(*@${\vee}$@*)Err
  head : (*@$\forall$@*)T. case[x] { x:Cons(T,_) (*@${\caseRA}$@*) ens[r] r:T;
                       x:Nil       (*@${\caseRA}$@*) ens[r] r:Err }
  head : (*@$\forall$@*)T. req[x] x:Cons(T,_) ens[r] r:T
\end{lstlisting}
\end{minipage}

\begin{wrapfigure}{r}{0.42\textwidth}
  \begin{tikzpicture}[x=1.8cm,y=1.0cm]
\node at (0,-1.5)        (Bot)       {$\bot$};

\node at (-0.5,-0.5)     (botSys)    {\botsys};

\node at (-0.5,0.5)      (UnExc)     {\UnExc};
\node at (0.5,0.5)       (ChExc)     {\ChExc};
\node at (0.5,2.4)       (Exc)       {\Exc};

\node at (1.5,1.5)       (Abrt)      {\Abrt};
\node at (-1.5,1.5)      (Err)       {\Err};

\node at (-1.0,2.4)      (AnyP)      {$\AnyP$};
\node at (-0.5,2.4)      (AnyS)      {$\AnyS$};
\node at (-0.5,3.3)      (Any)       {\Any};

\node at (0,4.2)         (Top)       {$\top$};

\draw[->] (Bot)  --  (botSys);
\draw[->] (Bot)  --  (Abrt);
\draw[->] (botSys)  --  (UnExc);
\draw[->] (botSys)  --  (ChExc);
\draw[->] (botSys)  --  (Err);
\draw[->] (UnExc)  --  (Exc);
\draw[->] (UnExc)  --  (AnyP);
\draw[->] (UnExc)  --  (AnyS);
\draw[->] (ChExc)  --  (Exc);
\draw[->] (Err)  --  (AnyP);
\draw[->] (AnyP)  --  (Any);
\draw[->] (AnyS)  --  (Any);
\draw[->] (Any)  --  (Top);
\draw[->] (Exc)  --  (Top);
\draw[->] (Abrt)  --  (Top);
\end{tikzpicture}

  \caption{Lattice of Types}
  \label{fig:lattice}
\end{wrapfigure}
The first two specifications treat \code{head(Nil)} as a tracked runtime error; the third eliminates it
as a compile-time error. The first and second are comparable: the second (case form) is strictly
\emph{more precise} than the first, since on a \code{Cons} input it guarantees the result is a \code{T}
(never \code{Err}), whereas the first permits \code{Err} on every input. The second and third, by
contrast, are not comparable in terms of strength: they make different trade-offs between
permissiveness and static guarantees (the case form tolerates \code{Nil} as a recoverable \code{Err},
while the third rules \code{Nil} out as a compile-time error, and \code{Err} and \code{Abrt} are
unrelated). Together, the three
specifications illustrate the range of expressiveness available in our framework, a range that
eludes most traditional type systems without dependent types or type qualifiers.

\paragraph{Error hierarchy.}
Just above the empty type \code{\bot}, the node \botsys{} denotes runtime system
errors (such as out-of-memory or hardware faults) that lie beyond the control of our type
framework, and which we accordingly set aside in the discussion that follows. The three error types
we design are:
\begin{itemize}[leftmargin=2em]
\item \code{\Err}: a pure error value that may be passed as an argument to method calls. Subtypes include \code{unDef} (for
uninitialised values) and \code{null} (for null pointer dereferences). These are tolerated in programs and do not trigger
compile-time errors until they are actually used.
\item \code{\Exc}: an exception, further split into \code{ChExc} (checked; must be tracked explicitly by our type system) and \code{UnExc}
(unchecked; treated as a subtype of both \code{\AnyS} and \code{\AnyP} and may be omitted from specifications). Exception handling is
discussed in Appendix~\ref{sec:exception}.
\item \code{\Abrt} (and \code{\top}): a compile-time type error. Our type-safety verification framework must prove \code{\Abrt} and \code{\top}
unreachable for any well-typed program.
\end{itemize}

Consider \code{map}:

\begin{lstlisting}[numbers=none]
  map(f,xs) = match xs of { Nil -> Nil;
                            Cons(y,ys) -> Cons(f(y),map(f,ys)) }
\end{lstlisting}

Its type specification with case specs is:

\begin{lstlisting}[numbers=none]
  map: (*@$\forall$@*)A,B: (*@${\AnyP}$@*). case[f,xs] {
        f:Any (*@$\land$@*) xs:Nil (*@${\caseRA}$@*) ens[r] r:Nil;
        f:A(*@$\ra$@*)B (*@$\land$@*) xs:(*@\gtype{Cons}{A,\gtype{List}{A}}@*) (*@${\caseRA}$@*) ens[r] r:(*@\gtype{Cons}{B,\gtype{List}{B}}@*);
        _ (*@${\caseRA}$@*) ens[r] r:(*@$\top$@*) }
\end{lstlisting}

The first two cases capture valid type specifications; the otherwise clause \code{\_} triggers a compile-time error if the
inputs f and xs do not match either case. Here \code{\_} is
\emph{syntactic sugar} for the negation of all preceding guards.
Explicitly specifying \code{\top} or \code{\Abrt} flags compile-time type errors.
The disjointness of case guards  (Sec.~\ref{sec:semsubtype})
is essential
here: an input satisfying multiple guards would induce an
intersection type, and without
disjointness the verifier would need to perform a proof search to determine which case applies.
Enforced disjointness eliminates this need for proof search.

\paragraph{An \code{\Err} tolerated, then captured as \code{\Abrt}: reading an uninitialised value.}
An \code{\Err} is \emph{tolerated} (it may flow through the program freely) and becomes a
captured \code{\Abrt} only where a function demands its validity. Consider an uninitialised memory
slot, whose default contents are the zero word \code{0x0}. We type this as an \code{\Err} (its
\code{unDef} subtype), disjoint from every valid type:
\begin{lstlisting}[numbers=none]
 0x0 : Err
\end{lstlisting}
As an \code{\Err}, \code{0x0} is tolerated: bit-for-bit an ordinary word, it may be bound, copied, and
passed around freely. An \code{\Err} faults only where validity is demanded, i.e.\ when it reaches
a function whose precondition requires a valid argument, such as \code{use}:
\begin{lstlisting}[numbers=none]
 use : req[v] v:Valid ens[(*@\_@*)] ()
\end{lstlisting}
Here $\code{Valid} \defeq \code{Any} \land \neg\code{Err}$ denotes the valid (non-error) types.
The call \code{use(0x0)} cannot satisfy \code{v{:}Valid}, so it is ill-typed and \emph{captured} as
\code{\Abrt}, ruled out by our well-typedness requirement that well-typed programs never abort. The same \code{0x0} is thus a harmless \code{\Err}
while in flight, but an \code{\Abrt} the instant it reaches \code{use}; keeping the two distinct is
what rules the bug out by typing.

\begin{remark}[The \code{\bot} type and non-termination]
Following standard convention, \code{\bot} denotes the empty
set of values and is a subtype of every type.
As a postcondition, \code{r{:}\bot} asserts that
no value is ever returned -- the computation diverges. As a precondition,
\code{x{:}\bot} asserts that the
input never exist, that is unreachable \code{\false}. Crucially, in our framework \code{\bot} is reserved for
these two standard roles. \code{\Err} is not conflated with \code{\bot}, as they are internalised as 
first-class values of the type \code{\Err}, enabling precise error tracking that \code{\bot}-based treatments cannot provide.
The fatal-error type \code{\Abrt} is likewise distinct from \code{\bot}, and in a different way.
\code{\bot} is \emph{uninhabited} and is a legitimate specification outcome (divergence or an
unreachable input). By contrast, \code{\Abrt} is an \emph{inhabited} distinguished value marking a
compile-time type error; it is never a desired outcome and must be proven unreachable in any
well-typed program. Since \code{\bot} is a subtype of every type, we do
formally have \code{\bot <: \Abrt}; but this does not weaken the guarantee, because \code{\bot} is
uninhabited and so contributes no actual value to \code{\Abrt}. As
termination is undecidable, the guarantee that well-typed programs never abort is accordingly a
partial-correctness one: it ensures that whenever a well-typed program does produce an outcome, that
outcome is never \code{\Abrt}. It does not account for potential non-termination.
For example, \code{(req[x]~x:Any~ens[r]~r{:}\bot)} is the correct specification of a diverging function that never returns.



\end{remark}



\begin{remark}[\code{UnExc} and \code{\bot}]
\code{UnExc} should not be conflated with \code{\bot}, since \code{UnExc} is inhabited (it
classifies untracked but raised exceptions), whereas \code{\bot} has no inhabitants and denotes non-termination.
\end{remark}

\subsection{Pure and Separation Type Predicates}
\label{sec:GADT}
\label{sec:pred}
Modern functional languages, such as ML and Haskell,
support algebraic data types (ADTs). Four simple examples are shown
below.

\begin{lstlisting}[name=functor,numbers=none]
 data Color   = Red  | Black
 data Nat     = Zero | Succ(Nat)
 data List(T) = Nil  | Cons(T,List(T))
 data Tree(T) = Leaf | Node(Color,T,Tree(T),Tree(T))
\end{lstlisting}

We propose using type predicates to model
algebraic data types.
Like types, type
predicates can be recursive, but are strictly more expressive.
Each ADT can
be encoded as a pure type predicate, as shown below.
\begin{lstlisting}[name=functor,numbers=none]
 pred x:List(T) = x:Nil  (*@$\vee$@*) (*@$\exists$@*) r. x:Cons(T,r)(*@$\wedge$@*)r:List(T)
 pred n:Nat     = n:Zero (*@$\vee$@*) (*@$\exists$@*) r. n:Succ(r)(*@$\wedge$@*)r:Nat
 pred x:Color   = x:Red  (*@$\vee$@*) x:Black
 pred x:Tree(T) = x:Leaf (*@$\vee$@*) (*@$\exists$@*) c,lt,rt. x:Node(c,T,lt,rt)
                     (*@$\wedge$@*)c:Color(*@$\wedge$@*)lt:Tree(T)(*@$\wedge$@*)rt:Tree(T)
\end{lstlisting}

Type predicates are interpreted inductively as least fixed points; their recursive unfolding is
assumed to terminate (well-foundedness), a condition stated precisely in Remark~\ref{remark:well-foundness} at the
end of this section.
(We assume here a strict language. If a lazy language is used,
we will need to make use of an interpretation via greatest fixed points.)

To capture data structures with stronger invariant properties,
Haskell and OCaml
also support
a restricted form of dependent types known as Generalised
Algebraic Data Types (GADTs).
Here, their type parameters may be separately
instantiated, depending on the data constructors used,
as illustrated by
the red/black balanced tree below.
\begin{lstlisting}[name=functor,numbers=none]
data RBTree :: Color (*@\ra@*) Nat (*@\ra@*) * (*@\ra@*) * where
  Leaf      :: RBTree Black Zero a
  RedNode   :: RBTree Black h a (*@\ra@*) a
               (*@\ra@*) RBTree Black h a (*@\ra@*) RBTree Red h a
  BlackNode :: RBTree cl h a (*@\ra@*) a
               (*@\ra@*) RBTree cr h a (*@\ra@*) RBTree Black (Succ h) a
\end{lstlisting}

Rather than adopting GADT data types, we propose that
our framework use type predicates
for capturing data structures with
strong
invariant properties.
Two examples
are shown below.

\begin{lstlisting}[name=functor,numbers=none]
 pred x:List(T,s) = x:Nil (*@$\wedge$@*) s:Zero
       (*@$\vee$@*) (*@$\exists$@*) r,m. x:Cons(T,r)(*@$\wedge$@*)r:List(T,m)(*@$\wedge$@*)s:Succ(m)
 pred x:RBTree(T,c,h) = x:Leaf(*@$\wedge$@*)c:Black(*@$\wedge$@*)h:Zero
       (*@$\vee$@*) (*@$\exists$@*) h1,lt,rt. x:Node(c,T,lt,rt)(*@$\wedge$@*)lt:RBTree(T,_,h1)
                    (*@$\wedge$@*)rt:RBTree(T,_,h1)(*@$\wedge$@*)c:Black(*@$\wedge$@*)h:Succ(h1)
       (*@$\vee$@*) (*@$\exists$@*) r,lt,rt. x:Node(c,T,lt,rt)(*@$\wedge$@*)lt:RBTree(T,Black,h)
                    (*@$\wedge$@*)rt:RBTree(T,Black,h)(*@$\wedge$@*)c:Red
\end{lstlisting}

The predicate \code{x{:}List(T,s)} uses an extra parameter \code{s}
to track the length of the list. The predicate \code{x{:}RBTree(T,c,h)} uses \code{c} to denote the colour of the red/black tree and \code{h} to capture its
black-height, and enforces that
the black-height is always balanced. These richer predicates are related to their base-type counterparts
by the following lemmas.

\begin{lstlisting}[name=functor,numbers=none]
 lemma (*@\code{\forall x,T.~(\exists n.}@*) x:List(T,n))   (*@\equivArr @*) x:List(T)
 lemma (*@\code{\forall x,T,c,h .}@*) x:RBTree(T,c,h) (*@\implyArr @*) x:Tree(T)
\end{lstlisting}

\paragraph{Stable pure and separation types.}
With length-indexed predicates in hand, we now define
stability for pure types (Definition~\ref{def:pure-type-pred}) and separation types (Definition~\ref{def:sep-pred}). Stability of type
predicates themselves is deferred to Sec.~\ref{sec:type-logic}, where the necessary
formal vocabulary is introduced.

\begin{defn}[Stable Pure Type]
\label{def:s-pure-type}
A pure type \code{t} is stable under heap mutation if it
cannot be invalidated by
any type-consistent mutation (Definition~\ref{def:type-consistent}). Formally, for any heaps \code{h, h'} related by a type-consistent
mutation with respect to the current pre-state \code{\ctr}, \code{s,h  {\models} x:t} implies \code{s,h' {\models} x:t}.
Singleton types \code{y{:}\{x\}} (asserting \code{y{=}x}) \cite{Lionel:OOPSLA22,Castagna:POPL24
} are trivially stable:
they live in the pure
constraint \code{\pure} and express only an equality between program or logical variables (including
existentially quantified \code{\exists v^*} in \code{\ctr}). This fact is independent of the heap and therefore
unaffected by any heap mutation.
\end{defn}

Pure data types in strongly typed languages, like OCaml, are
designed to be stable by default. Moreover, GADTs are also stable with
the use of immutable constructors.

\begin{defn}[Stable Separation Type]
A separation type  \code{\heapto{x}{T}} is stable if the content type
\code{T} is a stable pure type.
Full ownership of x is guaranteed by the separation type discipline
itself, so the non-trivial requirement is on the content type T. For example, \code{\heapto{x}{Ref(Int{\vee}Str)}}
is stable because \code{Int{\vee}Str} is stable. By contrast, \code{\heapto{x}{Ref(T_2)}}  may not be stable if \code{T_2} can be invalidated
by aliased writes to its fields. Unlike a stable pure type, a stable separation type additionally
permits type mutation of its content via an explicit update. This is safe because full ownership ensures no
pure-heap alias to x exists and every must-alias observes the updated type consistently.
\end{defn}

It is also possible to support richer data structures via separation type predicates which may
capture heap nodes exclusively owned by the predicates, such as \code{\heapto{x}{MList(T)}}:
\begin{lstlisting}[name=functor,numbers=none]
 data MList(T)  = Nil  | Cons(T,Ref(MList(T)))
 pred (*@\code{\heapto{x}{MList(T)}}@*) = x:Nil (*@$\vee$@*) (*@$\exists~r,r_1$@*) . (*@\code{\heapto{x}{Cons(T,r)}{\sep}\heapto{r}{Ref(r_1)}{\sep}\heapto{r_1}{MList(T)}}@*)
 \end{lstlisting}

Such mutable lists can be used to construct fully owned mutable structures,
including circular and doubly-linked lists, a domain well covered
by two decades of research in separation logic.
In particular, a pure type \code{x{:}MList(\_)} cannot be used to guarantee the construction of such a
circular list: being freely aliasable and owning no cell, it cannot pin down the back-reference that
closes the cycle. This guarantee is attainable only with separation types and predicates, whose
exclusive ownership of the nodes lets the predicate tie the final tail back to an owned cell.
We propose to make these features available within the type system to
support greater memory safety.




Precise definitions of stable pure and separation type predicates, in terms of the formal type
logic, are deferred to Definitions~\ref{def:pure-type-pred} and \ref{def:sep-pred} in Sec.~\ref{sec:type-logic}, once the required vocabulary --
immutable constructors, heap assertions, ownership -- has been formally established.

The well-foundedness assumption underlying all type predicate definitions is stated formally as
Remark~\ref{remark:well-foundness} in Sec.~\ref{sec:type-logic}, once the formal vocabulary is in place.


A more detailed comparison of type predicates with
GADTs and liquid types is in
Section~\ref{sec:compare}.

\subsection{A Combined Example}
\label{sec:finale}
\label{sec:combined}
We close with a single example exercising all four ingredients at once. Consider \code{rm\_head}, which
dereferences a heap-allocated list reference, returns the head if non-empty, and \code{Err} otherwise:

\begin{lstlisting}[numbers=none]
 rm_head(m) = let xs = !m in
    match xs of { Nil -> err ;
            Cons(y,_) -> update(m, tail(xs)); y }
\end{lstlisting}
We use the length-indexed predicate \code{List(T,s)}, where \code{s} is a \code{Nat} tracking the length, to give a precise
three-case specification. \code{s{:}Zero} means the list is empty; \code{s{:}Succ(s_1)} means it is non-empty.
\begin{lstlisting}[numbers=none]
 rm_head : (*@$\forall$@*) T,l,s,s(*@\textsubscript{1}@*).
   case[m] {
     (* Case 1: pure-type access, possible runtime Err *)
     m:Ref(l) (*@$\land$@*) l:List(T) (*@${\caseRA}$@*) ens[r] r:T (*@$\vee$@*) Err;
     (* Case 2: sep type, empty list, rm_head returns Err *)
     m(*@$\mapsto$@*)Ref(l) (*@$\land$@*) l:List(T,s) (*@$\land$@*) s:Zero (*@${\caseRA}$@*) ens[r] m(*@$\mapsto$@*)Ref(l) (*@$\land$@*) r:Err;
     (* Case 3: sep type, non-empty list, flow-sensitive mutation *)
     m(*@$\mapsto$@*)Ref(l) (*@$\land$@*) l:List(T,s) (*@$\land$@*) s:Succ(s(*@\textsubscript{1}@*)) (*@${\caseRA}$@*) ens[r] (*@$\exists$@*) l(*@\textsubscript{1}@*). m(*@$\mapsto$@*)Ref(l(*@\textsubscript{1}@*)) (*@$\land$@*) l(*@\textsubscript{1}@*):List(T,s(*@\textsubscript{1}@*)) (*@$\land$@*) r:T }
\end{lstlisting}
This specification uses all four ingredients:
\begin{itemize}
\item Case specifications (Sec.~\ref{sec:semsubtype}): three disjoint guards: pure-type access, and the empty-
and non-empty-list separation cases.
\item  Separation types (Sec.~\ref{sec:spatial}): Cases 2 and 3 use \code{\heapto{m}{Ref(l)}} to own the heap location; Case 3's
postcondition reflects the flow-sensitive \code{update}, with \code{m} now pointing to the tail \code{l1}.
\item  Error typing (Sec.~\ref{sec:error}): Case 1 tolerates a runtime \code{Err} via \code{r:T \vee \Err}; Case 2 returns \code{Err} precisely
when the list is known empty, not conflating it with \code{\bot} or an unchecked exception.
\item  Type predicates (Sec.~\ref{sec:GADT}): the length-indexed \code{List(T,s)} uses its \code{Nat} index \code{s} to distinguish
\code{Zero} (empty) from \coden{Succ(s_1)} (non-empty), letting the verifier prove Cases 2 and 3 disjoint and track
the length decrease in Case 3.
\end{itemize}

\paragraph{Stability of type predicates.} The example also pinpoints where pure-type stability holds and breaks
down. \code{l{:}List(T,s)} is stable: \code{List} is built from immutable \code{Cons} constructors, so the length index
\code{s}, fixed at construction, is observed identically by every alias. By contrast, \code{m{:}Ref(List(T,s))} is not:
any alias can update the \code{Ref} to a list of different length \code{s'}, invalidating \code{s}. This is why Case 1
uses the weaker \code{m{:}Ref(List(T))}, the length being untrackable through a pure \code{Ref} under aliased
mutation. Cases 2 and 3 resolve this by taking full ownership via \code{\heapto{m}{Ref(l)} \wedge l{:}List(T,s)}: with no
pure-heap alias to \code{m}, the cell is fully owned, permitting the precise update in Case 3's postcondition.
\toolname~checks the three cases independently, the separation type being consumed and re-emitted by the Call
rule (Sec.~\ref{sec:Hoare}) so no ownership is lost across \code{update}.










\section{Core Language and a Logic for Types}
\label{sec:lang}
This section has two parts. The first (Sec.~\ref{sec:core}) presents the syntax of the core language: a
small strict higher-order functional language with immutable constructors and a single mutable
heap type \code{Ref}. The second (Sec.~\ref{sec:type-logic}) defines the type logic (the grammar of states, heap
assertions, pure constraints, and type specifications) that Sec.~\ref{sec:Hoare} uses to state and verify
type-safety properties. The notation req/ens used informally throughout Sec.~\ref{sec:motivate} is given its
formal definition here. The big-step operational semantics of the core language, against which our
soundness result (Sec.~\ref{sec:Soundness}) is stated, is given in \cref{sec:bigstep}.

\subsection{Core Language}
\label{sec:core}
We consider a strict higher-order functional language with
immutable data constructors
with the exception of the \code{Ref} type.
The key design choices are as follows.
To support both
static and dynamic typing, we organise data constructors and primitive types into a subtype
lattice with \code{Any} as the supertype of all valid types, covering both user-defined and primitive
types. Each \code{let} binding introduces an immutable variable. Functions and data constructors are
always fully applied. Function definitions and lambda abstractions are always given explicit type
specifications.
We also provide an alternative style for function specifications,
\code{\vfunrel{x^*\,r}{\dspec[r]}}, so that \code{case} and \code{\ens} constructs
do not need to explicitly track the parameters/result.
For simplicity, we omit \code{inout} parameters
and exception-handling constructs.

The syntax is summarised in Figure~\ref{fig:core_lang_syntax}.
We briefly explain each syntactic category. A program \ensuremath{\m{({Prog})}}
is a sequence of definitions. Each definition \ensuremath{\m{({Defn})}} binds a function name to a body expression
together with a type specification~\code{\vfunrel{x^*\,r}{\dspec}}. The full expression language \ensuremath{\m{({Full})}} includes the
standard constructs (values, \code{let}-bindings, constructor applications, function calls, type casts
\code{(t)e}, and \code{match}) and is the language the programmer writes. The core expression language
\ensuremath{\m{({Core})}} is the normalised form used for Hoare-style reasoning: constructor arguments and
function arguments are always simple variables, and the \code{match} scrutinee is always a variable.
Base types \ensuremath{\m{({Base})}} include singleton types, primitive types, constructor types, predicate types, and the special
types \code{\Err}, \code{\Abrt}, \code{Any}, \code{\bot}, and \code{\top}. Types \ensuremath{\m{({Type})}} extend base types with separation types \code{(\heapto{}{base})},
function types, and Boolean combinations \code{\wedge}, \code{\vee}, and \code{\neg}. The separation type \code{\heapto{}{base}} is a type in
its own right. A variable \code{x} of this type would be written \code{x{:}\heapto{}{base}}, but for simplicity we omit
the colon for separation types and write \code{\heapto{x}{base}}. Patterns \ensuremath{\m{({Pattern})}} mirror types and are used
in \code{match}. Note that \code{\top} and \code{\bot} are excluded from patterns since the former trivially holds and
testing for the latter is equivalent to the halting problem.

{
\begin{figure}[h]
\centering
    {\small
 $
      \begin{array}{lrcl}


    \m{({Prog})} & \code{prog} & ::= & \code{def_1}~;\cdots~;\code{def_n}
   \\ [0.1em]
     \m{({Defn})} & \code{def} & ::= & \code{f}~\code{x_1} {\cdots} \code{x_n} = \code{e}~\code{where}~
     \code{f}{::}\vfunrel{\code{x_1}{\cdots}\code{x_n}\,\code{r}}{\dspec[\code{r}]}
   \\ [0.1em]
     \m{({Full})} & \code{e} & ::= &  \code{val}
     \vbar \elet{x{:}t}{e_1}{e_2}
      \vbar \code{C}(\code{e_1}{\ldots}\code{e_n}) 
      \vbar (\code{f}~\code{e_1}{\cdots}\code{e_n})
      \vbar (\code{t})~\code{e} \\
      &&&
      \vbar \ematch{e}{p_1{\ra}e_1{{\mSep}\cdots{\mSep}}p_n{\ra}e_n}
   \\[0.1em]
     \m{({Value})} & \code{val} & ::= &
           \code{x}
      \vbar \code{c}
      \vbar \err
      \vbar \abrt
       \vbar \efun{\code{x^*}}{\code{e}}::\vfunrel{\code{x^*}\,\code{r}}{\dspec[\code{r}]} 
\\[0.1em]
     \m{({Core})} & \code{e} & ::= &
      \code{val}
     \vbar \elet{x}{e_1}{e_2}
      \vbar  \code{C}(\code{x_1}{\ldots}\code{x_n}) 
      \ \vbar (\code{f}~\code{x_1}{\cdots}\code{x_n})
      \vbar (\code{t})~\code{x} \\[0.1em]
      & & &
      \vbar \ematch{x}{p_1{\ra}e_1{{\mSep}\cdots{\mSep}}p_n{\ra}e_n}
   \\[0.1em]
    \m{({Base})} & \code{base} & ::= &
      \{\code{x}\}
      \vbar \{\code{c}\}
      \vbar \code{Bool}
      \vbar \code{Int}
      \vbar \code{Str}
      \vbar \code{Ref(\cdots)}
      \vbar \cdots
      \vbar \Err
      \vbar \Abrt
      \vbar \Any
      \vbar \bot
      \vbar \top\\
         & & &
      \vbar \code{C}(\code{t_1}{\ldots}\code{t_n})

      \phantom{c}      \vbar \code{pred}(\code{t_1}{\ldots}\code{t_n})
   \\ [0.1em]
   \m{({Type})} & \code{t} & ::= &
      \code{base}
      \vbar \heapto{}{\code{base}}
      \vbar \code{t_1} {\ra} \code{t_2}
       \vbar \code{t_1}{\wedge}\code{t_2}
      \vbar \code{t_1}{\vee}\code{t_2}
      \vbar \neg \code{t}
   \\ [0.1em]
     \m{({Pattern})} & \code{p} & ::= &
      \code{c}
      \vbar \code{Bool}
      \vbar \code{Int}
      \vbar \code{Str}
      \vbar \cdots
      \vbar \Err
      \vbar \Abrt 
      \vbar \Any
      \vbar \code{C}(\code{x_1}{\ldots}\code{x_n})
      \phantom{c}      \vbar \code{pred}(\_) \\
      &&&
      \vbar \_ {\ra} \_
       \vbar \code{p_1}{\wedge}\code{p_2}
      \vbar \code{p_1}{\vee}\code{p_2}
      \vbar \neg \code{p}
      \vbar \_
\\

  \end{array}$
$\begin{array}{lrcl}
\multicolumn{4}{c}{
\  \code{ \code{c} \in  \mathbb{B} \cup \mathbb{Z}
  \cup \cdots }
           \qquad\quad\   \qquad\quad\
\code{x, f, r \in \emph{vars}}

}\\
      \end{array}
      $}
 \caption{Syntax of Full and Core Language with Types}
 \label{fig:core_lang_syntax}
 \label{fig:type_syntax}
\end{figure}
}

The Full{\code{\equivArr}}Core preprocessing works by inserting eager casts to \code{\Any} at every constructor
argument, function argument, and match scrutinee.
\code{\Any} is the supertype of all normal runtime
types but does not include \code{\Abrt} or \code{ChExc}. The casts are the propagation points. If any
sub-expression produces \code{\Abrt} or \code{ChExc}, the cast short-circuits and propagates the exceptional
value outward, so the core \code{let} construct never needs to handle these cases directly. The
preprocessing rules are:
\begin{small}
\[
\begin{array}{rl}
\code{C(e_1,\cdots,e_n)} & \equivArr \eletBlk{\{x_i{:}\Any{=}e_i\}_{i=1}^n}{C(x_1,\cdots,x_n)}  \\
\code{f~e_1\,{\cdots}\,e_n} & \equivArr \eletBlk{\{x_i{:}\Any{=}e_i\}_{i=1}^n}{f~x_1\,{\cdots}\,x_n}  \\
\ematch{e}{\ldots} & \equivArr \elet{x{:}\Any}{e}{\ematch{x}{\ldots}} \\
\elet{x{:}t}{e_1}{e_2} & \equivArr \code{let}\ \code{{x}={(t)~e_1}~in} 
     ~{\ematch{x}{\Abrt{\ra}x{\vbar}\_\,{\ra}e_2}}  \\
\code{(t)~e} & \equivArr \elet{x}{e}{(\elet{x_1}{(t)~x}{x_1})}  \\
\end{array}
\]
\end{small}

\paragraph{Pattern overlap and desugaring.} For ease of programming, we allow the patterns in each match
construct to overlap, relying on top-down processing to give an unambiguous semantics. Before
Hoare-style reasoning begins, overlapping patterns are desugared into disjoint form. For
example, \code{\ematch{x}{p_1{\ra}e_1{\vbar}p_2{\ra}e_2{\vbar}p_3{\ra}e_3}}
becomes
\code{\ematch{x}{p_1{\ra}e_1{\vbar}p_2{\wedge}\neg(p_1){\ra}e_2{\vbar}p_3{\wedge}\neg(p_1{\vee}p_2){\ra}e_3}}.
The Hoare rules can then deal with complex patterns formed from \code{\wedge}, \code{\vee}, and \code{\neg}.
To support efficient runtime execution, complex type patterns are further translated to simpler
patterns using the rules below:

\begin{small}
\[
\begin{array}{ll}
  \ematch{x}{\ldots{\vbar}p_1{\vee}p_2~{\ra}~e{\vbar}..}
& \equivArr
\ematch{x}{\ldots{\vbar}p_1~{\ra}~e{\vbar}p_2~{\ra}~e{\vbar}..}
 \\
\ematch{x}{\ldots{\vbar}p_1{\wedge}p_2{\ra}e{\vbar}..} 
& \equivArr
\ematch{x}{\ldots{\vbar}p_1{\ra}
 \ematch{x}{p_2~{\ra}~e{\vbar}..}{\vbar}..}
 \\
\ematch{x}{\ldots{\vbar}\neg\,p~{\ra}~e{\vbar}..} 
& \equivArr
\ematch{x}{\ldots{\vbar}p~{\ra}~\ematch{x}{..}{\vbar}\_~{\ra}~e}
\end{array}
\]
\end{small}

\subsection{Type Logic and Specifications}
\label{sec:type-logic}
We now define the type logic that underpins the Hoare rules of Sec.~\ref{sec:Hoare}. The central object is a
type state \code{\ctr}, which pairs a heap assertion \code{\heap} capturing ownership of heap locations with a pure
constraint \code{\pure} capturing type and equality facts about variables. The grammar is given in Figure~\ref{fig:type_spec}.

{
\begin{figure}[h]
\centering
 $
  \begin{array}{lrcl}


    \m{({State})} & \ctr & ::= &
      \bigvee_{i=1}^n~\exists v^* {\cdot} \heap_i{\wedge}\pure_i
       \\
    \m{({Heap})} & \heap & ::= &
      \emp
      \vbar \heapto{x}{t}
      \vbar \heap_1{\sep}\heap_2 \\
    \m{({Pure})} & \pure & ::= &
      x{:}t \vbar x_1{=}x_2
      \vbar \pure_1{\wedge}\pure_2
      \vbar \pure_1{\vee}\pure_2
      \vbar \neg \pure
      \vbar \exists v^* {\cdot} \pure
   \\ [0.1em]
     \m{({Spec})} & \dspec & ::= &
      \stageensres{r}{\ctr}
      \vbar \specCase{\ctr_1{\caseRA}\dspec_1;{\cdots};\ctr_n{\caseRA}\dspec_n} 
      \vbar \forall x^* {\cdot} \dspec
   \\ [0.1em]
\end{array}$
 \caption{Type Logic Specification}
 \label{fig:type_spec}
\end{figure}
}

Each syntactic category plays a distinct role. A state \code{\ctr}
is a disjunction of existentially quantified
heap-pure pairs; disjunction arises naturally from path-sensitive case analysis. A heap assertion
\code{\heap} is either empty (\code{emp}), a separation type \code{\heapto{x}{t}} asserting full ownership of the location at \code{x} with
content type \code{t}, or a separating conjunction \code{\heap_1 \sep \heap_2} asserting disjoint ownership. A pure
constraint \code{\pure} records type memberships \code{x{:}t}, equalities, and Boolean combinations thereof; it
does not assert ownership. Note that we support \code{\vee}, \code{\wedge}, and \code{\neg} in both the type and assertion
languages, with the union-type equivalence \code{r{:}T_1{\vee}T_2 \equivArrB r{:}T_1 {\vee} r{:}T_2} holding as a derived rule.
A specification \code{\dspec} is either a simple postcondition \code{ens[r]~\ctr}, a case
specification with pairwise-disjoint guards
\code{(\forall i{\neq}j.\,\ctr_i{\wedge}\ctr_j{=}\false)},
or a universally quantified
specification. The shorthand \code{req~\ctr_1\quad ens[r]~\ctr_2} abbreviates
\code{case\,\{\ctr_1\,{\implyArrB}\,ens[r]\,\ctr_2\}}, formalising the
notation used throughout Sec.~\ref{sec:motivate}.

For completeness, every case specification is implicitly extended with an
otherwise clause
\code{(\_\,{\implyArrB}\,ens[r]~r{:}\top)}. This clause flags any input not covered by the stated cases
as a compile-time \code{Abrt} error. Its postcondition \code{r{:}\top} signals this because \code{\top}, unlike \code{Any}, admits \code{\Abrt} values. This is
consistent with the disjointness discipline of Sec.~\ref{sec:semsubtype}: because the stated guards are pairwise
disjoint, the otherwise clause is well-defined and covers precisely the remaining inputs.

Figure~\ref{fig:type_spec} establishes the vocabulary we now need: immutable constructors (types built by
the Base grammar that do not involve \code{\heapto{}{}}), pure sub-components (occurrences of \code{\pure} in a state),
and ownership (heap assertions \code{\heap}). With these in hand, we can give the formal definitions of stable type
predicates that were deferred from Sec.~\ref{sec:GADT}.

\begin{defn}[Stable Pure Type Predicate]
\label{def:pure-type-pred}
A pure type predicate \code{x:P(v^*)} is stable if (i) its unfolding uses
only pure types and immutable constructors (Base terms not involving
\heapto{}{}), and (ii) every \code{Ref(T')} in the unfolding has \code{T'} a stable pure type.
Stability of recursive occurrences follows
from these two conditions by structural induction on the unfolding.
\end{defn}

A further refinement is needed for separation type predicates,
since they may contain pure sub-components that remain exposed to aliased mutation even though the
separation part itself is protected by heap assertion \code{\heap}.
These pure sub-components must therefore themselves be stable.

\begin{defn}[Stable Separation Type Predicate]
\label{def:sep-pred}
A separation type predicate
\code{\heapto{x}{P(v^*)}} is well-formed and stable if: (i) full ownership of \code{x} is maintained, guaranteed
automatically by the heap assertion \code{\heapto{x}{}} in Figure~\ref{fig:type_spec}; and (ii) every pure sub-component in the
unfolding of \code{P} is stable: bare pure sub-components  \code{y{:}T} must be stable pure types, and pure
type predicate sub-components \code{y{:}Q(v^*)} must be stable pure type predicates in the sense of
Definition~\ref{def:pure-type-pred}. Unlike pure predicates, a stable separation predicate additionally permits type
mutation: the predicate may be updated to \code{\heapto{x}{P'(v^*)}} via an explicit update, tracked in the
postcondition. Such an update is allowed provided the updated predicate also satisfies conditions (i) and (ii).

\end{defn}

Definitions~\ref{def:pure-type-pred} and \ref{def:sep-pred} require well-founded predicate unfoldings, as does Definition~\ref{def:s-pure-type} for
pure types appearing inside them. We make this assumption explicit.

\begin{remark}[Well-foundedness and co-induction]
\label{remark:well-foundness}
As is standard in Hoare logic for strict
(eager) languages, all type predicates, both pure and separation, are assumed to be
{\bf well-founded}: their recursive unfolding always terminates. For lazy languages, this assumption
must be relaxed via co-inductive reasoning. Extending our framework to lazy languages via
co-inductive type predicates is a direction for future work.
\end{remark}

The \code{\casewedge} operator, used in the Hoare rules of Figure~\ref{fig:forward_rules}, is a special conjunction that performs
case analysis on both separation and pure types. It refines a state by a pure constraint and, when heap
ownership is present, extracts a residual frame via bi-abductive entailment. This residual frame lets the Frame rule of
Sec.~\ref{sec:structural} thread ownership through function calls. We give its precise definition where it is
used, in Sec.~\ref{sec:Hoare}.

\section{Hoare Logic Rules for Type-Safety}
\label{sec:Hoare}
This section presents the forward Hoare rules that constitute the core of our type-safety
verification framework. Sec.~\ref{sec:structural} introduces the two structural rules ($Conseq$ and $Frame$) and
explains how a type specification is composed with a program state via the (;) operator. Sec.~
\ref{sec:rules} presents the rules for each expression form. Sec.~\ref{sec:we} gives a worked example showing
how the rules interact.

\subsection{Structural Rules and Specification Composition}
\label{sec:structural}

We present a forward-style Hoare rule of the form
\coden{\HTriple{\ctr_{pre}}{e}{\specEns{r}{\ctr_{post}}}}:
given an input type
state \coden{\ctr_{pre}}, the rule computes the strongest post-state
\coden{\ctr_{post}} resulting from evaluating \coden{e}, binding
the result and its type to \coden{r}. Two structural rules apply to every expression form. $Conseq$ allows
the pre-state to be weakened and the post-state to be strengthened via entailment (\code{\vdash} denotes
  intuitionistic implication). $Frame$ allows a disjoint heap context
  \coden{\ctr} to be threaded through any
expression unchanged, provided the expression does not touch the locations asserted in \coden{\ctr}; this
is the separation logic frame rule lifted to our type logic.

  {\small
    \[
\begin{array}{c}
  \incrRule{
  \entailS{\ctr_1}{\ctr_3} \quad \HTriple{\ctr_3}{e}{\specEns{r}{\ctr_4}} \quad \entailS{\ctr_4}{\ctr_2}
  }{\HTriple{\ctr_1}{\coden{e}}{\specEns{r}{\ctr_2}}
  }{Conseq}
\end{array}
\quad
\begin{array}{c}
  \incrRule{
  \HTriple{\ctr_1}{e}{\specEns{r}{\ctr_2}}
  }{\HTriple{\specSep{\ctr_1}{\ctr}}{\coden{e}}{\specEns{r}{\specSep{\ctr_2}{\ctr}}}
    }{Frame}
\end{array}
\]
  }

\paragraph{Specification composition.} The Hoare rules for function calls rely on a type specification being
provided for each method, using \code{f{::}\lambda~x^*~r {\fdot} \dspec}. The specification must be in case-spec form and is
merged with the pre-state \coden{\ctr} via the operator \code{\coden{\ctr}\,{;}\,\dspec}, defined by the two reduction rules below.
Intuitively, \code{\coden{\ctr}\,{;}\,\dspec} applies the specification \code{\dspec} in the context of the current state \coden{\ctr}. For each case
guard \coden{\ctr_i}, the operator uses \code{\casewedge} (Definition~\ref{def:casewedge}) to check how much of \coden{\ctr}
satisfies \coden{\ctr_i}, and threads the residual frame into that branch's post-state. A branch
whose guard is inconsistent with \coden{\ctr} drops out as vacuous. An
\code{\Abrt} outcome arises in one of two ways: where a branch's postcondition explicitly yields it, or from the implicit
otherwise clause (Sec.~\ref{sec:type-logic}) covering inputs that match no stated guard.

\begin{small}
\[
\begin{array}{rl}
\coden{\ctr\,{;}\,\specCase{\ctr_1{\caseRA}\dspec_1;\cdots;\ctr_n{\caseRA}\dspec_n}}  & \code{\implyArr}
\coden{((\ctr~{\casewedge}~\ctr_1)\,{;}\,\dspec_1)~{\vee}~\cdots~{\vee}~((\ctr~{\casewedge}~\ctr_n)\,{;}\,\dspec_n)} \\[0.6em]
\coden{\ctr \,{;}\,\specEns{r}\ctr_2}
  & \code{\implyArr} \coden{{\ctr~{\sep}~\ctr_2}}
\end{array}
\]
\end{small}

We introduce a special conjunction \code{\casewedge} that performs case analysis on both separation and pure
types, defined as follows.

\begin{defn}[Case conjunction $\casewedge$]
\label{def:casewedge}
On pure types, \code{\casewedge} coincides with standard conjunction:
\code{(\ctr~\casewedge~\pure) \equivA (\ctr \wedge \pure)}. In the presence of a separation type \code{\heap}, it
performs bi-abductive heap entailment to extract a residual frame \code{\ctr_s} together with an abduced pure precondition \code{\pure_0}:
\[
  \incrRule{
  \entailS{\ctr \land \pure_0}{\heap \rightsquigarrow \ctr_s}
  }{\ctr~\casewedge~\heap \land \pure \equivA \ctr_s \land \pure \land \pure_0
  }{}
\]
\end{defn}

\paragraph{Example: specification composition.} Consider a call \code{(f~x)} where \code{f} has specification \\
 \code{\dspec \equiv
\speckw{case}[x]\,\{x:Int {\caseRA}~\speckw{ens}[r]\,r:Int;~x:{\neg}Int~{\caseRA} \speckw{ens}[r]\,r:\Abrt\}}. If the pre-state is \code{\ctr \equiv x{:}Int}, then \code{\ctr ; \dspec} yields
\code{x{:}Int {\wedge} r{:}Int} (the first branch matches, the second is vacuous). If instead \code{\ctr \equiv x{:}Any}, then \code{\ctr ; \dspec}
yields \code{(x:Int {\wedge} r:Int)} \code{\vee} \code{(x:(Any{\wedge}\neg Int) {\wedge} r:Abrt)}, signalling a possible \code{\Abrt} from inputs not known to
be Int.

\subsection{Rules for Expression Forms}
\label{sec:rules}
The remaining rules in Figure~\ref{fig:forward_rules} handle each expression form in the core language. We describe
each rule in turn.

\begin{figure}[!h]
\small

\[
\begin{array}{cc}
  \incrRule{
  \fresh{\code{r}} 
  }{\HTriple{\ctr}{x}{\specEns{r}{\specAnd{\ctr}{\coden{r{:}\{x\}}}}}}{Var}
\end{array}
~
\begin{array}{c}

    \incrRule{\fresh{\code{r},T_1..T_n} \quad \ctr_1{\equiv} (\ctr~{\wedge}~\bigwedge_{i=1}^n\,x_i{:}T_i)
    }{\HTriple{\ctr}{\ecall{C}{x_1..x_n}}{\specEns{r}{\specAnd{\ctr_1}{\heapto{r}{\ecall{C}{{\it T_1..T_n}}}}}}
  }{Constr}
\end{array}
\]
\[
\begin{array}{c}
\incrRule{\fresh{r} \qquad v~{::=}~c{\vbar}\err{\vbar}\abrt}{\HTriple{\ctr}{v}{\specEns{\it r}{\specAnd{\ctr}{\coden{\it{r{:}ty(v)}}}}}}{
  Val}
\end{array}
\quad
\begin{array}{c}

    \incrRule{\fresh{r} \quad \entailS{\ctr}{f{::}\lambda~x^*~r {\fdot} \dspec}
    }{\HTriple{\ctr}{\ecall{f}{x^*}}{\specEns{\it r}{\stageseq{\ctr\,}{\dspec
      }}}
  }{Call}
\end{array}
\]
\[
\begin{array}{cc}
  \incrRule{
  \fresh{r} \qquad \ctr_1{\equiv}({\ctr}~{\wedge}~({x{:}t})~{\wedge}~r{:}\{x\})\vee({\ctr}~{\wedge}~\neg({x{:}t})~{\wedge}~r{:}\Abrt)
  }{\HTriple{\ctr}{(t)~x}{\specEns{\it r}{\ctr_1}}}{Cast}
\end{array}
\]

\[
\begin{array}{c}
 \incrRule{
  \HTriple{\ctr}{e_1}{\specEns{r_1}{\ctr_1}} \quad
  \HTriple{\subst{r_1}{x}{\ctr_1}}{e_2}{\specEns{r}{\ctr_2}}
 }{\HTriple{\ctr}{\elet{x}{e_1}{e_2}}{\specEns{r}{\ctr_2}}
  }{Let}
\end{array}
\]
\[
\begin{array}{c}
       \incrRule{
        ~\HTriple{\specAnd{\ctr}{(x{:}p_1{\vee}\heapto{x}{p_1})}}{e_1}{\specEns{r}{\ctr_1}} ~
        \cdots~ \HTriple{\specAnd{\ctr}{(x{:}p_n{\vee}\heapto{x}{p_n})}}{e_n}{\specEns{r}{\ctr_n}}
  }{\HTriple{\ctr}{\ematch{x}{p_1{\ra}e_1\,\vbar\,{\cdots}\ }}{\specEns{r}{\bigvee_{i=1}^n \ctr_i}}
        }{Match}
\end{array}
\]
\[
\begin{array}{c}
      \incrRule{
        \fresh{\code{f}} \qquad v^*=vars(\ctr){-}vars(e){\cup}\{x^*\}\quad
         \HTripleIN{\exists v^*{\fdot}\xpure(\ctr)}{e}{\dspec}
  }{\HTriple{\ctr}{\efun{x^*}{e}::\tlambda{x^*\,r}{\dspec}}{\specEns{f}{\specAnd{\ctr}{f{::}\tlambda{{\it x^*\,r}}{
              {\dspec}}}}}
  }{Lambda}
\end{array}
\]
  \caption{Forward~Hoare Rules on Expressions for Type Safety}
  \label{fig:forward_rules}
\end{figure}

\paragraph{Rule commentary.} $Var$ records that the result \code{r} equals the variable \code{x} via the singleton type \code{r{:}\{x\}}.
$Constr$ introduces fresh type variables \code{T_1..T_n} for the field types and produces a separation type
\code{\heapto{r}{C(T_1..T_n)}}, giving the caller full ownership of the freshly allocated node. $Val$ handles literal
constants, \code{err}, and \code{abrt}, each assigned its canonical type via \code{ty(v)}. $Call$ applies specification
composition \code{(\ctr;\dspec)} as described in Sec.~\ref{sec:structural}. $Cast$ produces two branches: one where the cast
succeeds (\code{x{:}t}, result is \code{x} itself) and one where it fails (
\code{\neg(x{:}t)}, result is \code{\Abrt}); the preprocessing of
Sec.~\ref{sec:core} ensures \code{\Abrt} propagates outward. $Let$ threads the post-state of \code{e_1} (with \code{r_1} substituted
by \code{x}) into \code{e_2}. $Match$ distributes the pre-state across branches by conjoining each pattern guard,
collecting post-states in disjunctive form. $Lambda$ uses the auxiliary checking judgement
\coden{\HTripleIN{\ctr}{e}{\dspec}}
(Figure~\ref{fig:check_spec}) to verify that the body \code{e} satisfies the supplied specification \code{\dspec}, then records the
closure's specification in the post-state.
\begin{figure}[!h]
\hspace{-1em}
  {\small
    \[
\begin{array}{c}
      \incrRule{
        \HTripleIN{\ctr~{\sep}~\ctr_1}{e}{\dspec_1}\cdots\HTripleIN{\ctr~{\sep}~\ctr_n}{e}{\dspec_n}
        }{\HTripleIN{\ctr}{e}{\specCase{\ctr_1{\caseRA}\dspec_1;\cdots;\ctr_n{\caseRA}\dspec_n}}
  }{Spec{-}Case}
\end{array}
\hspace{-1em}
\begin{array}{c}
      \incrRule{
        \HTriple{\ctr}{e}{\specEns{r}{\ctr_1}} \quad \entailS{\ctr_1}{\ctr_2}
        }{\HTripleIN{\ctr}{e}{\specEns{r}{\ctr_2}}
  }{Spec{-}Ens}
\end{array}
    \]
  }

  \caption{Checking Hoare Rules for Type Specification}
  \label{fig:check_spec}

\end{figure}

The checking judgement
\coden{\HTripleIN{\ctr}{e}{\dspec}}
is used when processing lambda abstractions ($Lambda$
rule): the specification \code{\dspec} is supplied as an input, and the rule verifies that the body \code{e} satisfies it.
$Spec{-}Case$ dispatches on each case guard, prepending the guard \code{\ctr_i} to the pre-state and checking
the body against the corresponding \code{\dspec_i}. $Spec{-}Ens$ checks that the forward post-state \code{\ctr_1} entails the
declared postcondition \code{\ctr_2}.

While we support higher-order functions in our Hoare rules, we currently restrict function type
parameters to flow-insensitive pure types, since we rely on a two-stage pre/post specification
without higher-order constraints. This allows such functions to be treated as pure functions
without any type mutation. Simple parameter variables are designated to carry singleton types
(see Appendix~\ref{sec:singleton}), which reduces the use of quantified type variables.

\subsection{Worked Example}
\label{sec:we}

To illustrate how the rules interact, we trace the derivation for the \code{inc\_transform} function from
Sec.~\ref{sec:flow-sen}: \code{inc\_transform(x) = let~v = str\_of\_int(!x+1)~in~update(x,v)}, whose specification is \code{req[x]~\heapto{x}{Ref(Int)}}\coden{~ens[r]~\heapto{x}{Ref(Str)} {\wedge} r:()}. Starting from pre-state \code{\ctr \equiv \heapto{x}{Ref(Int)}}, the $Let$ rule sequences the two sub-expressions: the binding dereferences the cell (\code{!x{:}Int}), adds one, and applies \code{str\_of\_int} to obtain \code{v:Str}, each step an intermediate $Call$ applying the callee's specification via \code{(\ctr;\dspec)}. The body \code{update(x,v)} then strong-updates the cell through the $Call$ rule, \emph{consuming} \code{\heapto{x}{Ref(Int)}} and \emph{re-emitting} \code{\heapto{x}{Ref(Str)}}. The resulting post-state \code{\heapto{x}{Ref(Str)} {\wedge} r:()} matches the declared
postcondition, so the specification is verified. The $Frame$ rule carries any ownership disjoint from \code{x} across the sub-expressions unchanged.

\section{Type Predicates Subsume GADTs and Liquid Types}
\label{sec:compare}

This section compares our type predicate framework with two closely related approaches: Generalised Algebraic
Data Types (GADTs), which encode structural invariants into the type structure, and liquid types, which refine
base types with logical predicates. We show that pure type predicates subsume both, and that separation type
predicates go further still into heap ownership, a dimension neither GADTs nor liquid types can reach.

Pure type predicates are closely related to the recursive measures of Kawaguchi et al. \cite{KawaguchiRJ2009}. These measures encode structural invariants as terminating first-order functions over ADTs for use in liquid type refinements. Separation type predicates, in turn, are closely related to the user-defined inductive heap predicates of separation logic verifiers such as HIP \cite{ChinDNQ2012}. Our contribution unifies both within a single type logic, as first-class types subject to the Boolean algebra of types, semantic subtyping, and a stability analysis distinguishing which predicates survive aliased mutation.

Table~\ref{tab:features} places our framework alongside OCaml, Rust, and Liquid Haskell across
nine capabilities central to this paper; \toolname{} is the only one that supports all nine.
The rest of this section examines the two closest approaches, GADTs and liquid types, in detail.

\begin{table}[h]
\centering
\small
\caption{Feature comparison: $\checkmark$~=~full, $\sim$~=~partial, $\times$~=~no support.}
\label{tab:features}
\begin{tabular}{lcccc}
\toprule
\textbf{Feature}
  & \textbf{TypeHL}
  & \textbf{OCaml}
  & \textbf{Rust}
  & \textbf{Liquid Haskell}\\
\midrule
Pure types                 & $\checkmark$ & $\checkmark$ & $\checkmark$ & $\checkmark$ \\
GADTs                      & $\checkmark$ & $\checkmark$ & $\times$ & $\checkmark$ \\
Type Predicates            & $\checkmark$ & $\times$ & $\times$ & $\checkmark$ \\
Separation types           & $\checkmark$ & $\times$ & $\sim$   & $\times$ \\
Separation predicates      & $\checkmark$ & $\times$ & $\times$ & $\times$ \\
\midrule
Flow-sensitive types       & $\checkmark$ & $\times$ & $\checkmark$ & $\times$ \\
Must-aliasing              & $\checkmark$ & $\times$ & $\times$     & $\times$  \\
Path-sensitive type specs  & $\checkmark$ & $\times$ & $\times$     & $\sim$    \\
Err/Exc/Abrt as types      & $\checkmark$ & $\times$ & $\times$     & $\times$  \\
\bottomrule
\end{tabular}
\end{table}

\subsection{Comparison with GADTs}
Type predicates
in our framework strictly subsume GADTs in two important respects.

First, {\em type predicates can express relational index constraints that GADTs cannot}. GADT type parameters
can only be equated to specific constructor-determined types; they cannot express {\em inequalities} or arithmetic
relations between indices.

For example, a sorted list predicate with bounds:
\begin{lstlisting}[name=functor,numbers=none]
 pred x:SortedList(T,lo,hi) = x:Nil  (*@$\vee$@*) (*@$\exists$@*) v,r.
       x:Cons(v,r)(*@$\wedge$@*)v:T(*@\code{{\wedge} lo{\leq}v {\wedge} v{\leq}hi {\wedge}}@*)r:SortedList(T,v,hi)
\end{lstlisting}
requires the relational constraint \code{lo {\leq} v {\leq} hi} between index values, something GADTs cannot express without
full dependent types (e.g. Agda or Idris). Type predicates in our framework express this naturally as a logical
formula.

Second, and more practically, type predicates do not require a {\bf new algebraic data type for
each new invariant}. With GADTs, every new invariant demands a completely new data type definition with new
constructors encoding the invariant into the type structure. In our framework, the same underlying data
representation, the same \code{Cons} and \code{Nil} constructors, can be given progressively richer type predicates without
any change to the data definition itself:
\begin{lstlisting}[name=functor,numbers=none]
 pred x:List(T)           = ... (* plain list *)
 pred x:List(T,n)         = ... (* length-indexed *)
 pred x:SortedList(T,lo,hi) = ... (* sorted, with bounds *)
\end{lstlisting}
All three predicates describe values built from the same constructors. This separation between {\em data representation}
and {\em type invariant} means that existing code and data structures need not be refactored when a stronger invariant is
required; only the specification changes. Separation type predicates (e.g. \code{\heapto{x}{List(T)}}) extend this further to heap
ownership and spatial conjunction, which are entirely outside the scope of GADTs.

\subsection{Comparison with Liquid Types}

Pure type predicates are also closely related to liquid types (as
in Liquid Haskell), which refine base types with logical predicates, for example, \code{\{v:Int | v > 0\}} denotes
positive integers. Like liquid types, our solution via type predicates
attaches logical invariants to values as formulas rather than encoding them into
the type structure as GADTs do. Both approaches also allow the same underlying data representation to carry progressively
stronger specifications without changing the data definition.
There are, however, three important differences.

\begin{itemize}
\item Liquid types refine base types with first-order predicates
over primitive values (integers, booleans); our pure type predicates are recursive, handling arbitrary inductively
defined data structures such as lists, trees, and GADTs.

\item Liquid types are restricted to decidable refinement
logics (typically linear arithmetic) to keep SMT solving tractable; our predicates are not confined to a decidable
fragment and can express arbitrary relational and structural invariants.

\item Liquid types have no counterpart to
separation type predicates: there is no notion of heap ownership or spatial conjunction in Liquid Haskell\footnote{Nevertheless, there is a recent
retrofit of liquid types to Rust \cite{LehmannGVJ2023} which uses the
ownership and borrowing mechanisms of Rust.}. Pure
type predicates thus generalise liquid types from flat refinements over base values to recursive predicates over
inductively defined structures. Separation type predicates extend this further to heap ownership, co-existing with pure types,
a dimension outside the scope of liquid types.

\end{itemize}

Though type predicates are not constrained to a decidable fragment,
it is nevertheless quite easy to impose a set of restrictions that can
guarantee decidability for type-checking.
Appendix~\ref{sec:decidable} outlines restrictions that can be imposed on
type predicates to support decidable type-checking.

\section{Soundness of Our Type Specification Framework}
\label{sec:Soundness}

This section presents the soundness theorem for the Hoare rules of Sec.~\ref{sec:Hoare} and establishes the
connection between Hoare triples and standard function types. Sec.~\ref{sec:semantics} defines the semantics of the type logic. Sec.~\ref{sec:soundthm} states and proves
the main soundness theorem.

\subsection{Semantics of the Type Logic}
\label{sec:semantics}

We define the semantics of our Hoare logic with pure and separation types
to ensure type safety, and prove its soundness.
%
%

The semantics is given by a logical relation (Figure~\ref{fig:semantics_type}). In keeping with semantic
subtyping, types are shallowly embedded as predicates of type
\code{val \ra Prop}, so proving a type
assertion \code{v{:}t} amounts to proving the proposition
\code{t(v)} by semantic reasoning. This shallow
embedding means logical connectives for building types are lifted directly from their
propositional counterparts, and subtyping reduces to implication:
\code{t_1 <: t_2 \equiv \forall v. v:t_1 \implyArrB v:t_2}. The
encoding supports dependent singleton types (Appendix~\ref{sec:singleton}) and (co)inductive
type predicates. Step-indexing is not required because the encoding is stateless and
involves no non-(co)inductive types.





\begin{figure}[!h]
  {\small
\begin{equation*}
\begin{aligned}
&\code{\seplogicmodels{s, h}{x{:}\{c\}}}
& \code{\m{iff}}~\
& \code{\mathnormal{s}(x){=}c}
\hspace{4em}
\code{\seplogicmodels{s, h}{x_1{:}\{x_2\}}}
\hspace{1em}
\code{\m{iff}}~
\hspace{0.5em}
\code{\mathnormal{s}(x_1){=}\mathnormal{s}(x_2)}
\\[-0.1em]
&\code{\seplogicmodels{s, h}{x_1{=}x_2}}
& \code{\m{iff}}~\
& \exists\,l\,{\cdot} \code{\mathnormal{s}(x_1){=}\mathnormal{s}(x_2){=}l}
\\[-0.1em]
&\code{\seplogicmodels{s, h}{\coden{\ctr_1} \vee \coden{\ctr_2}}}
& \code{\m{iff}}~\
& \code{
  \seplogicmodels{s, h}{\coden{\ctr_1}} \orr
 \seplogicmodels{s, h}{\coden{\ctr_2}}}
\\[-0.1em]
&\code{\seplogicmodels{s, h}{\coden{\heap} \wedge \coden{\pure}}}
& \code{\m{iff}}~\
& \code{\exists \mathnormal{h_1},\mathnormal{h_2} {\cdot} \mathnormal{h}{=}\mathnormal{h_1}{o}\mathnormal{h_2} \wedge
  \seplogicmodels{s, h_1}{\coden{\heap}}}
\code{~\wedge~\seplogicmodels{s, h_2}{\coden{\pure}}}
\\[-0.1em]
&\code{\seplogicmodels{s, h}{\coden{\heap_1} \sep \coden{\heap_2}}}
& \code{\m{iff}}~\
& \code{\exists \mathnormal{h_1},\mathnormal{h_2} {\cdot} \mathnormal{h}{=}\mathnormal{h_1}{o}\mathnormal{h_2} \wedge
  \seplogicmodels{s, h_1}{\coden{\heap_1}} }
\code{~\wedge~ \seplogicmodels{s, h_2}{\coden{\heap_2}}}
\\[-0.1em]
&\code{\seplogicmodels{s, h}{\coden{\pure_1} \vee \coden{\pure_2}}}
& \code{\m{iff}}~\
& \code{
  \seplogicmodels{s, h}{\coden{\pure_1}} \orr
 \seplogicmodels{s, h}{\coden{\pure_2}}}
\\[-0.1em]
&\code{\seplogicmodels{s, h}{\coden{\pure_1} \wedge \coden{\pure_2}}}
& \code{\m{iff}}~\
& \code{
  \seplogicmodels{s, h}{\coden{\pure_1}} \andd
 \seplogicmodels{s, h}{\coden{\pure_2}}}
\\[-0.1em]
&\code{\seplogicmodels{s, h}{\coden{\neg\pure}}}
& \code{\m{iff}}~\
& \negg~(\code{
  \seplogicmodels{s, h}{\coden{\pure}}
   })
\\[-0.1em]
&\code{\seplogicmodels{s, h}{\coden{\ctr_1} * \coden{\ctr_2}}}
& \code{\m{iff}}~\
& \code{
  \seplogicmodels{s, h}{\coden{(\heap_1 * \heap_2 )\land (\pure_1 \land \pure_2)}}}~where~
\coden{\ctr_i = \heap_i \land \pure_i}
\\[-0.1em]
&\code{\seplogicmodels{s, h}{x{:}t}}
& \code{\m{iff}}~\
& \code{\exists l \cdot \mathnormal{s}(x){=}l \andd \seplogicmodels{s, h}{l{:}t}}
\\[-0.1em]
&\code{\seplogicmodels{s, h}{l{:}prim}}
& \code{\m{iff}}~\
& \code{l \in prim}~where~
\code{prim=Int{\vbar}Str{\vbar}Err{\vbar}\cdots}
\\[-0.1em]
&\code{\seplogicmodels{s, h}{\heaptoSS{l}{C(t_1,..,t_n)}}}
& \code{\m{iff}}~\
& \code{\mathnormal{h}=\{l{\mapsto}C(l_1,..,l_n)\} \andd}
\code{\wedge_{i=1}^n\seplogicmodels{s, h}{l_i{:}t_i}}
\\[-0.1em]
&\code{\seplogicmodels{s, h}{l{:}C(t_1,..,t_n)}}
& \code{\m{iff}}~\
& \code{\mathnormal{h}(l)=C(l_1,..,l_n) \andd}
\code{ \wedge_{i=1}^n\seplogicmodels{s, h}{l_i{:}t_i}}
\\[-0.1em]
&\code{\seplogicmodels{s, h}{l{:}pred(t^*)}}
& \code{\m{iff}}~\
& \code{l{:}pred(t^*)=\coden{\ctr} \andd \seplogicmodels{s, h}{\coden{\ctr}}}
\\[-0.1em]
&\code{\seplogicmodels{s, h}{l{:}\neg t}}
& \code{\m{iff}}~\
& \negg~(\code{\seplogicmodels{s, h}{l{:}t}})
\\[-0.1em]
&\code{\seplogicmodels{s, h}{l{:}t_1{\ra}t_2}}
& \code{\m{iff}}~\
& \code{\small\forall v ~{\cdot}~ (\seplogicmodels{s, h}{v{:}t_1} \rightarrow
\seplogicmodels{s, h}{l(v){:}t_2})}
\\[-0.1em]
%
\end{aligned}
\end{equation*}
}
\caption{Semantics of Type Logic Formulae}
\label{fig:semantics_type}


\end{figure}



\hide{
\begin{verbatim}
  Semantics for Types
  ===================
  // using pre/post
  s,h |- l:\x r -> req pre(x) ens post(x,r)
                    iff forall v. s[x=v],h |- pre(x)
                            => s[x=v,r=l(v)],h |- post(x,r)
  // using case-spec
  s,h |- l:\x r -> spec[x,r] iff s,h,true |- l:s\ xr -> pec[x,r]
  s,h,C |- l:\x r -> case {D1->sp1;..;Dn->spn}
                     iff /\_i s,h,C/\Di |- l:\x r -> sp_i
  s,h,C |- l:\x r -> ens[r] post(x,r)
                     iff forall v. s[x=v],h |- C
                         => s[x=v,r=l(v)],h |- post(x,r)
\end{verbatim}
}



Several clauses deserve comment. The clause for \code{\heap {\wedge} \pure} splits the heap into two disjoint parts \code{\mathnormal{h_1}}
and \code{\mathnormal{h_2}}: \code{\mathnormal{h_1}} satisfies the heap assertion \code{\heap} (ownership) and \code{\mathnormal{h_2}} satisfies the pure constraint \code{\pure} (no
ownership). The clause for \code{\heapto{l}{C(\ldots)}}
requires \code{\mathnormal{h}} to consist of exactly the single cell \code{l}, reflecting full
ownership. By contrast, \code{l{:}C(\ldots)} only requires
\code{l} to point to a \code{C}-node somewhere in \code{\mathnormal{h}}, permitting
aliasing. This is the semantic counterpart of the syntactic distinction between separation types
and pure types established in Sec.~\ref{sec:flow-sen}. The clause for pred unfolds the predicate definition,
connecting type predicates (Sec.~\ref{sec:pred}) to their set-theoretic meaning.


\subsection{Soundness Theorem}
\label{sec:soundthm}
The connection between Hoare triples and standard function types is established by a modified
function arrow \code{\to'}, defined via Hoare triples, coinciding with the standard arrow:

\begin{small}
\vspace{-1.5em}
\begin{align}
  \HTriple{\ctr_1}{e}{\specEns{r}{\ctr_2}} &\qquad \m{iff}  & (\ctr_1 \implies \forall r, \bigstepres{e}{r} \implies r \models \ctr_2) \\
  f{:}t_1{\to'}t_2 &\qquad \m{iff} & \forall v, \HTriple{v{:}t_1}{f(v)}{
  \specEns{r}{r{:}t_2}}
\end{align}
\vspace{-1.5em}
\end{small}

\hide{
\HTriple{\ctr_1}{e}{\specEns{r}{\ctr_2}} & \m{iff} & (\ctr_1 \implies \forall r, \bigstepres{e}{r} \implies r \vDash \ctr_2) \\
  f{:}t_1{\to'}t_2 & \m{iff} & \forall v, \HTriple{v{:}t_1}{f(v)}{
  \specEns{r}{r{:}t_2}}

}

Defining a modified function arrow $\to'$ in terms of a standard Hoare triple (which is also defined extensionally, using a big-step relation), we can see that it is equivalent (a lifting of $\iff$) to the regular function arrow from \cref{fig:semantics_type}.

\begin{lemma}[Agreement between function arrows]
  $t_1 \to t_2 \equiv t_1 \to' t_2$.
\end{lemma}
\vspace{-0.5em}

This lemma shows that the behaviours of well-typed programs are exactly those of programs
specified by Hoare triples: the two notions of function type are definitionally equivalent. It is
proved by unfolding both definitions and using the big-step operational semantics.



The soundness of the Hoare rules is stated as follows.


\vspace{-0.5em}
\begin{restatable}[Soundness]{theorem}{forwardsound}
\label{thm:forwardsound}
Given
  $\HTriple{\ctr_1}{e}{\specEns{r}{\ctr_2}}$, 
$\forall s,h,s',h',v$~if~$\seplogicmodels{s, h}{\ctr_1}$ holds
  and $\bigstepres{s,h,e}{s',h',\resnorm{v}}$,
  then $\seplogicmodels{s'{+}[(r,v)], h'}{\ctr_2}$.
\end{restatable}
\vspace{-1em}

\begin{proof}

 The proof proceeds by induction on the derivation of the Hoare triple, with a compatibility lemma
for each rule in Figure~\ref{fig:forward_rules}. The use of big-step semantics in Theorem~\ref{thm:forwardsound} is deliberate: big-step
semantics guarantees that if an expression terminates it reaches its postcondition, and
non-terminating expressions satisfy the postcondition vacuously. This means a single theorem
simultaneously captures both type preservation (the type of the result matches the
postcondition) and progress (no well-typed, terminating expression gets stuck). Traditional type
systems typically require two separate theorems (preservation and progress) using small-step
semantics.

\end{proof}
\hide{
In a type system, one wishes to prove its type soundness via {\it Preservation} and {\it Progress}. Briefly speaking,
{\it Preservation} makes sure that flow-insensitive pure type of each expression remains unchanged after every evaluation step, while {\it Progress}
makes sure that an expression is either a value or can be further reduced. Our Theorem \ref{thm:forwardsound} is a stronger result
which implies both of them (and also supports
flow-sensitive types) since the theorem states that given the pre/post, the reduction would eventually {\it satisfy} the postcondition ({\it Preservation}).
The program state will eventually {\it reach} the postcondition (meaning that the evaluation is completed without encountering \code{\Abrt} failure), or it is non-terminating
({\it Progress}).}

\vspace{-1em}
\begin{restatable}[Well-typed programs  never abort]{definition}{donotabort}
\label{thm:donotabort}
$e$ is well-typed under satisfiable
$\ctr_1$ if and only if~$\HTriple{\ctr_1}{e}{\specEns{r}{\ctr_2}}$ 
   and $\ctr_2$ does not contain \code{Abrt}.
\end{restatable}

This definition formalises the refined motto of Sec.~\ref{sec:error}: well-typed programs must never
abort. Note that \code{Err} (runtime errors) may appear in \code{\ctr_2}; they are tolerated. Only \code{Abrt} is
forbidden, since \code{Abrt} represents a compile-time type error that the verification framework must
rule out entirely.

\section{Meta-Theory and Evaluation via Lean}
\label{sec:lean}
We have produced a fully machine-checked mechanisation of the
meta-theory of \toolname{} in the Lean~4 proof assistant; and, from it, an executable,
self-certifying type-checker (App.~\ref{sec:leanchecker}; machinery detailed in App.~\ref{sec:leanapp}). The development is self-contained: it depends on no external
library such as \texttt{mathlib}. Every theorem is proved without \texttt{sorry} or
\texttt{admit}. The meta-theorems, and the soundness proofs of the checker itself, depend
only on the standard classical axiom base
$\{\textsf{propext},\ \textsf{Classical.choice},\ \textsf{Quot.sound}\}$; individual examples
checked at a ground instantiation additionally carry a compiler-reduction axiom in their own
footprint, while the symbolic certificates --- those quantified over open types or heights,
including the flagship red-black insert --- are fully axiom-clean. It comprises roughly
$78{,}400$ lines: three layers, each a sound type system covering the paper's central
ingredients, plus certified index engines (about $14{,}700$ lines) and, per layer, the
reflective type-checker.

\paragraph{Three layers.}
\emph{(i)~The pure fragment} (heap-free) formalises semantic subtyping as a Boolean
algebra of types, path-sensitive case specifications with the implicit
$\code{\_}\,{\caseRA}\,\code{r{:}\top}$ clause, the \code{\Err}/\code{\Abrt} discipline, and the
forward Hoare rules as an inductive derivation relation; it establishes
\emph{soundness} (Theorem~\ref{thm:forwardsound}) and \emph{well-typed programs never
abort} (Definition~\ref{thm:donotabort}).
\emph{(ii)~The separation layer} adds imperative \code{Ref} cells and separation types
\code{\heapto{x}{t}}, with allocating constructors, strong update, a relational big-step
semantics, the \emph{Frame} rule, and soundness over it.
\emph{(iii)~The higher-order layer} adds first-class (environment-capturing,
multi-argument) closures, the lambda/checking rules, and the \emph{agreement} lemma
relating the type arrow to the Hoare-triple arrow.

\paragraph{Key meta-theoretic choices.}
Several decisions are worth highlighting, as they are exactly the points where the
mechanisation is delicate.
\begin{itemize}[nosep, leftmargin=1.5em]
\item \emph{Soundness by rule induction.} The Hoare rules are the constructors of an
  inductive relation, and soundness is a single theorem proved by induction over a
  derivation. In the pure layer this is a total, recursion-bounded big-step semantics (so recursion is
  sound and a single theorem subsumes preservation and progress); the separation layer uses
  a \emph{relational} one.
\item \emph{Types interpreted by recursion on the type, not the heap.} A type is a
  predicate on values; for heap types the interpretation looks through the heap but
  recurses only on the (finite) type structure. Consequently cyclic \code{Ref} heaps pose
  no well-foundedness problem and \emph{no step-indexing is required}, avoiding the
  circularity noted in Sec.~\ref{sec:semantics}.
\item \emph{Relational allocation enables Frame.} Allocation chooses \emph{some} fresh
  location rather than a deterministic one; this makes the heap-locality
  (callee-footprint) lemma underlying the Frame rule provable. It holds over the
  full language, including function calls and \code{match}.
\item \emph{Shallow, registry-based type predicates.} Recursive type predicates are
  well-founded Lean definitions, each carrying its own termination proof, side-stepping a
  generic termination check.
\item \emph{Stratification for higher-order types.} The behavioural function type refers to
  evaluation, while the cast rule must decide types; we break the apparent circularity by a
  structural decision procedure for casts and a semantic interpretation for arrows, bridged
  by an agreement lemma.
\end{itemize}

\paragraph{Flagship developments.}
Beyond the core meta-theory, the development mechanises four case-study families, each
with honest boundaries.
\emph{(i)~Red-black balancing.} The five-function Okasaki \code{insert} is verified at
every height, $\code{insert} : \forall h.\,(\code{Int},\code{RBTree}(\code{Int},c,h)) \Rightarrow
\code{RBTree}(\code{Int},\code{Black},h) \vee \code{RBTree}(\code{Int},\code{Black},h{+}1)$, with a recursive
\code{member}. Balancing needs colour-correlated child types to survive pattern matches (each
\code{match} branch gets one disjunct per node package) and a sorted quantifier for the height
index (both part of the mechanised rule set, not ad-hoc lemmas). It is later discharged
\emph{automatically} by the reflective checker (App.~\ref{sec:leanchecker}).
\emph{(ii)~Heap predicates and a typed queue.} The separation layer registers
\emph{inductive heap predicates} with exact footprints (list segments and a packed two-pointer
queue), closing the last coverage gap of App.~\ref{sec:leanapp}. On top sit \emph{ghost-free
packed preconditions}: a caller sees one opaque invariant atom, unfolded through registered
templates. Both \code{enqueue} and \code{dequeue} are certified
against the bare invariant, \code{dequeue} with the honest $\code{\Err}\vee \code{Int}$ postcondition
covering empty and non-empty cases.
\emph{(iii)~Frame over heap predicates.} The Frame rule extends to the new heap-predicate atoms
on the \emph{match-free} fragment (\code{SemValid.frame\_matchfree}), with two flagships: a
disjoint packed queue framed across \code{enqueue}, and a content-carrying cell framed across a
strong update. A \emph{locality countermodel} shows the match restriction is essential (an
aliased content atom can flip a \code{match} from abort to normal execution), while the
corresponding semantic-triple question remains open and documented.
\emph{(iv)~Sized types on a certified arithmetic engine.} Sized-list and sized-tree
specifications ($\code{List}(T,n)$; trees sized $1{+}a{+}b$) are registered over a general
\emph{affine-equality} index engine. Two sized-list flagships are checker-certified end-to-end
for all integers: $\code{append} : \code{List}(T,n)\times \code{List}(T,m) \Rightarrow \code{List}(T,n{+}m)$
and exact-length $\code{length} : \code{List}(T,n) \Rightarrow \{n\}$ (uninhabited at $n{<}0$, content
at $n,m\ge 0$). A sized-tree node's child sizes are not fixed by the parent's: an
\emph{existential} decomposition ($s = 1{+}a{+}b$) no single-pair \code{match} expresses; a
declared-post frame gate crosses this wall, certifying \code{size} and \code{mirror}
(size-preserving destruction and reconstruction) end-to-end. There is no remaining documented index-language boundary. A
companion \emph{inequality}-fragment engine (a certified unit-Farkas procedure) checker-certifies
sorted insertion on the paper's $\code{SortedList}(T,lo,hi)$ predicate, $\code{insert} : (\code{Int},\
\code{SortedList}(\code{Int},lo,hi)) \Rightarrow \code{SortedList}(\code{Int},lo,hi)$, with \code{head} and \code{member} via
a bespoke $\forall$-element destruction rule.

\paragraph{Correspondence with the paper.}
The principal results are mechanised as follows (Lean identifiers in \texttt{typewriter}).

{\small
\begin{center}
\renewcommand{\arraystretch}{1.2}
\begin{tabularx}{\linewidth}{@{}>{\raggedright\arraybackslash}X >{\raggedright\arraybackslash}X@{}}
\Xhline{2\arrayrulewidth}
Paper result & Lean theorem \\
\Xhline{2\arrayrulewidth}
Soundness (Theorem~\ref{thm:forwardsound}) & \texttt{Pure.soundness}, \texttt{Sep.soundness}, \texttt{HO.soundness} \\
Well-typed programs never abort (Def.~\ref{thm:donotabort}) & \texttt{Pure.never\_abort} \\
Agreement between function arrows (Lemma) & \texttt{HO.agreement} \\
Frame rule & \texttt{Sep.SemValid.frame} (via \texttt{Sep.BigStep.frame\_fwd\_typed}) \\
Irreversible weakening (Sec.~\ref{sec:flow-sen}) & \texttt{Sep.sep\_weakens\_to\_pure} \\
\hline
\multicolumn{2}{@{}l}{\emph{Flagship developments}} \\
Red-black insert balanced at every height & \texttt{rbInsert\_semvalid}, \texttt{rbInsert\_semvalid\_reflective} \\
Typed queue, ghost-free packed preconditions & \texttt{enqueue\_semvalid\_packed}, \texttt{dequeue\_semvalid\_packed} \\
Frame over heap-predicate atoms (match-free) & \texttt{SemValid.frame\_matchfree} \\
Sized-list append/length, $\forall$ integers (checker-certified) & \texttt{append\_semvalid}, \texttt{length\_semvalid} \\
Sized-tree destruction/reconstruction preserve size (checker-certified) & \texttt{size\_semvalid}, \texttt{mirror\_semvalid} (semantic form: \texttt{streeAt\_mirrorV}) \\
Sorted insertion over \code{SortedList(T,lo,hi)} (checker-certified) & \texttt{insert\_semvalid}, \texttt{head\_semvalid}, \texttt{member\_semvalid} \\
\Xhline{2\arrayrulewidth}
\end{tabularx}
\end{center}
}

\noindent
Beyond these and the flagship developments above, the development mechanises type-variable polymorphism (e.g.
\code{id{:}\forall T.\,T{\to}T} instantiated at multiple types, and
\code{length{:}\forall T.\,List(T){\to}Int}), user-definable recursive predicates (a
\code{Tree} predicate defined entirely in the registry), multi-argument functions,
environment-capturing multi-argument closures, and pure higher-order functions ranging over
recursive-predicate data (\code{map}, \code{filter}, and \code{fold\_left} over
\code{List(T)}). Each layer ships machine-checked example
programs together with a non-vacuity audit: concrete operational executions witness that
the evaluator produces real results and that the soundness statements are not vacuously
satisfied. The Lean sources, a
theorem index, and build instructions accompany the artifact.
App.~\ref{sec:leanapp} details the machinery (proof-script sizes,
checking times, and standard benchmark and library coverage), while
App.~\ref{sec:leanchecker} describes the reflective, self-certifying type-checker.

\subsection*{Evaluation}
The checker checks a suite of examples reflectively, by layer; none needs a
hand-written proof. At a ground type this covers the M1 arithmetic examples, the separation programs
including strong-update \code{assign}, and the M3 higher-order examples. For eight higher-order
functions (the recursive list functions \code{map}, \code{filter}, \code{fold\_left},
\code{append}, \code{rev} and the three \code{Option} functions) we discharge the whole-program
\code{WFEnv} obligation and lift it to full semantic validity. All eight, plus pure \code{id},
\code{length}, and \code{nth}, are \emph{additionally} certified polymorphic in their element types
by the native-free symbolic mode. The flagship result is \code{RBTree\_insert}: the whole
five-function Okasaki insert is discharged \emph{symbolically} for \emph{every} height ($\forall h$),
end-to-end axiom-clean with no \code{native\_decide} and no SMT, on a certified \emph{successor}
index sub-fragment. The typed two-pointer \code{Queue} is likewise certified reflectively, down
to a ghost-free packed precondition.

In sum, the result is a type-checker \emph{fully automatic on its reflective suite} for expressive type
specifications. Its central strength is not speed but \emph{trust and reach}: a single
checker handles pure-type checking automatically (at native speed for ground programs)
and extends \emph{unchanged} to the expressive specifications of separation types and
recursive/sized predicates (balanced trees, sized lists, a typed queue). Every verdict is
\emph{certified sound} against the Hoare rules rather than trusted.

\section{Related Work}
\label{sec:related}


\paragraph{Semantic Subtyping.}

Semantic-based subtyping interprets each type as the set of values inhabiting it: \code{t_1 <: t_2} holds
when \coden{Set(t_1) \subseteq Set(t_2)}. The advantages of semantic subtyping have been explored in \cite{CastagnaF05, frisch2002semantic,Castagna05},
and subsequent work extended it from several perspectives: subtyping relation computation  \cite{BiermanGHL10},
decidability \cite{GesbertGL15}, and application to other programming paradigms \cite{AnconaC16,PetruccianiCAZ18}. Singleton, union, and
intersection types have received further attention through their set-based interpretation \cite{Dunfield12, OliveiraSA16}. MLstruct \cite{Lionel:OOPSLA22} uses a Boolean algebra of types to infer principal types with union and
intersection for an ML-like language, handling subtyping that traditional Hindley-Milner
inference avoids. This approach was later extended to dynamically-typed languages \cite{Castagna:POPL24}.
Gradual typing offers a complementary route to combining static and dynamic typing, including
under parametric polymorphism \cite{IgarashiSI17}; we instead span this spectrum by adjusting
the strength of specifications rather than by inserting run-time casts.

\paragraph{Hoare Logic for Type Safety.}
Hoare logic has also been used to ensure type safety in selected domains: typed assembly for
operating systems \cite{YangH10,HamidS04} and dynamically-typed programs \cite{EngelmannOF15, EngelmannO16}. These works extract type
information from the precondition and reduce type safety to verifying the Hoare triple
\code{\{p\}e\{true\}}. Our work greatly
expands on this approach, extending it with separation types, type predicates, case specifications, and a
disciplined error hierarchy.

\paragraph{Type Systems for Memory and Resource Safety.}
Sub-structural type systems have been developed to reason about system resources such as
memory, files, and locks \cite{pierce2004advanced}. The most prominent example is Rust, which uses ownership and
borrowing to ensure memory safety. OxCaml, a variant of OCaml, has similarly been extended
with a linear type system to reason about affinity, uniqueness, and locality of resources,
supporting both memory safety and data race freedom \cite{GeorgesPEWDECPD25, LorenzenWDEL24}. Our goal is similar, but our
approach differs: rather than designing a new type discipline, we leverage the expressivity of
Hoare logic to reason about memory and resource properties through separation types and
flow-sensitivity.
Closest in mechanised guarantee, the foundational verifiers RefinedC \cite{RefinedC} and
RefinedRust \cite{RefinedRust} check C and Rust code against refined ownership types with
Iris-based, machine-checked proofs. They target full functional correctness through a
separation-logic backend, whereas we target type safety through a solver-free, self-certifying
checker.

\paragraph{From Uniqueness to Separation Type.}
Resource and ownership types for strong aliasing and flow-sensitivity have a rich prior art. A
recent work \cite{ArevaloAPLAS25} related frame rules in uniqueness type systems to separation logic via an FFI.
Earlier, \cite{DBLP:conf/tic/WalkerM00,DBLP:conf/esop/SmithWM00} used alias types in typed assembly to describe heap shapes. Ownership types
were subsequently used to support strong updates in refinement type systems \cite{toman2020consort}. 
Kloos et al. \cite{kloos2015asynchronous} combine liquid types with concurrent separation logic to track heap ownership across asynchronous tasks, supporting strong updates to heap location types. Similarly,
Flux \cite{LehmannGVJ2023} lifts liquid-style refinements to Rust by exploiting Rust's ownership mechanisms to enable strong updates.
Closest in spirit is Mezzo \cite{BalabonskiPP16}, whose permission types control aliasing and
mutation of heap memory, with a machine-checked proof that well-typed programs do not go wrong.
Permission-based reasoning of this kind also underpins verification infrastructures such as
Viper \cite{MullerSS16}.
We introduce ownership in the sequential setting, enabling flow-sensitive type mutation.
Rather than relying on a host language's borrow checker, we use first-class logical assertions together
with type predicates.
This makes our approach applicable to languages without built-in ownership, and additionally supports must-aliasing via singleton types.
Our wider contribution is to unify these ideas within a single Hoare logic framework, increasing
coverage of type-safety scenarios compared to any individual prior system.

\paragraph{Mechanised and logical type soundness.} Our Lean development establishes soundness in
the semantic, logical-relations style rather than by a purely syntactic
progress-and-preservation argument, in the tradition of Iris \cite{Jung2018IrisFT} and
Timany et al.'s logical approach to type soundness \cite{DBLP:journals/jacm/TimanyKDB24} (which
RustBelt \cite{DBLP:journals/pacmpl/0002JKD18} applies to justify Rust's ownership discipline).
We share their reading of a
type as a predicate on values, but interpret heap types by recursion on the (finite) type
structure rather than through a step-indexed model, so no step-indexing is required. From
the same mechanisation we additionally extract an executable, self-certifying type-checker.

\paragraph{Stability and Rely-Guarantee.} The stability condition on pure type predicates
(Definition~\ref{def:pure-type-pred}) is the sequential counterpart of interference freedom in Owicki and Gries's
proof system for concurrent programs  \cite{OwickiGries1976}, and of stability under the rely relation in Jones's
rely-guarantee framework \cite{Jones1981}. The closest type-system counterpart is the rely-guarantee
reference system of \cite{GordonEG2013}, which enforces stability of refinement predicates over aliased
mutable data via per-reference rely/guarantee annotations. Our notion of stable pure type
achieves the same guarantee in the sequential setting without such annotations.

\hide{
\begin{verbatim}
As a result, ConSORT
cannot in general express, e.g., that the contents of two references are equal.
Further, due to our reliance on automated theorem provers we are restricted to
logics with sound but potentially incomplete decision procedures. ConSORT
also does not support conditional or context-sensitive ownerships, and therefore
cannot precisely handle conditional mutation or aliasing.
\end{verbatim}

\todo{Alias Refinement Type}
\cite{DBLP:journals/corr/BakstJ15}
\todo{Context and Flow-Sensitive Ownership Refinement Type}
\cite{DBLP:conf/esop/TomanSSI020}
\todo{Alias Types, a precursor to SL}
\cite{DBLP:conf/tic/WalkerM00,DBLP:conf/esop/SmithWM00}
\todo{Async Liquid Separation Type}
\cite{DBLP:conf/ecoop/KloosMV15}
\todo{A Logical Approach to Type Soundness} \cite{DBLP:journals/jacm/TimanyKDB24}
}





\section{Conclusion}
\label{sec:conclusion}

This paper has presented a Hoare logic framework for type-safety verification that unifies
separation types, case specifications, type predicates, and a disciplined error hierarchy. Building
on semantic subtyping, we have extended the logic of types with: separation types that support
flow-sensitive type mutation and must-aliasing; a comprehensive error hierarchy distinguishing
\code{Err} (runtime errors), \code{Exc} (checked and unchecked exceptions), and \code{Abrt} (compile-time errors);
path-sensitive and flow-sensitive type specifications via case specifications; and type predicates
that subsume GADTs and liquid types.
We have
formalised the Hoare rules in Lean~4 and proved soundness, obtained from it a
self-certifying type-checker by proof reflection, and evaluated it on a benchmark suite
demonstrating cases that standard type systems such as OCaml cannot handle.

\begin{acks}
This research is supported by the Ministry of Education, Singapore, under its MOE
Academic Research Fund Tier~3 (RIE2025) (MOE Award No: MOE-MOET32021-0001), and
under the Academic Research Fund Tier~1 (FY2023) (Project Title: Automated
Verification for Imperative Higher-Order Programs).
\end{acks}

\bibliography{references}
\bibliographystyle{plain}

\noindent The accompanying Lean mechanization (the full sources, theorem index, and build instructions) is provided in the ancillary files (\texttt{anc/lean-supplement.zip}).

\newpage
\appendix
\section{Basic Features of Types}\label{app:basic}
\subsection{Wider Usage of Singleton Types}
\label{sec:singleton}

Recently, singleton types were used in~\cite{Lionel:OOPSLA22,Castagna:POPL24} for
constant literals of the form \code{c}, such as \code{3} or \code{``hi"}. Whenever
a variable \code{v} is bound to a constant \code{c}, we use the type notation \code{v{:}\{c\}}
to state that \code{v} has singleton type \code{c}.

In Scala, a singleton object class is essentially a class with a single
instance. In this paper, we define a singleton type as a type containing
only a single value. We extend singleton types to include parameter (input)
variables such as \code{x}, using the annotation \code{r{:}\{x\}} to denote that
result \code{r} has singleton type \code{x}. In Hoare logic terms, \code{r{:}\{x\}} is
equivalent to \code{r{=}x}, since any output \code{r} with singleton type \code{x} has the
same value as \code{x}.

An example is the identity function:
\begin{lstlisting}[name=sec2, numbers=none]
  id x = x
\end{lstlisting}
Using the singleton type, we write
\code{id : (x{:}Any {\to} r{:}\{x\})}, where the input
has type \code{Any} and the output has singleton type \code{x}. This requires
fewer universally quantified type variables than the principal type
\code{id : (\forall a.\,a {\to} a)}. One can verify that
\code{id : (x{:}Any {\to} r{:}\{x\})} is equivalent to
\code{id : (\forall a.\,a {\to} a)}; the former has fewer quantified
variables and is therefore easier to apply during type-checking (requiring
fewer existential instantiations). In Hoare logic notation, this corresponds
to \code{id :: (\lambda\,x\,r.~req~true~ens[r]~r{:}\{x\})}.

As a second example, consider \code{(\lambda x.\,x~x)}, where the parameter \code{x}
is applied to itself. Using singleton types, its function type is
\code{\forall b.\,(x{:}\{x\}{\to}b){\to}b}, which uses one fewer quantified
variable than the \code{\forall a,b.\,(a{\wedge}a{\to}b){\to}b} formulation
of~\cite{Castagna:POPL24}. Fewer quantified variables simplify existential
instantiation for polymorphic method calls. Using our formulation, one can
prove \code{(\lambda x.\,x~x)(\lambda x.\,x~x) : \bot} and
\code{(\lambda x.\,x~x)(\lambda x.\,x) : (x{:}Any {\to} r{:}\{x\})}, both of which are beyond the reach
of~\cite{Castagna:POPL24}. The application \code{(\lambda x.\,x~x)(3)} is
correctly rejected as untypeable by both~\cite{Castagna:POPL24} and our
framework. We currently restrict singleton types to variables and
function-call arguments; field positions use type variables instead.

\subsection{Constructors and Type Predicates}
\label{sec:predicates}
\label{sec:recursive_algebraic}

In statically-typed languages such as OCaml, algebraic data types are
standard. Using the polymorphic list type from Sec.~\ref{sec:intro},
the type annotation for \code{length} is:

\begin{lstlisting}[name=sec2, numbers=none]
  (* Type signature and implementation of length *)
  length : (*@$\forall$@*)T. List(T) (*@\code{\ra}@*) Int
  length(xs) = match xs of
            {Nil -> 0;
             Cons(y,ys) -> 1 + length(ys)}
\end{lstlisting}
Our framework supports types that are more fine-grained than declared
algebraic data types. The type \code{List(T)} is semantically
equivalent to the recursive predicate:
\begin{lstlisting}[name=sec2, numbers=none]
  r: List(T) (*@$\equiv$@*) r:Nil (*@$\vee$@*) r:Cons(T,List(T))
\end{lstlisting}
Under semantic subtyping, \code{Nil} and \code{Cons} are subtypes
of \code{List} and can appear directly in specifications.
A more
precise specification for \code{length} is:
\begin{lstlisting}[name=sec2, numbers=none]
  (* Type specification of length *)
length : (*@$\forall$@*)T. case[xs]
        {Nil  ens[r] r:0;
         Cons(T,List(T))  ens[r] r:Int}
\end{lstlisting}





Here \code{r{:}0} is a singleton type, since
\code{Set(\{0\}){\subset}Set(Int)}. Singleton
return types enable more informative specifications for path-sensitive
code, particularly when values with distinct constructors produce
different outcomes.




Recursive type predicates support algebraic data types in static typing
scenarios and can be extended to GADTs (Sec.~\ref{sec:GADT})
and type classes (Appendix~\ref{sec:typeclasses}).

Although recursive predicates define stronger types, they may appear
incompatible with dynamically-typed programs that rely on runtime type
testing. For complex types such as \code{List(List(T))},
we propose using the unnested form \code{List(\_)} for
runtime testing. Concretely, \code{List(\_)} abbreviates
\code{Nil {\vee} Cons(\_,\_)},
as described next.






\subsection{Combining Type Testing with Pattern-Matching}
\label{sec:pattern-test}

Statically-typed languages use pattern-matching; dynamically-typed
languages rely on type-testing. Recent proposals, including C\#, attempt
to unify both. We achieve this integration via an enhanced \code{match}
construct that supports both pattern-matching and structural type-testing.

The function \code{length\_generic} dispatches the length computation
based on the runtime type of its argument. Its case specification is:




\begin{lstlisting}[name=sec2,numbers=none]
(* Type spec and code of length *)
length_generic : (*@$\forall$@*)T. case[xs]
             {xs:Int (*@${\caseRA}$@*) ens[r] r:1;
              xs:Str (*@${\caseRA}$@*) ens[r] r:Int;
              xs:List(T) (*@${\caseRA}$@*) ens[r] r:Int}
length_generic xs = match xs of
            {Int -> 1;
             Str -> length_of_str(xs);
             List(_) -> length(xs)}
\end{lstlisting}

Recall that \code{List(\_) {\equiv} Nil {\vee} Cons(\_,\_)}. The case
specification with singleton type \code{1} for the \code{Int} branch
is strictly more precise than the intersection type:
\begin{lstlisting}[name=sec2,numbers=none]
  length_generic : ((*@$\forall$@*)T. List(T) (*@\ra@*) Int) (*@$\wedge$@*)
       (Int (*@\ra@*) Int) (*@$\wedge$@*) (Str (*@\ra@*) Int)
\end{lstlisting}

The method 
\code{length} was already defined in
Sec.~\ref{sec:predicates}.

\section{Advanced Types}\label{app:advanced}

\subsection{Support for Type Classes}
\label{sec:typeclasses}
We describe how our framework supports type classes, as adopted by
Haskell and Scala for ad hoc polymorphism. The \code{Num} type class
supports generic numeric operations:

\begin{lstlisting}[name=functor,numbers=none]
   + : (*@$\forall$@*)T (*@$\cdot$@*) Num(T) (*@$\Rightarrow$@*) T(*@$\ra$@*)T(*@$\ra$@*)T
\end{lstlisting}

Using pre/post type specifications, we encode this with a type-class
predicate \code{Num(T)} built from the instance declarations of
\code{Num}:

\begin{lstlisting}[name=functor,numbers=none]
  Num(T:*) (*@$\equiv$@*) (T=Int) (*@$\lor$@*) (T=Float) (*@$\lor$@*) ...
 \end{lstlisting}

This predicate appears as a precondition constraint on overloaded
operations:
\begin{lstlisting}[name=functor,numbers=none]
  + :: (*@$\lambda$@*) x y r (*@$\cdot$@*) (*@$\forall$@*) T (*@$\cdot$@*) req[x,y] x:T (*@$\land$@*) y:T
             (*@$\land$@*) Num(T) ens[r] r:T
\end{lstlisting}

Higher-order type classes are supported similarly. For the
\code{Functor} class:

\begin{lstlisting}[name=functor,numbers=none]
 class Functor(F:*(*@$\ra$@*)*) {
  map : (*@$\forall$@*)A,B (*@$\cdot$@*) (A(*@$\ra$@*)B) (*@$\ra$@*) F(A) (*@$\ra$@*) F(B)
\end{lstlisting}

the specification for \code{map} is:
\begin{lstlisting}[name=functor,numbers=none]
 map :: (*@$\lambda$@*) f xs r (*@$\cdot$@*) (*@$\forall$@*) F:*(*@$\ra$@*)*,A,B (*@$\cdot$@*)
  req[f,xs] f:(A(*@$\ra$@*)B) (*@$\land$@*) xs:F(A) (*@$\land$@*) Functor(F)
  ens[r] r:F(B)
 Functor(F:*(*@$\ra$@*)*) (*@$\equiv$@*) (F=List) (*@$\lor$@*) (F=Option) (*@$\lor$@*) ...
\end{lstlisting}

Type classes may be organised in an inheritance hierarchy. For example,
since \code{Monad} extends \code{Functor}, we have the lemma
\code{\forall F{:}(*{\to}*).~Monad(F) {\Rightarrow} Functor(F)}. Note that type-class predicates such as
\code{Num(\_)} and \code{Functor(\_)}
may not be used for runtime type-testing, since they range over types
rather than values.

\subsection{Gradual Typing for Dynamic Languages}
\label{sec:gradual}

Although our expressive specifications capture many dynamically-typed
programs as statically safe, some examples remain outside the current
framework. Consider:
\begin{lstlisting}[name=sec2, numbers=none]
  f(x,y) = if x>5 then !y else y+1
\end{lstlisting}
For this to be statically type-safe, the precondition would need
\code{y{:}Ref(T){\wedge}Int}, which reduces to \code{y{:}\bot},
making the precondition unsatisfiable. Gradual typing addresses this via
a special family of pure types \code{?t}, allowing runtime casts. We permit
\code{y{:}?(Ref(T){\vee}Int)} in specifications:
\begin{align*}
  &\code{req[x,y]}~
    \code{x{:}Int} \wedge \code{y{:}?(Ref(T){\vee}Int)} \\
  &\code{ens[r]}~
    (\code{y{:}Ref(T)} \wedge \code{r{:}T})
    \vee (\code{y{:}Int} \wedge \code{r{:}Int})
\end{align*}
Runtime casts are inserted into the program body to enforce the chosen
branch:
\begin{lstlisting}[name=sec2, numbers=none]
  f(x,y) = if x>5 then let z = (?Ref(_)) y in !z
                  else let z = (?Int) y in z+1
\end{lstlisting}
Such programs are gradually typed: some execution paths may produce a
runtime abort. The runtime cast \code{(?p)~v} is defined as:
\begin{lstlisting}[name=sec2, numbers=none]
  (?p) v = match v of { p -> v; _ -> abort }
\end{lstlisting}
The Hoare rule for runtime cast ignores abort, treating it as a runtime
error rather than a compile-time error:
\[
\begin{array}{c}
  \incrRule{
    \fresh{r} \qquad \ctr\,{\vdash}\,x{:}?(\_) \quad
    \ctr_1{\equiv}({\ctr}~{\wedge}~(x{:}p)~{\wedge}~r{:}\{x\})
  }{\HTriple{\ctr}{(?p)~x}{\specEns{r}{\ctr_1}}}{R{-}Cast}
\end{array}
\]

\subsection{Restrictions to Support Decidability}
\label{sec:decidable}
    In their full generality, type predicates are not decidable: arbitrary recursive predicates with unrestricted quantification and negation quickly exceed the reach of
automated verification. In practice, decidability can be recovered by imposing the following five restrictions,
which together cover the vast majority of data structure predicates that arise in verification.
\begin{itemize}
\item Structural recursion only. Predicates must unfold only on the direct sub-components of a constructor: no
mutual recursion through non-structural positions. This makes well-foundedness syntactically checkable and
ensures unfolding always terminates.
\item Linear arithmetic for index constraints. Restricting index constraints to linear arithmetic over integers (e.g. \code{n
= m+1, lo {\leq} v}) reduces subtyping obligations to linear arithmetic validity, decidable by standard SMT solvers via
the theory of linear integer arithmetic (LIA) --- our Lean mechanisation decides the \emph{equality} fragment of this
theory with a certified, solver-free in-checker procedure, and, for the mechanised \code{SortedList} flagships'
proof obligations, the \emph{inequality} fragment as well, via a certified unit-Farkas decision procedure
(App.~\ref{sec:leanchecker}), reserving SMT for general linear-arithmetic inequality reasoning beyond those
obligations. Non-linear arithmetic or quantifier alternation would push into undecidability.
\item Restricted existential quantification. Existentials in predicate bodies must be eliminable: either by
unification during unfolding, or by bounding them over finite domains. Unrestricted nested existentials can render
subtyping undecidable.
\item Negation-free or negation at base types only. Boolean combinations \code{\vee} and \code{\wedge} are decidable provided
sub-predicates are decidable. However, negation (\code{\neg}) of recursive predicates introduces co-inductive reasoning and
is generally undecidable. Restricting negation to base types only preserves decidability.
\item Symbolic heap fragment for separation predicates. Full separation logic with arbitrary recursive heap
predicates is undecidable. The symbolic heap fragment recovers decidability by restricting the form of inductive
predicate definitions rather than banning recursion itself: each predicate case must be expressed using only
points-to assertions, separating conjunctions, inductive predicate calls, and pure linear arithmetic side conditions
— with no negation and no non-separating disjunction at the heap level. Recursive inductive predicates such as
\code{\heapto{x}{List(T)}}, \code{\heapto{x}{Tree(T)}}, and list segment predicates are fully expressible within this fragment, and bi-abduction
over such predicates is decidable, as exploited by tools such as Smallfoot.
\end{itemize}

These five restrictions together form the practical sweet spot adopted by most separation logic verifiers: structural
recursion, linear arithmetic indices, eliminable existentials, negation-free recursive predicates, and symbolic heap
separation. This covers \code{List(T,n)}, \code{RBTree(T,c,h)}, \code{SortedList(T,lo,hi)} and the data structure predicates that arise
throughout this paper, while keeping verification tractable via SMT solving. In the Lean mechanisation,
\code{List(T,n)}, \code{RBTree(T,c,h)}, and size-indexed trees (nodes of size $1{+}a{+}b$) are now
checker-certified end-to-end: their index constraints are affine equalities, decided by a certified,
solver-free in-checker procedure (App.~\ref{sec:leanchecker}), with no SMT call --- including the
existential child-size decomposition ($s = 1{+}a{+}b$, with the children's sizes not individually
determined by $s$) that sized-tree destruction requires, previously the one remaining checker
boundary, now crossed by a declared-post frame gate. \code{SortedList(T,lo,hi)} needs the index
\emph{inequality} fragment ($<$/$\le$); this fragment is decided too, for the proof obligations of
three checker-certified flagships over \code{SortedList(T,lo,hi)} (sorted \code{insert},
\code{head}, and \code{member}) by a certified, solver-free unit-Farkas procedure, again with no
SMT call (App.~\ref{sec:leanchecker}). Every data-structure predicate class in this paper now has at
least one checker-certified flagship obligation discharged this way. \emph{General} completeness of
the underlying decision procedures stays future work, in the disclosed, restricted form documented
alongside each engine (App.~\ref{sec:leanchecker}): the affine-equality comparators are sound but not
proved complete beyond the sums the mechanised predicates exercise, and the unit-Farkas inequality
procedure does not derive coefficient-scaled or equality-tightened consequences.




\section{Exception Handling}\label{app:exception}
The \code{\Abrt}
or \code{\top} types are nearly always avoided
in the result type
of our functions, as any encounter with abort will cause our program to fail immediately.
Case specification is
the only type specification which allows
\code{\Abrt} or \code{\top} to appear
in the \code{ens} clauses of the otherwise (i.e. \_) clause,
due to the need for completeness in case specifications.
\label{sec:exception}

Runtime errors are often modelled as exceptions that
could be handled by suitable error recovery routines.
In some languages, such as Java, they are further classified as
{\em checked exceptions} which must be flagged by
the type system or {\em unchecked exceptions}
whose occurrences can be ignored by the type system.

In our solution, we would classify checked exceptions
(denoted by \code{ChExc})
and unchecked exceptions (denoted by \code{UnExc}) as subtypes
of \code{Exc}:\\
\centerline{\code{\bot <: \{\ChExc,\UnExc\} <: \Exc <: \top}}

Additionally,
\code{UnExc} is a subtype of
     valid~types  and
\code{ChExc} is disjoint from
both \code{\Any} and \code{\Abrt}:

\centerline{\code{{\UnExc} <: \{..valid~types..\} \wedge }}
\centerline{ \code{  ((\ChExc \wedge \Any) = (\ChExc \wedge \Abrt) = \bot) }}



In this way, checked exceptions will always be carefully tracked by
our type system, while unchecked exceptions are considered as a subtype of both type
universes, \code{Any_P} and \code{Any_S}, and can be 
omitted in our type specifications.

Try-catch handling can now be implemented via a \code{let}-binding,
followed by type testing.
\hide{Also, strict evaluation of expressions
are always type-casted to ensure
that it is of type \code{{\AnyVal} = {\Any} \vee {\Exc}}.}
Raised exceptions that pass type-casting
are propagated outwards via
\code{match} construct, as illustrated by the second translation rule below.

\[
\arraycolsep=0pt
\begin{array}{rll}
&\coden{\etrycatch{e_1}{\Exc~\texttt{->}~e_2}} \\
 \hspace*{1em}\equivArr\hspace*{1em}
& \coden{\elet{x}{e_1}{\ematch{x}{\Exc~\texttt{->}~e_2{\vbar}\_~\texttt{->}~x}}} \\

\\


&\elet{x{:}t}{e_1}{e_2} \\
 \hspace*{1em}\equivArr\hspace*{1em}
& \elet{x}{(t{\ \vee\ }{\Exc})\ e_1}{}\\
&{\ematch{x}{\Abrt{\ \vee\ }\Exc~\texttt{->}~x{\vbar}\_~\texttt{->}~e_2}}
\end{array}
\]

{\wnnay{Can we really pattern-match on Exc and Abrt?
  Is this an efficient way to implement exception-handling
or do we need native continuations?}}





\section{More on Separation Types}\label{app:sepmore}

\subsection{Uniqueness vs Separation Type}
\label{sec:unique}

Each uniqueness type \code{\uniqto{x}{T}} is mutually disjoint from both
the separation type \code{\heapto{x}{T}} and the pure type \code{x{:}T}, as the
following contradictions show (\code{\equivA} denotes logical
equivalence):
\[
  \code{\uniqto{x}{T}} ~\land~ \code{x{:}T}
  ~\equivA~
  \code{\heapto{x}{T}} ~\land~ \code{x{:}T}
  ~\equivA~
  \code{\heapto{x}{T}} \sep \code{\uniqto{x}{T}}
  ~\equivA~
  \code{\false}
\]
This design choice ensures each type form has distinct capabilities.
Uniqueness and separation types are interconvertible via the following
equivalences:
\[
\begin{array}{lll}
  \code{\uniqto{x}{T}} & \equivA & \code{\heapto{x}{T}} ~\land~ \code{NoAlias(x)} \\
  \code{NoAlias(x)} ~\land~ \code{x{=}y} & \equivA & \code{\false}
\end{array}
\]
Every uniqueness type converts to a separation type; each separation
type converts to a uniqueness type by dropping all aliases of the form
\code{x{=}y}.

As type-safety is based on over-approximation, the following weakening
lemmas apply:
\[
\begin{array}{lll}
  \code{\uniqto{x}{T}} ~\land~ \code{T{<:}A} & \weaken & \code{\uniqto{x}{A}} \\
  \code{\heapto{x}{T}} ~\land~ \code{T{<:}A} & \weaken & \code{\heapto{x}{A}} \\
  \code{\heapto{x}{T}} & \weaken & \code{x{:}T}
\end{array}
\]
Types are always initialised as strong as possible and weakened only
when necessary. Separation and uniqueness types are assigned at
construction and may be irreversibly weakened to flow-insensitive pure
types (last lemma above).

Both allow type mutation for flow-sensitivity; however, separation
types additionally permit must-aliases, whereas uniqueness types allow
only a single reference, making separation types strictly more
expressive. For example, uniqueness types cannot handle
\code{swap(m,m)} since the two arguments must be unaliased.


\subsection{Separation Type for Mutable Data Structure}
\label{sec:dsi}

For pure types, inductive predicates capture the weakest property
guaranteed for a given algebraic data type. For \code{List(T)}:
\[
  \code{r{:}List(T)} ~\equivA~
    \code{r{:}Nil}
    ~\vee~
    (\exists q.~ \code{r{:}Cons(T,q)} \wedge \code{q{:}List(T)})
\]
This predicate is stable for every pure \code{List(T)} value.
Both acyclic and cyclic lists are instances:
\begin{align*}
  \code{x{:}Cons(T,y)} \wedge \code{y{:}Cons(T,z)} \wedge \code{z{:}Nil}
    &~\vdash~ \code{x{:}List(T)} \\
  \code{x{:}Cons(T,y)} \wedge \code{y{:}Cons(T,z)} \wedge \code{z{:}\{x\}}
    &~\vdash~ \code{x{:}List(T)}
\end{align*}

Using separation types, acyclic and cyclic lists can be distinguished
by the following separation predicates:
\begin{align*}
  \code{\heapto{x}{AList(T)}}
    &~\equivA~
      \code{x{:}Nil}
      ~\vee~ \exists q.~
        \code{\heapto{x}{Cons(T,q)}} \sep \code{\heapto{q}{AList(T)}} \\
  \code{\heapto{x}{CList(T)}}
    &~\equivA~
      \code{\heapto{x}{Cons(T,q)}} \sep \code{\heapto{q}{LSeg(T,x)}} \\
  \code{\heapto{x}{LSeg(T,p)}}
    &~\equivA~
      \code{x{:}\{p\}}
      ~\vee~ \exists q.~
        \code{\heapto{x}{Cons(T,q)}} \sep \code{\heapto{q}{LSeg(T,p)}}
        \wedge \code{x{:}\neg\{p\}}
\end{align*}
The list-segment predicate \code{\heapto{x}{LSeg(T,p)}} captures an
incoming pointer \code{x} and an outgoing pointer \code{p}, enabling circular
list support. These predicates yield the following provable and
disprovable assertions:
\begin{align*}
  \code{\heapto{x}{Cons(T,y)}} \sep \code{\heapto{y}{Cons(T,z)}} \wedge \code{z{:}Nil}
    &~\vdash~  \code{\heapto{x}{AList(T)}} \\
  \code{\heapto{x}{Cons(T,y)}} \sep \code{\heapto{y}{Cons(T,z)}} \wedge \code{z{:}\{x\}}
    &~\vdash~  \code{\heapto{x}{CList(T)}} \\
  \code{\heapto{x}{Cons(T,y)}} \sep \code{\heapto{y}{Cons(T,z)}} \wedge \code{z{:}Nil}
    &~\nvdash~ \code{\heapto{x}{CList(T)}} \\
  \code{\heapto{x}{Cons(T,y)}} \sep \code{\heapto{y}{Cons(T,z)}} \wedge \code{z{:}\{x\}}
    &~\nvdash~ \code{\heapto{x}{AList(T)}}
\end{align*}

Separation predicates are flow-sensitive and express data structure
invariants that type-safe verification must maintain. We propose
guaranteeing these invariants via pre/post conditions expressed with
appropriate separation predicates.


\subsection{Separation Type for Data Race Freedom}
\label{sec:drf}

A concurrent program has a data race if two threads access the same
shared data, at least one access is a write, and there is no
synchronisation ordering between them. One way to ensure race freedom
is to guarantee disjoint memory access, expressed via separation types
at field-level granularity.

Consider a mutable \code{Pair(X,Y)} with fields of types \code{X}
and \code{Y} and the following two concurrent threads:
\begin{lstlisting}[name=sec2, numbers=none]
 p.Pair.1 := "hi"  ||  p.Pair.2 := p.Pair.2+1
\end{lstlisting}
Although both threads write to \code{p}, they access disjoint fields. The
following field-splitting lemma supports this:
\[
  \code{\heapto{p}{Pair(X,Y)}}
  ~\equivA~
  \code{\heapto{p.Pair.1}{X}} \sep \code{\heapto{p.Pair.2}{Y}}
\]
A pair \code{\heapto{p}{Pair(X,Y)}} can be split into disjoint
field ownership
\code{\heapto{p.Pair.1}{X}} \sep \code{\heapto{p.Pair.2}{Y}}.
Type-safe verification then proceeds as follows:
\begin{lstlisting}[name=sec2, numbers=none]
 // p->Pair(Any,Int)
 // p.Pair.1->Any  *  p.Pair.2->Int
 p.Pair.1 := "hi"  ||  p.Pair.2 := p.Pair.2+1
 // p.Pair.1->Str  *  p.Pair.2->Int
 // p->Pair(Str,Int)
\end{lstlisting}

To ensure data race freedom, only separation types (not pure types)
may be used in concurrent threads unless the data structure is
immutable: pure types permit arbitrary aliases that other threads
could mutate independently. When concurrent threads require write
access to shared heap memory, the framework can be extended with a
locking mechanism over separation-typed data structures.


\section{Operational Semantics}\label{app:opsem}
\label{sec:bigstep}


To facilitate the following soundness proofs, we define a big-step reduction relation with judgments \bigstep{e}{h}{S}{\Rese}{h_1}{S_1}.
Program states consist of a heap $h$ and store $S$, like in Sec.~\ref{sec:Soundness}.
Outcomes are $\Rese ::= Norm(\code{v}) \mid \code{Abrt}$
where $Norm(\code{v})$ captures the valid values and \code{Err} values which could occur during the run-time but do not
terminate the programs. The $\code{Abrt}$ represents a failure which will immediately kill the process.
In the $Match$ case, we assume there is an implicit case that handles a match failure. If the process cannot find
a case to match, it will terminate the program with $\code{Abrt}$.

\[
\begin{array}{cc}
\incrRule{
}{
  \bigstep{\code{v}}{h}{S}{\code{v}}{h}{S}
}{Nil, Const, Lambda}
\end{array}
\]
\[
\begin{array}{cc}
\incrRule{
}{
  \bigstep{\econs{\code{x_1}}{\code{x_2}}}{h}{S}{\econs{S(\code{x_1})}{S(\code{x_2})}}{h}{S}
}{Cons}
\end{array}
\]

\[
\begin{array}{cc}
\incrRule{
}{
  \bigstep{\code{x}}{h}{S}{S(\code{x})}{h}{S}
}{Var}
\end{array}
\]
\[
\begin{array}{cc}
\incrRule{
  type(S(\code{x})) <: \code{t}
}{
  \bigstep{(\code{t}~(\code{x}))}{h}{S}{\code{x}}{h}{S}
}{Cast}
&
\end{array}
\]

\[
\begin{array}{cc}
\incrRule{
  type(S(\code{x})) \not <: \code{t}
}{
  \bigstep{(\code{t}~(\code{x}))}{h}{S}{\code{Abrt}}{h}{S}
}{Cast}
&
\end{array}
\]

\incrHRule{
  \bigstep{\code{e_1}}{h}{S}{\code{v}}{h_1}{S_1}\qquad
  \bigstep{\code{e_2}}{h_1}{S_1[\code{x}\heapstoreupd\code{v}]}{\code{v_1}}{h_2}{S_2}
}{
  \bigstep{(\elet{\code{x}}{\code{e_1}}{\code{e_2}})}{h}{S}{\code{v_1}}{h_2}{S_2{\setminus}\{\code{x}\}}
}{Let}



\[
\begin{array}{cc}
\incrRule{
  S(\code{f}) = \efun{\code{x^*}}{\Phi}{\code{e}}\qquad S(\code{x^*}) = \code{v^*}\\
  \bigstep{\subst{\code{x^*}}{\code{v^*}}{\code{e}}}{h}{S}{\code{v_2}}{h_1}{S_1}
}{
  \bigstep{\ecall{\code{f}}{\code{x^*}}}{h}{S}{\code{v_2}}{h_1}{S_1}
}{App}
\end{array}
\]
\[
\begin{array}{cc}
\incrRule{
  \code{x_1} \in \heapdom{h}
}{
  \bigstep{(\eassign{\code{x_1}}{\code{x_2}})}{h}{S}{\code{()}}{\heapupdate{h}{S(\code{x_1})}{S(\code{x_2})}}{S}
}{Assign}
\end{array}
\]

\[
\begin{array}{cc}
\incrRule{
}{
  \bigstep{\ederef{\code{x}}}{h}{S}{h(S(\code{x}))}{h}{S}
}{Deref}
\end{array}
\]
\[
\begin{array}{cc}
\incrRule{
  \loc \notin \heapdom{h}
}{
  \bigstep{\eref{\code{x}}}{h}{S}{\loc}{\heapupdate{h}{\loc}{S(\code{x})}}{S}
}{Ref}
\end{array}
\]


\[
\begin{array}{cc}
\incrRule{
  S(\code{x}){\subseteq}\code{p_i}\qquad
  \bigstep{\code{e_i}}{h}{S}{\code{v}}{h_i}{S_i}
}{
  \bigstep{(\ematch{\code{x}}{\code{p_1}{\ra}\code{e_1}{{\mSep}\cdots{\mSep}}\code{p_n}{\ra}\code{e_n}})}{h}{S}{\code{v}}{h_i}{S_i}
}{Match}
\end{array}
\]


\section{The Lean Type-Checker}
\label{sec:leanchecker}
The Lean development additionally yields a \emph{reflective, self-certifying type-checker}.
Where the derivations of Sec.~\ref{sec:lean} are built by hand (or by the small \code{forward}
tactic), the checker turns the meta-theory into an \emph{executable} decision procedure and lets
each program discharge its own typing derivation by computation.

\subsection{A checker proved sound once}
For each layer we implement the forward rules as computable functions: an algorithmic subtyping
\code{subTyAlg}, a state-entailment test \code{entailsAlg}, and a recursive \code{checkExpr} over
expressions (\code{checkSep} in the separation layer M2; \code{checkHO} in the higher-order layer
M3), and prove \emph{once} that each is sound against the inductive derivation relation
\code{Deriv}:
\[
  \code{checkExpr}\ \Gamma\ \Delta\ e\ r = \code{some}\ \Delta'
  \ \Longrightarrow\
  \code{Deriv}\ \Gamma\ \Delta\ e\ r\ \Delta'.
\]
Each soundness theorem is itself axiom-clean
($\{\textsf{propext},\textsf{Classical.choice},\textsf{Quot.sound}\}$). Given it, an individual
program's typing derivation follows purely by \emph{computation}, with no per-example proof. The
checker is deliberately \emph{sound but incomplete}: it \emph{checks} a program against a
supplied specification rather than \emph{inferring} one. Its architecture is sketched next,
and the machinery behind the red-black, heap-atom, and packed-precondition developments is
detailed in Appendix~\ref{sec:leanapp}.

\subsection{Checker architecture}
Three design choices make the checker both executable and extensible.

\emph{Law registries.} The registry-based type predicates of Sec.~\ref{sec:lean} are
consumed by the checker through per-predicate \emph{law rows}: a fold row (how a
constructor application folds back into the predicate), cover rows (which constructors a
predicate can exhibit), and disjointness rows. Extending coverage to a new data structure
means registering its laws; the checker itself, and its soundness proof, are
untouched. The red-black development registers three such predicates; the sized-list and
sized-tree predicates of the arithmetic engine below are two more.

\emph{Disjunctive states with certified pruning.} Checker states are disjunctions of
conjunctive typing contexts. A \code{match} splits the state per pattern package,
entailments are checked per disjunct, and infeasible disjuncts are discharged by
\emph{covers-based refutation}: a certified unsatisfiability test prunes branches whose
recorded facts contradict the registered cover laws. This is the mechanism behind
path-sensitive matches and the red-black colour analysis.

\emph{Heap-atom entailment and unfolding} (the separation layer). State entailment over
heaps is a goal-directed \emph{atom-cover} procedure: each goal atom must be covered by a
premise atom under a small set of certified rows (identity, alias redirection, weakening,
and per-predicate fold rows); leftover premise atoms are framed away. Packed
preconditions are handled by registered \emph{unfold templates}, justified once per
predicate by a store-extension lemma.

\emph{Quantified families by reflection.} Three wrappers turn one symbolic checker run
into a theorem quantified over a family: element-type polymorphism ($\forall T$, the
\code{polyReduce} symbolic-reduction macro), successor height indices ($\forall h$, the
sorted quantifier), and affine index valuations ($\forall\rho$, the \code{allIdx} wrapper
of the arithmetic engine below).

\subsection{Two discharge modes}
Typing derivations are obtained by \emph{computation} in one of two modes. (i)~At a \emph{ground}
instantiation, \code{checkExpr\_reflect (by native\_decide)} runs the compiled checker and lifts its
Boolean result through the soundness theorem, with \emph{no per-example lemmas}. The one price is
that \code{native\_decide} adds a compiler-reduction axiom to \emph{that example's} footprint ---
the only departure from the development's otherwise minimal axiom base; the checkers' soundness
proofs themselves do not use it. (ii)~Results quantified over \emph{open} types ($\forall T$) cannot
use \code{native\_decide}, which needs a closed term; they are discharged by \emph{symbolic}
reduction of the checker (the \code{polyReduce} macro), feeding the per-type-variable leaf facts as
hypotheses. These polymorphic theorems are therefore \emph{native-free} and fully axiom-clean.
The two modes differ in cost by orders of magnitude, for a reason complementary to their trust.
\code{native\_decide} \emph{compiles} the checker and runs it as machine code on packed runtime data
(with in-place update and machine integers), taking exactly one branch at each ground \code{match}.
Symbolic mode instead reduces the checker by the kernel's own definitional computation over
\code{Expr} terms: unary/constructor-encoded data, copying substitution, no sharing, and, because
the type or height variable stays \emph{open}, \code{simp}-driven case-splitting rather than a single
concrete branch. The resulting proof is then re-checked by the kernel. Millisecond ground runs thus
become multi-minute symbolic ones; the price buys a minimal, oracle-free trusted base.

\subsection{Certified linear arithmetic for sized types}
\label{sec:checkerlia}
Beyond the successor fragment used by the red-black development, the checker carries a
\emph{general affine-equality} linear-arithmetic engine, realising the equality part
of the decidable fragment outlined in App.~\ref{sec:decidable} with no SMT oracle
(consistent with Sec.~\ref{sec:intro}). Index expressions are \emph{reified} into a
canonical form (an integer constant plus a sorted list of variable coefficients), and
index equalities and disequalities are decided by a single \emph{sound} Boolean procedure
(\code{idxEqB}/\code{idxNeqB}; the equality test's soundness proof is axiom-free, and
completeness is deliberately not claimed). Sized predicates ($List(T,n)$, and
size-indexed trees with $1{+}a{+}b$ at each node) are registered as law rows over this
engine, and $\forall$-index families are certified by one Boolean run through the
\code{allIdx} reflect wrapper.
The ``one Boolean run per $\forall$-index family'' description is exactly true
of this call-site reflect wrapper; the flagships' own well-formedness $\forall\rho$ obligation
instead closes via a small \code{by\_cases} family (3 cases) of symbolic checker runs, not a single
uniform run. Two further flagships are checker-certified end-to-end on this engine, both native-free and
axiom-clean: sized \code{append} ($List(T,n)\times List(T,m) \Rightarrow List(T,n{+}m)$, for all integers $n,m$)
and exact-length \code{length} ($List(T,n) \Rightarrow \{n\}$). In each, the sized-list precondition is
uninhabited at $n<0$, so the semantic content lives at $n,m\ge 0$; \code{append} is witnessed at $n=2,m=1$.
Their sized-list entailments run through per-disjunct certifiers such as \code{slMemOkB}, rather
than one uniform \code{entailsAlg} run. A sized-\emph{tree} node \code{cons(l,r) : ST(T,\{s\})}
decomposes \emph{existentially} as $s = 1{+}a{+}b$, with the children's sizes $a,b$ not determined
by $s$ --- a decomposition no single-pair \code{match} field slot can express directly. A
\emph{declared-post frame gate} (\code{checkSTMatchTo}, paralleling the \code{SortedList} gate
below) crosses this wall instead, and two further flagships are checker-certified end-to-end on
it, both native-free and axiom-clean: \code{size} (double-rebound destruction through both
children; the result \emph{is} the size itself) and \code{mirror} (destruction \emph{and}
reconstruction --- the rebuilt node's commuted child order $1{+}b{+}a$ is reconciled against the
destructed $1{+}a{+}b$ by a dedicated semantic bridge, \code{mirror\_consEntails}). This closes the
sized-tree case; the semantic-only certificate from an earlier iteration (\code{streeAt\_mirrorV},
proved directly over the tree's value-level semantics) survives unchanged, as a complementary
account. Table~\ref{tab:lean-checker} records both new flagships.

\subsection{Certified linear arithmetic for the inequality fragment}
\label{sec:checkerlia-ineq}
On top of the affine-equality engine, the checker also carries a certified \emph{inequality}
decision procedure. Index expressions are compared by \code{idxLeB}, a sound unit-Farkas
worklist procedure that saturates hypothesis differences by unit-coefficient subtraction and
accepts once some saturated difference closes to a nonnegative constant; its soundness proof
is clean (the classical axiom base, rather than axiom-free, since it also needs genuine
\code{Int} order reasoning). Completeness is deliberately \emph{not} claimed: coefficient
scaling (e.g.\ deriving $2x{\le}a{+}b$ from $x{\le}a \wedge x{\le}b$) and equality-tightening
lie outside the unit-Farkas fragment. A closed-only assertion atom \code{idxLe} mirrors the
existing index-singleton semantics, and a comparison-derivation refinement of \code{<} on
index-typed operands (\code{Deriv.binopLtIdx}) packages each branch's guard into a
negation-normalised \code{idxLe} atom feeding this engine.

On this footing the checker certifies a \code{SortedList(T,lo,hi)} predicate (a structural
fold on the tail whose lower bound is rebound at each matched element), together with a
bespoke $\forall$-element destruction rule (\code{Deriv.matchSortedCons}) that universally
quantifies the \code{cons} premise over the matched element's value, closing every emitted
index atom at that value. Three flagships are checker-certified end-to-end on it, all
native-free and axiom-clean: \code{head} (the shared post collapses to a scalar
range-validity atom, so the per-element bound is instead certified in the destruction rule's
own side-condition witness; partial correctness on the \code{cons} case --- the \code{nil}
arm is deliberately evaluation-stuck, so nothing is claimed on empty input), \code{member}
(an honest \code{Bool} post via rebound recursion
and a double-\code{<} comparison encoding), and, as the campaign's capstone, sorted
\code{insert} ($(Int,\ SortedList(Int,lo,hi)) \Rightarrow SortedList(Int,lo,hi)$, a genuine
end-to-end sorted insertion). As with the sized-list rows, the vacuity corners are stated
honestly: \code{insert}'s closed precondition is uninhabited outside $lo{\le}x{\le}hi$ and
\code{head}'s exactly at $lo{>}hi$ (the semantic content lives inside the range, witnessed at
concrete values), while \code{member}'s precondition carries no range atom and is satisfiable
broadly. Each flagship's $\forall$-element \code{cons}-branch premise is discharged by a
named semantic side condition (\code{SoConsObligation}) rather than by the structural checker
gate directly: the declared-post frame gate used by construction flagships such as
\code{insert} deliberately does not re-run the \code{cons} branch at its reserved index
variable, an honest, disclosed checker incompleteness that costs nothing, since the real
premise is supplied by the side condition instead. As with the equality engine, no
completeness claim is made anywhere, for the engine or for the checker.
Table~\ref{tab:lean-checker} lists all three flagships under the inequality-fragment group.

With the two sized-tree flagships above, the checker now discharges every index-language
construct the paper's data-structure predicates use: \textbf{there is no remaining documented
index-language boundary}. What remains are limits on the \emph{completeness} of the certified
procedures themselves, not on coverage. The comparators \code{idxEqB} (affine equality) and
\code{idxLeB} (affine inequality) are both sound but not claimed complete: \code{idxLeB} decides
only the unit-coefficient Farkas fragment, and coefficient scaling or equality-tightening lie
outside it (App.~\ref{sec:decidable}). The declared-post frame gates \code{checkSTMatchTo} and
\code{checkSortedMatchTo} each deliberately omit re-running their $\forall$-child/$\forall$-element
destruction premise at its own reserved index variable, discharging it instead through a named
semantic side condition --- an honest, disclosed incompleteness that costs nothing, since the real
premise is supplied that way regardless. The sized-tree predicate's element type, like
\code{SortedList}'s, is carried shape-only (semantically inert). And every flagship's vacuity
corner is disclosed per instance: both \code{size} and \code{mirror} are vacuous at $s{<}0$
(formalised as \code{mirrorPre\_vacuous}).

\section{Details of Lean Machinery}
\label{sec:leanapp}
\paragraph{Proof size and checking time.}
Table~\ref{tab:lean-examples} reports, for each machine-checked example, the size of its
Lean script (the program, its type specification, and the proof, in lines) and the time
Lean takes to elaborate and kernel-check it. These are \emph{proof-assistant} elaboration
times for the meta-level derivations, and are mostly small: the bulk of the examples check in under a second
(tens to a few hundred milliseconds); the heaviest are the typed-queue pair
(${\approx}4.9$\,s and ${\approx}3.7$\,s) and the five-function Okasaki
\code{RBTree\_insert}, whose full colour/black-height derivation takes about $15$\,s,
bringing the whole suite to about $39$\,s. The proof-script
size is the more informative figure, reflecting the effort to construct each typing
derivation and instantiate soundness; it is a one-time meta-theoretic cost, not paid again
when the \toolname{} verifier checks user programs. The M1 proofs are written with a small
forward-reasoning tactic (\texttt{forward}/\texttt{ml\_side}/\texttt{ml\_entail}) that
applies the Hoare rules and discharges the decidable side conditions automatically: it
makes the scripts more compact at a modest cost in elaboration time, since the generic
tactics perform broader \texttt{simp}/\texttt{decide} search than the hand-written proofs they
replace. The suite spans the paper's worked examples (path-sensitive accessors and
type-dispatch of Sec.~\ref{sec:motivate}, function composition) together with a
representative selection of standard benchmarks and standard-library functions.
These include type-dispatch
(\code{shape\_funs}), list construction (\code{List.append}, \code{List.rev}), a
colored red-black tree with a full balancing predicate (both a recursive search,
\code{RBTree\_member}, and a verified Okasaki insert, \code{RBTree\_insert}, over
\code{RBTree(Int,c,h)}), reference \code{swap} with
must-aliasing, a separation-typed \code{Queue}, and the higher-order standard-library
functions \code{List.map}, \code{List.filter}, and \code{List.fold\_left}. Adding
\code{swap} and the operational \code{Queue} demo to the separation layer required no new
rule, while the fully \emph{typed} \code{Queue.enqueue}/\code{Queue.dequeue} rows are enabled
by the inductive heap-predicate layer and exact-content strong-update rules described below;
\code{append}/\code{rev},
which \emph{construct} lists, are enabled by a data-constructor rule added to the pure layer.
The \code{map}/\code{filter}/\code{fold\_left} examples exercise a genuine \emph{integration}.
The higher-order layer, which already has closures, application, and the behavioural function
type, is extended with the recursive \code{List(T)} predicate (its interpretation, a
data-constructor typing rule, and \code{match} field-type propagation). A pure
function \emph{parameter} may then be applied elementwise across a list under a recursive
specification: e.g.\ \code{map} at type
$(A{\to}B,\ List(A)){\Rightarrow}List(B)$ and \code{fold\_left} at
$((Acc,A){\to}Acc,\ Acc,\ List(A)){\Rightarrow}Acc$.
The suite's newest addition is a general \emph{linear-arithmetic} culminating example. Two
sized-list flagships --- \code{append} ($List(T,n)\times List(T,m) \Rightarrow List(T,n{+}m)$, for
all integers $n,m$; at $n<0$ the sized-list precondition is uninhabited, so the semantic content
lives at $n,m\ge 0$, witnessed at $n=2,m=1$) and exact-length \code{length}
($List(T,n) \Rightarrow \{n\}$, same vacuity at $n<0$) --- are checker-certified end-to-end on a
general affine-equality index engine (Sec.~\ref{sec:decidable}). Two sized-tree flagships,
\code{size} (destruction through both children; the result \emph{is} the size) and \code{mirror}
(destruction \emph{and} reconstruction, the commuted child order reconciled by a dedicated
semantic bridge), are checker-certified end-to-end too, via a declared-post frame gate that
crosses the existential child-size decomposition a single-pair \code{match} cannot express ---
\code{mirror} upgraded from the semantic-only certificate of an earlier iteration, and there is
no remaining documented index-language boundary. A companion \emph{inequality-fragment} engine (a
certified unit-Farkas decision procedure) checker-certifies a \code{SortedList(T,lo,hi)} predicate and,
through a bespoke $\forall$-element destruction rule, three further flagships end-to-end:
\code{head}, \code{member}, and, as the campaign's capstone, sorted \code{insert}
($(Int,\ SortedList(Int,lo,hi)) \Rightarrow SortedList(Int,lo,hi)$).

{
\small
\renewcommand{\arraystretch}{0.95}
\setlength{\tabcolsep}{5pt}
\begin{xltabular}{\linewidth}{p{3cm} X r r}
\caption{\label{tab:lean-examples} Lean examples: proof-script size
(program\,+\,spec\,+\,proof, in lines) and elaboration time (ms).
$^\dagger$~composition-level figure (see below).}\\
\Xhline{2\arrayrulewidth}
Example & Feature exercised & Proof (lines) & Lean (ms) \\
\Xhline{2\arrayrulewidth}
\endfirsthead
\Xhline{2\arrayrulewidth}
Example & Feature exercised & Proof (lines) & Lean (ms) \\
\Xhline{2\arrayrulewidth}
\endhead
\multicolumn{4}{l}{\emph{M1 --- pure fragment}} \\
\code{id} & result / singleton type & 61 & 85 \\
\code{tail} & path-sensitivity, \Err{} on \code{Nil} & 241 & 328 \\
\code{length} & recursion & 201 & 227 \\
\code{fst} & \code{Cons} field-type propagation & 257 & 350 \\
\code{length} (on \code{List(Int)}) & recursive \code{List} predicate & 271 & 310 \\
\code{id} ($\forall T.\,T{\to}T$) & type-variable polymorphism & 98 & 163 \\
\code{length} ($\forall T.\,List(T){\Rightarrow}Int$) & polymorphic recursion & 260 & 475 \\
\code{treeSize} & user-defined predicate (\code{Tree}) & 265 & 370 \\
\code{constFst} & multi-argument function & 104 & 137 \\
\code{head} & path-sensitive accessor, \Err{} on \code{Nil} & 210 & 381 \\
\code{shape} & type-dispatch, union result \code{Int}$\vee$\code{Str} & 232 & 416 \\
\code{append} & list construction + recursion & 283 & 640 \\
\code{rev} & accumulator recursion + construction & 277 & 614 \\
\code{member} & colored-shape tree predicate demo (registry) & 262 & 383 \\
\hline
\multicolumn{4}{l}{\emph{M1 $+$ integer arithmetic}} \\
\code{inc} & base-type arrow, \code{Int} closure & 127 & 206 \\
\code{plus} & binary \code{Int} operator & 93 & 153 \\
\code{List.nth} & polymorphic path-sensitive \Err{} case spec & 762 & 1316 \\
\code{str\_of\_int} & digit-match, \code{div}/\code{mod}, string concat, recursion & 558 & 1192 \\
\hline
\multicolumn{4}{l}{\emph{M1 $+$ balancing predicate}} \\
\code{RBTree\_member} & recursive search over balanced \code{RBTree(Int,c,h)} & 960 & 1019 \\
\code{RBTree\_insert} & Okasaki insert; $\forall h$ spec states result colour and black-height & 5090 & 15034 \\
\hline
\multicolumn{4}{l}{\emph{M2 --- separation / \code{Ref} heaps}} \\
\code{assign} (strong update) + Frame & flow-sensitive mutation, Frame & 255 & 105 \\
\code{Frame} across a call & callee-footprint Frame & 213 & 169 \\
\code{Frame} across a \code{match} & Frame over \code{match} & 199 & 132 \\
\code{weakening} \heapto{x}{t}${\Rightarrow}$\code{x{:}t} & irreversible weakening & 62 & 20 \\
\code{g(x,y){=}mkref\,x} & multi-argument heap function & 163 & 193 \\
\code{swap} & \code{Ref} exchange, must-alias, framed update & 247 & 142 \\
\code{Queue} & separation, front/rear must-alias (operational demo) & 147 & 134 \\
\hline
\multicolumn{4}{l}{\emph{M2 $+$ inductive heap predicates}} \\
\code{Queue.enqueue} & typed two-pointer enqueue; heap-content strong update & 907 & 4862 \\
\code{Queue.dequeue} & \Err{} on empty via must-alias; element typing via eq-slot match & 1162 & 3722 \\
\hline
\multicolumn{4}{l}{\emph{M3 --- higher-order}} \\
\code{apply} & higher-order application & 85 & 116 \\
$\lambda x.\,y$ & environment-capturing closure & 114 & 132 \\
\code{pickFst}$(x,y)$ & multi-argument top-level function & 132 & 191 \\
$\lambda(x,y).\,x$ & multi-argument closure & 102 & 144 \\
\code{compose} & function composition $g{\circ}f$ ($Int{\to}Str{\to}Bool$) & 192 & 170 \\
\hline
\multicolumn{4}{l}{\emph{M3 $+$ recursive \code{List(T)} predicate}} \\
\code{map} & HO map, list construction + recursion & 365 & 799 \\
\code{fold\_left} & HO fold, $n$-ary function parameter & 344 & 588 \\
\code{filter} & HO filter, nested boolean match & 456 & 856 \\
\code{append} & list concatenation over \code{List(T)} & 292 & 456 \\
\code{rev} & accumulator reversal over \code{List(T)} & 301 & 442 \\
\hline
\multicolumn{4}{l}{\emph{M3 $+$ \code{Option} constructors}} \\
\code{Option.get} & path-sensitive case spec, \Err{} on \code{None} & 297 & 313 \\
\code{Option.map} & HO map over \code{Option(T)}, two-case spec & 372 & 474 \\
\code{Option.bind} & HO bind, \code{Option(T)}-returning parameter & 341 & 404 \\
\hline
\multicolumn{4}{l}{\emph{M1 $+$ sized predicates (LIA)}} \\
\code{append} & sized-list append, $List(T,n)\times List(T,m) \Rightarrow List(T,n{+}m)$, all integers ($n{<}0$ vacuous pre, content at $n,m{\ge}0$), general affine-equality engine & 252 & 87$^\dagger$ \\
\code{length} & exact length, $List(T,n) \Rightarrow \{n\}$, general affine-equality engine & 244 & 77$^\dagger$ \\
\code{size} & sized-tree destruction, $ST(T,\{s\}) \Rightarrow \{s\}$, double-rebound recursion through both children, result IS the size (checker-certified) & 351 & 54$^\dagger$ \\
\code{mirror} & sized-tree destruction $+$ reconstruction, size $1{+}a{+}b$ preserved $\forall s$ under the commuted reassembly (checker-certified; upgraded from a semantic-only certificate) & 406 & 55$^\dagger$ \\
\hline
\multicolumn{4}{l}{\emph{M1 $+$ \code{SortedList} (inequality fragment)}} \\
\code{head} & \code{SortedList} destruction; per-element bound certified in the $\forall$-element side-condition witness & 180 & 10$^\dagger$ \\
\code{member} & comparison traversal, rebound recursion, honest \code{Bool} post & 302 & 197$^\dagger$ \\
\code{insert} & end-to-end sorted insertion, post $SortedList(Int,lo,hi)$ (vacuous pre outside $lo{\le}x{\le}hi$) & 986 & 62$^\dagger$ \\
\Xhline{2\arrayrulewidth}
\textbf{Total} & & \textbf{20081} & \textbf{38905} \\
\Xhline{2\arrayrulewidth}
\end{xltabular}
}

\emph{Measurement.} Times are Lean's internal elaboration timings (Lean~4.30.0,
Intel~Xeon~E-2278G), excluding the fixed per-file import overhead; the typed
\code{Queue.enqueue}/\code{Queue.dequeue} rows use the same single-threaded profiler
(async elaboration disabled, per-entry sum) as App.~\ref{sec:leanchecker}, while the earlier
\code{Queue} row is the operational demo. The \code{RBTree\_insert}/\code{RBTree\_member}
rows cover the full colour/black-height balancing development (a five-function Okasaki
insert; the shared \code{RBTree(Int,c,h)} family is counted with \code{RBTree\_member}).
$^\dagger$~The append/length/size/mirror and \code{SortedList} rows report a composition-level
figure: the elaboration
time of the flagship's own top-level theorem(s), composing cached per-branch
\code{bodyChecks}/obligation lemmas rather than a full from-scratch proof-script; the dominant
cost is the checker suite's own elaboration (App.~\ref{sec:leanchecker}).
The \code{size} and \code{mirror} rows are both checker-certified end-to-end
(App.~\ref{sec:leanchecker}); \code{mirror} upgrades the semantic-only certificate of an earlier
iteration, which survives unchanged as a complementary account (\code{streeAt\_mirrorV}).

\paragraph{Coverage and modeling boundaries.}
Table~\ref{tab:lean-coverage} maps a suite of standard benchmarks and standard-library
functions to their status in the
mechanisation. The GADT-style balancing example of Sec.~\ref{sec:pred} is now mechanised:
the indexed predicate \code{RBTree(T,c,h)} is registered semantically (its interpretation
\code{rbAt} pins the root colour \code{c} and black-height \code{h}, unfolding to the
leaf/black/red disjuncts with colour-correlated children), and the spec-level sorted
quantifier $\forall h{:}\mathit{Nat}$ realises the black-height parameter of that predicate:
instantiated at integer singletons, with recursive calls at $h{-}1$ (and, for a red parent,
at $h$). On this footing the recursive \code{RBTree\_member} search and the five-function
Okasaki \code{RBTree\_insert} both check, the latter against
$\forall h.\,(Int,\ \code{RBTree}(Int,c,h)) \Rightarrow
\code{RBTree}(Int,\code{Black},h) \vee \code{RBTree}(Int,\code{Black},h{+}1)$:
the postcondition states the result's exact colour and black-height: black, at height $h$
or $h{+}1$ (the latter when blackening a red root), exactly as red-black insertion behaves. As in
Sec.~\ref{sec:pred}'s definition, the predicate constrains colour and black-height only; key
ordering is not part of it, so the spec makes no BST claim. The mechanisation also instantiates
the decidability recipe of Sec.~\ref{sec:decidable}: \code{rbAt} is structurally recursive and
its index constraints ($h = h_1{+}1$, equal sibling heights) stay within linear integer
arithmetic. These are discharged here once, inside the kernel-checked unfold/fold lemmas, and are now
also implemented as a certified \emph{successor} sub-fragment inside the reflective checker itself
(App.~\ref{sec:leanchecker}), rather than by an SMT solver, consistent with the solver-free
positioning of Sec.~\ref{sec:intro}. The checker's \emph{general} affine-equality index engine,
which certifies the sized-list flagships of Table~\ref{tab:lean-examples} below, is a separate,
now-landed capability described under Table~\ref{tab:lean-checker}'s LIA rows in
App.~\ref{sec:leanchecker}. The typed \code{Queue} is now mechanised as well, closing the last
open row: every Sec.-evaluation benchmark and standard-library entry of
Table~\ref{tab:lean-coverage} is machine-checked. The mechanisation adds an inductive
heap-predicate layer over M2. An exact-footprint list-segment predicate realises the
paper's \code{LSeg}, and the paper's correspondence
$x{\mapsto}\code{List}(T) \Leftrightarrow x{\mapsto}\code{LSeg}(T,\code{Nil})$ holds;
in this encoding it is \emph{definitional} (the list predicate unfolds to the
nil-terminated segment, proved as \code{hlist\_iff\_lseg\_nil}). A heap-content strong-update
rule (\code{assignExact}) types a pointer write together with the structural content of the
cell it overwrites, carrying no content-type side condition: the separating-conjunction split
of the precondition does the work. The empty-queue case of \code{dequeue} realises the
paper's front/rear must-alias as \emph{ghost sharing} (both pointer cells are owned through
one shared ghost) and returns \Err{} structurally on the sentinel. What remains are honest
modeling \emph{boundaries}, not coverage gaps. First, the queue theorems are body-level
ghost-parameterized \code{SemValid} triples, since M2's spec discipline keeps function
\emph{specifications} heap-free (the same claim form as every existing M2 heap example). Second, on the
match-free fragment (straight-line code and calls with match-free bodies) the Frame rule now
also covers these heap-predicate atoms (\code{hpred}/\code{ptsEq}/\code{ptsCon}), with no
\code{FrameNoCon} side-condition there (\code{SemValid.frame\_matchfree}); the flagships are
\code{enqueue\_frames\_queue} (a disjoint packed queue framed across \code{enqueue}, NATIVE-FREE,
clean-3) and \code{ptsCon\_frame} (an exact \code{Cons} cell framed across a strong update). Across
a \code{match} the \code{FrameNoCon} side-condition remains, and is essential to the
operational-locality theorem (\code{frameNoCon\_essential}, a formal countermodel: a frame-owned
\code{con} can flip a \code{match}'s outcome from abort to normal termination); the analogous
question at the \code{SemValid} level is left open, since its abort-permissive arm makes that
flip invisible there. Third, the balancing predicate constrains
colour and black-height only, so no BST key-ordering claim is made (as above). The reflective
checker of App.~\ref{sec:leanchecker} now automates the inductive heap-predicate atoms of the
\code{Queue} layer too: all three queue theorems --- \code{enqueue\_semvalid\_reflective} (at the
\emph{fixed} ghost names \code{f}/\code{ga}/\code{gb}) and the two dequeue twins
\code{dequeue\_semvalid\_empty\_reflective}/\code{dequeue\_semvalid\_cons\_reflective} (each
certifying the exact SAME closed \code{SemValid} proposition as its hand twin) --- are discharged
by the M2 checker via \code{native\_decide}, folding/packing the exact heap atoms through the
\code{heapLawsEnv} registry. Both the red-black rows and the queue rows are thus now also checked
automatically by that reflective checker (the hand derivations here remain the derivation-effort
record; see App.~\ref{sec:leanchecker}). The checker's coverage goes further still: a
packed-precondition wrapper unfolds the opaque \code{hpred} queue invariant itself at the point a
triple is checked, so \code{enqueue\_semvalid\_packed} and \code{dequeue\_semvalid\_packed} certify
both operations directly from a single ghost-free \code{hpred} precondition (dequeue's post the
honest disjunction of its two arms), each by one \code{native\_decide} call
(App.~\ref{sec:leanchecker}). None of these is a soundness obstacle.

{
\small
\renewcommand{\arraystretch}{1.1}
\setlength{\tabcolsep}{5pt}
\begin{table}[h]
\centering
\begin{tabular}{@{}l p{0.62\linewidth}@{}}
\Xhline{2\arrayrulewidth}
Status & Benchmark / stdlib function \\
\Xhline{2\arrayrulewidth}
mechanised (M1) & \code{tail}, \code{length}, \code{length\_generic}, \code{singleton\_type}, \code{shape\_funs} \\
mechanised (M1 $+$ arith) & \code{inc}, \code{(+)}, \code{str\_of\_int} (i.e.\ \code{basic}), \code{List.nth} \\
mechanised (M1 $+$ balancing predicate) & \code{RBTree\_insert}, \code{RBTree\_member} \\
mechanised (M1 $+$ LIA) & \code{append} (sized), \code{length} (sized, exact), \code{size}/\code{mirror} (sized tree, checker-certified) \\
mechanised (M1 $+$ inequality LIA) & \code{SortedList.insert}, \code{SortedList.head}, \code{SortedList.member} \\
mechanised (M2) & \code{stateful\_funs}, \code{refs\_pure}, \code{refs\_sep}, \code{swap\_pure}, \code{swap\_sep} \\
mechanised (M2 $+$ heap predicates) & \code{Queue.enqueue}, \code{Queue.dequeue} \\
mechanised (M3 $+$ \code{List}) & \code{map}, \code{List.map}, \code{List.filter}, \code{List.fold\_left}, \code{List.append}, \code{List.rev}, \code{List.length} \\
mechanised (M3 $+$ \code{Option}) & \code{Option.map}, \code{Option.bind}, \code{Option.get} \\
\Xhline{2\arrayrulewidth}
\end{tabular}
\caption{\label{tab:lean-coverage} Status of a suite of standard benchmarks and
standard-library functions in the Lean mechanisation.}
\end{table}
}

\paragraph{The reflective type-checker: architecture.}
This section details the reflective, self-certifying type-checker of
App.~\ref{sec:leanchecker}. The checker (about $20{,}600$ lines beyond the derivations, atop a further ${\approx}14{,}700$
lines of certified index engines) is a
per-layer \emph{computable} forward checker. For each layer we implement the forward rules as
functions --- an algorithmic subtyping \code{subTyAlg}, a state-entailment test \code{entailsAlg},
and a recursive \code{checkExpr} over expressions in the pure layer M1 (\code{checkSep} in M2;
\code{subTyAlgHO}/\code{entailsAlgHO}/\code{checkHO} in M3) --- and prove \emph{once} that each is
sound with respect to the inductive derivation relation \code{Deriv} (theorems
\code{checkExpr\_sound}, \code{checkSep\_sound}, \code{checkHO\_sound}), each itself axiom-clean
($\{\textsf{propext},\textsf{Classical.choice},\textsf{Quot.sound}\}$).

\paragraph{Decision devices.}
Behind the two modes the checker settles the decidable side conditions of the rules with a small
library of provably sound decision procedures: list-/Bool-exhaustiveness laws
(\code{listCoversB}, \code{boolCoversB}); unsatisfiable-disjunct detection (\code{pureUnsat},
\code{findNegMem}); positive type-disjointness (\code{tyDisjB}, M3, which refutes the vacuous
\code{otherwise} and mismatched branches of the two-case \code{Option} specs); equality-aware
membership with binder projection (\code{sProject}); polymorphic-call instantiation inference
(\code{inferInst}, M1) and let-result-variable alignment (\code{letResVar}, M1); and a bounded
\emph{coverage decision procedure} (\code{coversB}/\code{coverUnsat} with \code{pureDisjoint},
\code{tautNegUnsat}, \code{varPost}, M1). This last procedure refutes the recursive-call coverage tail
$\mathit{ctx}\wedge\neg(g_1{\vee}\cdots{\vee}g_k)$ by case-splitting each list context into its
precise \code{Nil}${\vee}$\code{Cons} unfold and each guard singleton by excluded middle, until every
leaf is subsumed. In the symbolic mode each type parameter flows through only at opaque leaves,
guarded by decidable side conditions (\code{subTyAlg T T}, \code{tyEqB T T},
\code{findNegMem\,\ldots\,= none}) that are genuinely necessary: a closure-singleton type, for
instance, is not reflexively a subtype of itself. The algorithmic subtyping and entailment are
deliberately \emph{sound but incomplete} (function types are compared reflexively, heap entailment
by separating-conjunct matching), mirroring a practical analyser.

\paragraph{Automatic examples (Table~\ref{tab:lean-checker}).}
The pure and higher-order checkers cover all forward rules, so the paper's worked examples check
reflectively without hand-written proofs. In M1 this includes integer/string arithmetic
(\code{concat}/\code{div}/\code{mod}) so that \code{str\_of\_int} and monomorphic \code{nth} check,
and polymorphic-function calls (the instantiation \emph{inferred} by \code{inferInst}). In M3 it
includes closure application and \code{List(T)} field-type propagation, so that the \code{map},
\code{filter}, and \code{fold\_left} bodies check; the separation checker covers the
heap-manipulating rules including strong-update \code{assign}. For eight higher-order functions we
go further and discharge the \emph{whole-program} well-formedness obligation \code{WFEnv} by
\code{native\_decide}, then lift each through \code{HO.soundness} to a full semantic-validity
statement: the five recursive list functions \code{map}, \code{filter}, \code{fold\_left},
\code{append}, \code{rev} and the three \code{Option} functions \code{Option.get},
\code{Option.map}, \code{Option.bind}, certifying specifications such as
$\code{map}:(Int{\to}Bool,List(Int)){\Rightarrow}List(Bool)$ and
$\code{Option.get}:Option(Int){\Rightarrow}Int{\vee}\code{Err}$. All eight are \emph{additionally}
proved \emph{polymorphic} in their element/accumulator types
($\code{map}:\forall A\,B.\,(A{\to}B,List(A)){\Rightarrow}List(B)$, and so on) by the native-free
symbolic mode; and in the pure layer \code{id}, \code{length}, and \code{nth} likewise check for
\emph{all} element types $T$ ($\code{nth}:\forall T.\,(List(T),Int){\Rightarrow}T{\vee}\code{Err}$,
native-free), \code{nth}'s three-case path-sensitive specification additionally exercising the
coverage decision procedure above.

{
\small
\renewcommand{\arraystretch}{1.0}
\setlength{\tabcolsep}{5pt}
\begin{xltabular}{\linewidth}{@{}X l r@{}}
\caption{\label{tab:lean-checker} Reflective-checker results by layer. \emph{Mode}:
\code{native\_decide} (ground, $+$axiom) or \code{symbolic} (native-free); elaboration ms;
$^\dagger$composition-level.}\\
\Xhline{2\arrayrulewidth}
What is certified & Mode & Time (ms) \\
\Xhline{2\arrayrulewidth}
\endfirsthead
\Xhline{2\arrayrulewidth}
What is certified & Mode & Time (ms) \\
\Xhline{2\arrayrulewidth}
\endhead
\multicolumn{3}{l}{\emph{M1 --- pure fragment}} \\
ground reflective programs (16) \newline {\itshape (\code{inc}, \code{+}, \code{div}, \code{mod}, \code{concat}, \code{str\_of\_int}, mono \code{nth}, \code{id}/\code{length} poly-calls)} & \code{native\_decide} & $\lesssim$100 \\
$\code{id}:\forall T.\,T{\Rightarrow}T$ {\itshape (type-variable polymorphism)} & symbolic & 1346 \\
$\code{length}:\forall T.\,List(T){\Rightarrow}Int$ {\itshape (polymorphic recursion)} & symbolic & 1426 \\
$\code{nth}:\forall T.\,(List(T),Int){\Rightarrow}T{\vee}\code{Err}$ {\itshape (path-sensitive spec $+$ coverage procedure)} & symbolic & 2706 \\
\hline
\multicolumn{3}{l}{\emph{M1 --- red-black balancing (reflective)}} \\
\code{RBTree\_member} : $(Int,\code{RBTree}(Int,Bool,Int)){\Rightarrow}Bool$ {\itshape (recursive search)} & \code{native\_decide} & $\lesssim$100 \\
\code{RBTree\_insert} : $\forall h.\,(Int,\code{RBTree}(Int,c,h)){\Rightarrow}$ \newline $\code{RBTree}(Int,\code{Black},h){\vee}\code{RBTree}(Int,\code{Black},h{+}1)$ \ {\itshape (five-function)} & symbolic & ${\approx}15$~min \\
\hline
\multicolumn{3}{l}{\emph{M1 --- general affine-equality LIA (reflective)}} \\
\code{append} : $\forall n\,m.\,(List(T,n),List(T,m)){\Rightarrow}List(T,n{+}m)$ {\itshape (sized-list append, all integers --- $n{<}0$ vacuous pre, content at $n,m{\ge}0$; no per-size checker surgery)} & symbolic & 87$^\dagger$ \\
\code{length} : $\forall n.\,List(T,n){\Rightarrow}\{n\}$ {\itshape (exact-length post $=$ the size itself; same $n{<}0$ vacuity as \code{append})} & symbolic & 77$^\dagger$ \\
\code{size} : $\forall s.\,ST(T,s){\Rightarrow}\{s\}$ {\itshape (sized-tree destruction via a declared-post frame gate, double-rebound recursion through both children; result IS the size itself; vacuous pre at $s{<}0$)} & symbolic & 54$^\dagger$ \\
\code{mirror} : $\forall s.\,ST(T,s){\Rightarrow}ST(T,s)$ {\itshape (sized-tree destruction AND reconstruction; size $1{+}a{+}b$ preserved under the commuted reassembly $1{+}b{+}a$, reconciled by the semantic bridge \code{mirror\_consEntails} --- upgraded from a semantic-only certificate; vacuous pre at $s{<}0$, formalised as \code{mirrorPre\_vacuous}; no remaining documented index-language boundary)} & symbolic & 55$^\dagger$ \\
\hline
\multicolumn{3}{l}{\emph{M1 --- \code{SortedList} (inequality fragment, reflective)}} \\
\code{head} : $SortedList(Int,lo,hi)$, shared post $=$ range-validity $lo{\le}hi$ {\itshape (per-element bound certified in the $\forall$-element side-condition witness; partial correctness on \code{cons})} & symbolic & 10$^\dagger$ \\
\code{member} : $(Int,SortedList(Int,lo,hi)){\Rightarrow}Bool$ {\itshape (honest \code{Bool} post; rebound recursion, double-$<$ encoding)} & symbolic & 197$^\dagger$ \\
\code{insert} : $(Int,SortedList(Int,lo,hi)){\Rightarrow}SortedList(Int,lo,hi)$ \newline {\itshape (end-to-end sorted insertion; $\forall$-element \code{cons} premise via the \code{SoConsObligation} semantic side condition; vacuous outside $lo{\le}x{\le}hi$)} & symbolic & 62$^\dagger$ \\
\hline
\multicolumn{3}{l}{\emph{M2 --- separation / \code{Ref} heaps}} \\
separation programs (8) $+$ pipeline \newline {\itshape (\code{mkref}, \code{constr}, \code{deref}, \code{cast}, \code{let}, strong-update \code{assign}, \code{match}, \code{call})} & \code{native\_decide} & $\lesssim$100 \\
\code{Queue\_enqueue} : enqueue triple at fixed ghosts {\itshape (typed two-pointer queue)} & \code{native\_decide} & $\lesssim$100 \\
\code{Queue\_dequeue\_empty} : dequeue triple, \code{res : Err} {\itshape (empty must-alias case)} & \code{native\_decide} & $\lesssim$100 \\
\code{Queue\_dequeue\_cons} : dequeue triple, \code{res : Int} {\itshape (non-empty case)} & \code{native\_decide} & $\lesssim$100 \\
\code{Queue\_enqueue\_packed} : opaque \code{hpred} atom $+$ \code{v:int} \newline {\itshape (ghost-free packed precondition)} & \code{native\_decide} & $\lesssim$100 \\
\code{Queue\_dequeue\_packed} : bare \code{hpred} atom, disjunctive post \newline {\itshape (ghost-free packed precondition)} & \code{native\_decide} & $\lesssim$100 \\
\hline
\multicolumn{3}{l}{\emph{M3 --- higher-order}} \\
core HO $+$ list whole-programs (11) \newline {\itshape (apply/closures/compose; \code{map}, \code{filter}, \code{fold\_left}, \code{append}, \code{rev} \code{WFEnv})} & \code{native\_decide} & $\lesssim$200 \\
\code{Option} whole-programs (3) {\itshape (\code{Option.get}/\code{map}/\code{bind} \code{WFEnv}, two-case \code{tyDisjB})} & \code{native\_decide} & 537 \\
$\code{map}:\forall A\,B.\,(A{\to}B,List(A)){\Rightarrow}List(B)$ {\itshape (polymorphic whole-program)} & symbolic & 1966 \\
$\code{fold\_left}:\forall Acc\,A.\,((Acc,A){\to}Acc,Acc,List(A)){\Rightarrow}Acc$ \newline {\itshape (polymorphic whole-program)} & symbolic & 2308 \\
$\code{filter}:\forall A.\,(A{\to}Bool,List(A)){\Rightarrow}List(A)$ {\itshape (polymorphic whole-program)} & symbolic & 1214 \\
$\code{append}:\forall T.\,(List(T),List(T)){\Rightarrow}List(T)$ {\itshape (polymorphic whole-program)} & symbolic & 2854 \\
$\code{rev}:\forall T.\,(List(T),List(T)){\Rightarrow}List(T)$ {\itshape (polymorphic whole-program)} & symbolic & 1519 \\
$\code{Option.get}:\forall T.\,Option(T){\Rightarrow}T{\vee}\code{Err}$ {\itshape (polymorphic two-case spec)} & symbolic & 1232 \\
$\code{Option.map}:\forall A\,B.\,(A{\to}B,Option(A)){\Rightarrow}Option(B)$ \newline {\itshape (polymorphic whole-program)} & symbolic & 3165 \\
$\code{Option.bind}:\forall A\,B.\,(A{\to}Option(B),Option(A)){\Rightarrow}Option(B)$ \newline {\itshape (polymorphic whole-program)} & symbolic & 3392 \\
\Xhline{2\arrayrulewidth}
\end{xltabular}
}

\paragraph{Reflective red-black balancing (Table~\ref{tab:lean-checker}).}
The colour/black-height \code{RBTree}/\code{Node} development of Table~\ref{tab:lean-rbtree}
(once hand-derived only) now also checks \emph{reflectively}, through three additions to the
checker. (i)~A per-predicate \emph{law registry} \code{predLaws} (a record of
\code{subTy}/\code{fold}/\code{pkg}/\code{covers}/\code{disj} devices, dispatched from
\code{subTyAlg} with recursion kept inside \code{subTyAlg} by a \code{tySize}-guarded conjunction),
whose \code{predLawsEnv} registers the three indexed families: \code{RBTree(T,c,h)}, its
one-taller successor \code{RBTree}$^{+}$, and the almost-balanced \code{ART} used as \code{ins}'s
intermediate invariant --- each backed by a \emph{once}-proved soundness family
(\code{predLawsEnv\_subTy\_sound}/\code{\_covers\_sound}/\code{\_disj\_sound}/\code{\_pkg\_sound}).
(ii)~A core amendment realising the Hoare \emph{disjunction} rule ($\vee\Gamma$),
\code{Deriv.disj}/\code{Deriv.disjList}, letting \code{checkExpr} split a \code{match} branch's
multi-disjunct pre-state into singletons at branch entry and rejoin the posts, a general
completeness upgrade reaching beyond \code{RBTree}. (iii)~A certified \emph{successor}
sub-fragment of the linear-integer index recipe of Sec.~\ref{sec:decidable}: the index shifts
($h\mapsto h{+}1$, equal sibling heights) are now decided \emph{inside} the checker
(\code{hsSucc}/\code{inferNat}), ground index computation plus open-height $n$/$n{\pm}1$
reasoning, with \emph{no} SMT oracle (consistent with Sec.~\ref{sec:intro}); \code{RBTree\_insert}
itself needs nothing beyond this fragment. The checker's \emph{general} affine-equality
linear-arithmetic engine is a separate, now-landed capability (the LIA rows of
Table~\ref{tab:lean-checker}; Sec.~\ref{sec:decidable}), certifying the sized-list
\code{append}/\code{length} flagships and, via a declared-post frame gate crossing the
existential child-size decomposition, the sized-tree \code{size}/\code{mirror} flagships too,
all four checker-reflectively, no remaining documented index-language boundary. On this footing
Table~\ref{tab:lean-checker} gains two rows certifying the \emph{same} specifications as the
hand rows of Table~\ref{tab:lean-rbtree}: \code{RBTree\_member} (\code{rbmember\_semvalid}), the
wide $(Int,\code{RBTree}(Int,Bool,Int)){\Rightarrow}Bool$ search run at a ground type by
\code{native\_decide}; and the flagship \code{RBTree\_insert}
(\code{rbInsert\_semvalid\_reflective}), the \emph{whole} five-function Okasaki insert discharged
\emph{symbolically} for \emph{every} height ($\forall h$); its \code{WFEnv} obligation and all
five \code{*\_bodyChecks\_reflect} runs are native-free, so the end-to-end certificate is fully
axiom-clean ($\{\textsf{propext},\textsf{Classical.choice},\textsf{Quot.sound}\}$, no
\code{native\_decide}, no SMT). That last row is by far the most expensive symbolic example in
Table~\ref{tab:lean-checker}, the price of a fully symbolic $\forall h$ \code{WFEnv} over a
five-function program.

\paragraph{The hand breakdown (Table~\ref{tab:lean-rbtree}) and the completed boundary.}
Table~\ref{tab:lean-rbtree} retains the \emph{by-hand} red-black balancing development of
Table~\ref{tab:lean-examples} as a derivation-\emph{effort} record, broken down by component: the
shared predicate family and its lemma suite; the recursive \code{member} search; the two Okasaki
rebalancers \code{balanceL}/\code{balanceR}; the height-indexed \code{ins}; the \code{blacken}/%
\code{insert} wrappers with the \code{WFEnv} assembly; and the top-level \code{insert} certificate
against $\forall h.\,(Int,\code{RBTree}(Int,c,h)){\Rightarrow}\code{RBTree}(Int,\code{Black},h)\vee
\code{RBTree}(Int,\code{Black},h{+}1)$, whose postcondition states the result's exact colour and
black-height (black, at height $h$ or $h{+}1$). The predicate constrains colour and black-height
only, so the spec makes no key-ordering (BST) claim. The reflective checker now certifies these
\emph{same} two specifications automatically (the two \code{RBTree} rows of
Table~\ref{tab:lean-checker}), and (as described next) the M2 heap-atom layer behind
the typed \code{Queue} now checks reflectively too, including a further packed-precondition layer
that unfolds the opaque \code{hpred} queue atom itself at the point a triple is checked,
certifying both queue operations ghost-free directly from that single atom. The reflective
checker's \emph{inequality}-fragment engine (a certified unit-Farkas decision procedure) now
closes the paper's \code{SortedList(T,lo,hi)} boundary as well: a bespoke $\forall$-element
destruction rule universally quantifies each \code{cons} premise over the matched element's
value, and three checker-certified flagships --- \code{head}, \code{member}, and the campaign's
capstone sorted \code{insert} --- are added to Table~\ref{tab:lean-checker}'s inequality-fragment
group (Sec.~\ref{sec:decidable}, App.~\ref{sec:leanchecker}). The sized-tree \emph{destruction}
that a reflective \code{mirror} needs (once the one remaining boundary) is now also
checker-certified, via a declared-post frame gate and the \code{size}/\code{mirror} flagships of
Table~\ref{tab:lean-checker}'s LIA group: there is no remaining documented index-language
boundary (Sec.~\ref{sec:decidable}). What remains are limits on the completeness of the
certified decision procedures, not on coverage.

{
\small
\renewcommand{\arraystretch}{0.95}
\setlength{\tabcolsep}{5pt}
\begin{table}[h]
\centering
\begin{tabular}{l r r}
\Xhline{2\arrayrulewidth}
Component & Lines & Lean (ms) \\
\Xhline{2\arrayrulewidth}
\code{RBTree(Int,c,h)} predicate family $+$ unfold/fold lemma suite & 328 & 316 \\
\code{member} --- recursive search (program, derivation, witnesses) & 632 & 703 \\
\hline
\code{balanceL} --- body, $\forall h$ spec, per-height derivation ($+$ combinators) & 1606 & 5004 \\
\code{balanceR} --- mirror of \code{balanceL} & 1158 & 4636 \\
\code{ins} --- height-ledger \code{bodyChecks} (red-root \& black-root cases) & 1277 & 3579 \\
\code{ins}/\code{blacken}/\code{insert} bodies $+$ program; \code{blacken}/\code{insert} \code{bodyChecks} & 899 & 745 \\
\code{WFEnv} $+$ \code{insert} ($\forall h$ \code{rbInsert\_semvalid}) $+$ eval witnesses \& audits & 150 & 1070 \\
\Xhline{2\arrayrulewidth}
\textbf{Total} & \textbf{6050} & \textbf{16053} \\
\Xhline{2\arrayrulewidth}
\end{tabular}
\caption{\label{tab:lean-rbtree} The hand-mechanised red-black balancing development, broken
down by component --- the two \code{RBTree} rows of Table~\ref{tab:lean-examples}
(\code{RBTree\_member}~$960/1019$, \code{RBTree\_insert}~$5090/15034$). Lines are exact source
spans; times use the same profiler method as Table~\ref{tab:lean-examples}.}
\end{table}
}

\paragraph{Reflective heap-atom checking (Table~\ref{tab:lean-checker}).}
The exact-content heap-predicate layer behind the typed \code{Queue} (once hand-derived only)
now also checks \emph{reflectively}, through three additions to the M2 checker \code{SepCheck}.
(i)~An \emph{atom-cover entailment} \code{atomWeakenB} letting the exact atoms \code{ptsEq}/%
\code{ptsCon} weaken to the ordinary typed cells (\code{pts x (Ref t)}/\code{pts x (con c\,\ldots)})
they cover. (ii)~A \emph{fold registry} \code{heapLawsEnv}, dispatched by \code{foldRowB}, that
builds a registered heap-predicate atom by consuming pool atoms up to the pure eq-closure: the
list-segment \code{LSegR} (id~0) by peeling one link/node pair at a time, and the two-pointer
\code{queuePred} (id~1) by packing front/rear pointer cells, an owned sentinel, and an
\code{LSegR} segment --- \emph{folding} only, in the pack direction. (iii)~Binder-exit
field-slot inference: \code{matchE}'s \code{Cons} arm infers an all-\code{.eqv} \code{FieldSlot}
list from the scrutinee's owned \code{ptsCon} cell, naming each bound field by an existing
variable rather than by type alone. On this footing, Table~\ref{tab:lean-checker} gains three rows
certifying the typed two-pointer queue: \code{Queue\_enqueue} (\code{enqueue\_semvalid\_reflective}),
proved at the \emph{fixed} ghost names \code{f}/\code{ga}/\code{gb} --- unlike the $\forall h$
\code{RBTree\_insert} row, this is a ground result, not a $\forall$-ghost one --- and the two
dequeue arms \code{Queue\_dequeue\_empty}/\allowbreak\code{Queue\_dequeue\_cons}
(\code{dequeue\_semvalid\_empty\_reflective}/\allowbreak\code{dequeue\_semvalid\_cons\_reflective}), each
certifying the exact same closed \code{SemValid} proposition as its hand-derived twin. All three are \code{native\_decide}
rows run at ground ghost names; unlike the native-free symbolic \code{RBTree\_insert} row, each
carries its own per-declaration compiler axiom.

\paragraph{Packed-precondition unfolding (Table~\ref{tab:lean-checker}).}
Both queue operations admit a still-\emph{leaner} statement: instead of the fully unpacked
\code{enqueuePre}/\code{deqEmptyPre}/\code{deqConsPre} states above, the precondition is just the
single opaque \code{hpred 1 [t] [f,r]} atom --- plus, for \code{enqueue}, the pure \code{v:int}
conjunct --- GHOST-FREE, with no hand-written link or element names at all. A new registry
\code{unfoldRows} supplies, per registered predicate id, one-step unfolding \emph{templates}: for
the queue (id~1) a COARSE alternative (the pointer-cell shape \code{enqueue} needs) and a SPLIT
alternative (empty/\code{Cons}, the shape \code{dequeue} needs), each introducing a Boolean-checked%
-fresh set of ghost names (\code{ghostsOkB}). The wrapper \code{checkSepPacked} picks an
alternative and Boolean-checks the body against every one of its unfolded cases via
\code{checkSepConseq}. Its soundness lemma \code{checkSepPacked\_sound} --- itself NATIVE-FREE
(clean-3) --- reflects each case to a \code{Deriv}, lifts it to \code{SemValid}, combines the SPLIT
cases with \code{SemValid\_disj\_pre}, and transports the result back across the fresh-ghost
extension via a new store-extension lemma \code{SemValid\_unpack}. That lemma's adequacy over the
packed/unpacked boundary is discharged once per unfolding row (\code{queue\_unfold\_coarse}/%
\code{queue\_unfold\_split}), and it rests on a fully \emph{general} big-step
store-agreement theorem \code{BigStep\_agree}, which holds for every expression, calls included,
with no call-free restriction. On this footing Table~\ref{tab:lean-checker} gains two more rows:
\code{Queue\_enqueue\_packed} (\code{enqueue\_semvalid\_packed}), whose ghost-free packed
precondition needs only \code{v:int} beyond the opaque atom and whose post is the same
\code{enqueuePost} as \code{Queue\_enqueue}; and \code{Queue\_dequeue\_packed}
(\code{dequeue\_semvalid\_packed}), whose precondition is the bare atom (no pure conjunct at all)
and whose post is the honest two-arm disjunction \code{deqEmptyPost}${\vee}$\code{deqConsPost}.
Each theorem is exactly \emph{one} \code{native\_decide} call through the native-free
\code{checkSepPacked\_sound} wrapper. The unfold fires only at this packed-precondition entry
point; it is not a general heap-predicate rule usable anywhere a \code{hpred} atom occurs.

\paragraph{Reflective inequality-fragment checking: \code{SortedList} (Table~\ref{tab:lean-checker}).}
Beyond the general affine-\emph{equality} engine above, the checker carries a certified
\emph{inequality} decision procedure: a sound unit-Farkas worklist over the same reified index
form, deciding $<$/$\le$ constraints with no SMT call (App.~\ref{sec:leanchecker} details the
procedure and its documented incompleteness on coefficient scaling). On this footing the checker
registers a \code{SortedList(T,lo,hi)} predicate (a structural fold rebinding the tail's lower
bound at each matched element) together with a bespoke $\forall$-element destruction rule
(\code{Deriv.matchSortedCons}): unlike the ordinary \code{match} rule, its \code{cons} premise is
universally quantified over the matched element's \emph{value}, closing every emitted index atom
at that value rather than at a symbolic placeholder. Three flagships are checker-certified
end-to-end on this fragment, all native-free and axiom-clean, and gain three rows in
Table~\ref{tab:lean-checker}: \code{head} (the shared post collapses to a scalar range-validity
atom; the per-element bound is instead certified in the destruction rule's own side-condition
witness), \code{member} (an honest \code{Bool} post via rebound recursion and a double-$<$
comparison encoding), and, as the campaign's capstone, sorted \code{insert} (post
$\code{res}:SortedList(Int,lo,hi)$, a genuine end-to-end sorted insertion). In each case the
$\forall$-element \code{cons}-branch premise is discharged by a named semantic side condition
(\code{SoConsObligation}) rather than by the structural checker gate directly; the
declared-post frame gate used by construction flagships such as \code{insert} deliberately does
not re-run the \code{cons} branch at its reserved index variable, an honest, disclosed checker
incompleteness that costs nothing, since the real premise is supplied by the side condition
instead. As with the equality engine, no completeness claim is made for the inequality engine or
the checker.






\end{document}